\documentclass[aps,rmp,twocolumn,showpacs,amsfonts,amssymb,amsmath,twoside]{revtex4}

\usepackage{xspace}
\usepackage{amsxtra}
\usepackage{amsthm}
\usepackage{mathptmx}
\usepackage{pslatex}
\usepackage{url}
\usepackage{graphicx}

\usepackage[bookmarksnumbered,colorlinks,linkcolor=blue,citecolor=blue,urlcolor=blue]{hyperref}

\theoremstyle{plain}
\newtheorem{theorem}{Theorem}
\newtheorem{lemma}[theorem]{Lemma}
\newtheorem{algorithm}{Algorithm}
\newtheorem*{example}{Example}

\input{Qcircuit}

\xyoption{curve}
\xyoption{knot}
\xyoption{cmtip}

\newcommand{\<}{\langle}
\renewcommand{\>}{\rangle}
\newcommand{\ket}[1]{|#1\rangle}

\newcommand{\ketbra}[2]{|#1\rangle \langle #2|}
\newcommand{\braket}[2]{\langle #1|#2\rangle}

\newcommand{\be}{\begin{equation}}
\newcommand{\ee}{\end{equation}}
\def\ba#1\ea{\begin{align}#1\end{align}}
\newcommand{\nn}{\nonumber\\}

\newcommand{\A}{\mathbb{A}}
\newcommand{\C}{{\mathbb C}} 
\newcommand{\R}{{\mathbb R}} 
\newcommand{\F}{{\mathbb F}} 
\newcommand{\N}{{\mathbb N}}
\renewcommand{\P}{\mathbb{P}}
\newcommand{\Q}{{\mathbb Q}} 
\newcommand{\Z}{{\mathbb Z}} 
\newcommand{\Reg}{{\mathcal{R}}}

\newcommand{\Tr}{\operatorname{Tr}}
\newcommand{\I}{1}
\newcommand{\e}{{\mathrm e}}
\chardef\dotlessi=\i
\newcommand{\ii}{{\mathrm i}}

\newcommand{\semidirect}{\rtimes}

\newcommand{\ceil}[1]{\left\lceil #1 \right\rceil}
\newcommand{\floor}[1]{\left\lfloor #1 \right\rfloor}
\newcommand{\nint}[1]{\left\lfloor #1 \right\rceil}
\newcommand{\smfrac}[2]{\mbox{$\frac{#1}{#2}$}}

\newcommand{\lcm}{\operatorname{lcm}}
\newcommand{\ord}{\operatorname{ord}}

\newcommand{\poly}{\operatorname{poly}}
\newcommand{\rank}{\operatorname{rank}}
\newcommand{\stab}{\operatorname{stab}}

\newcommand{\Aut}{\operatorname{Aut}}
\newcommand{\Ind}{\operatorname{Ind}}
\newcommand{\Res}{\operatorname{Res}}
\newcommand{\normalin}{\unlhd}
\newcommand{\pnormalin}{\lhd}

\newcommand{\ZZ}[1]{\Z/\!#1\Z}
\newcommand{\ZZx}[1]{(\Z/\!#1\Z)^\times}
\newcommand{\ZN}{\ZZ{N}}
\newcommand{\ZNx}{\ZZx{N}}
\newcommand{\Zp}{\ZZ{p}}
\newcommand{\Zpx}{\ZZx{p}}
\newcommand{\Zpr}[1]{\ZZ{p_{#1}^{r_{#1}}}}
\newcommand{\Zprx}[1]{(\Z/\!p_{#1}^{r_{#1}}\Z)^\times}

\newcommand{\FF}[1]{\F\!_{#1}}
\newcommand{\FFsup}[2]{\F\!_{#1}^{\,\,#2}}
\newcommand{\FFx}[1]{\FFsup{#1}{\times}}
\newcommand{\Fp}{\FF{p}}
\newcommand{\Fq}{\FF{q}}
\newcommand{\Fqd}{\FFsup{q}{d}}
\newcommand{\Fpx}{\FFx{p}}
\newcommand{\Fqx}{\FFx{q}}
\newcommand{\Fpr}{\FF{p^r}}
\newcommand{\Fqr}{\FF{q^r}}
\newcommand{\GF}[1]{\FF{#1}}
\newcommand{\GL}{{\mathsf{GL}}}  

\newcommand{\rou}[1]{\omega_{#1}}
\newcommand{\rouN}{\rou{N}}
\newcommand{\Cl}{\operatorname{Cl}}
\newcommand{\pai}{\mathcal{O}}
\renewcommand{\div}{\operatorname{div}}
\newcommand{\Div}{\operatorname{Div}}

\renewcommand{\Re}{\operatorname{Re}}
\renewcommand{\Im}{\operatorname{Im}}

\newcommand{\eq}[1]{Eq.~(\ref{eq:#1})}
\newcommand{\eqs}[2]{Eqs.~(\ref{eq:#1}) and (\ref{eq:#2})}

\newcommand{\app}[1]{Appendix~\ref{app:#1}}
\renewcommand{\sec}[1]{Section~\ref{sec:#1}}
\newcommand{\thm}[1]{Theorem~\ref{thm:#1}} 
\newcommand{\lem}[1]{Lemma~\ref{lem:#1}} 
\newcommand{\alg}[1]{Algorithm~\ref{alg:#1}} 
\newcommand{\fig}[1]{Figure~\ref{fig:#1}} 
\newcommand{\tab}[1]{Table~\ref{tab:#1}}

\newcommand{\ctrld}[1]{\Lambda(#1)}
\newcommand{\cctrld}[1]{\Lambda^2(#1)}

\newcommand{\Four}[1]{F_{#1}}
\newcommand{\QFT}{QFT\xspace}
\newcommand{\QFTs}{{\QFT}s\xspace}
\newcommand{\HSP}{HSP\xspace}
\newcommand{\OCP}{OCP\xspace}

\newcommand{\cclass}[1]{\textsf{#1}\xspace}
\newcommand{\BPP}{\cclass{BPP}}
\newcommand{\BQP}{\cclass{BQP}}
\newcommand{\NP}{\cclass{NP}}

\newcommand{\PSPACE}{\cclass{PSPACE}}
\newcommand{\sharpP}{\cclass{{\#}P}}

\begin{document}
\title{Quantum algorithms for algebraic problems}

\author{Andrew M. Childs}
\email[]{amchilds@uwaterloo.ca}
\affiliation{\mbox{Department of Combinatorics \& Optimization
             and Institute for Quantum Computing} \\
             University of Waterloo,
             Waterloo, Ontario, Canada N2L 3G1}
\author{Wim van Dam}
\email[]{vandam@cs.ucsb.edu}
\affiliation{\mbox{Departments of Computer Science and Physics} \\ 
  University of California,
  Santa Barbara, California 93106,  USA} 

\begin{abstract}
Quantum computers can execute algorithms that dramatically outperform classical computation.  As the best-known example, Shor discovered an efficient quantum algorithm for factoring integers, whereas factoring appears to be difficult for classical computers.  Understanding what other computational problems can be solved significantly faster using quantum algorithms is one of the major challenges in the theory of quantum computation, and such algorithms motivate the formidable task of building a large-scale quantum computer.  This article reviews the current state of quantum algorithms, focusing on algorithms with superpolynomial speedup over classical computation, and in particular, on problems with an algebraic flavor. \hfill \\[-4ex]
\end{abstract}


\maketitle
\tableofcontents

\section{Introduction}
\label{sec:intro}

In the early 1980s, \textcite{Man80} and \textcite{Fey82} independently observed that computers built from quantum mechanical components would be ideally suited to simulating quantum mechanics.  Whereas brute-force classical simulation of a system of $n$ quantum particles (say, two-level atoms) requires storing $2^n$ complex amplitudes, and hence exponentially many bits of information, a quantum computer can naturally represent those amplitudes using only $n$ quantum bits.  Thus, it is natural to expect a quantum mechanical computer to outperform a classical one at quantum simulation.\footnote{In principle, quantum systems evolving according to simple interactions from a simple initial configuration can be described using fewer parameters, and classical simulations exploiting this idea have been developed (see for example \textcite{PVWC06}).  But while these ideas are extremely fruitful for simulating some quantum systems, we do not expect them to be efficient for \emph{any} physically reasonable system---in particular, not for systems capable of performing universal quantum computation.  However, we emphasize that there is no unconditional proof that classical simulation of quantum systems requires exponential overhead.}

The perspective of quantum systems as abstract information processing devices subsequently led to the identification of concrete tasks, apparently unrelated to quantum mechanics, for which quantum computers have a quantifiable advantage.  \textcite{Deu85} gave the first such example, a black-box problem that requires two queries to solve on a classical computer, but that can be solved with only one quantum query.  A series of related results \cite{DJ92,BV97} gave increasingly dramatic separations between classical and quantum query complexity, culminating in an example of \textcite{Sim97b} providing an exponential separation.  Building on this work, \textcite{Sho97} discovered in 1994 that a quantum computer could efficiently factor integers and calculate discrete logarithms.  Shor's result drew considerable attention to the concept of quantum information processing (see \textcite{EJ96} for an early review), and since then, the design and analysis of quantum algorithms has become a vibrant research area.

Quantum computers achieve speedup over classical computation by taking advantage of interference between quantum amplitudes.  Of course, interference occurs in classical wave mechanics as well, but quantum mechanics is distinguished by the ability to efficiently represent a large number of amplitudes with only a few quantum bits.\footnote{A similar situation occurs for the description of $n$ probabilistic bits by $2^n$ real-valued probabilities.  However, probabilities do not interfere; and contrary to the quantum case, randomized algorithms are not believed to be dramatically more powerful than deterministic ones (see for example \textcite{IW97}).}
In Shor's algorithm and its predecessors, the ``exponential interference'' leading to quantum speedup is orchestrated using a unitary operation called the \emph{quantum Fourier transform} (\QFT), an algebraic operation.  In this article, we review the state of the art in quantum algorithms for \emph{algebraic problems}, which can be viewed as continuations of the line of work leading from Deutsch to Shor.  Many, though not all, of these algorithms make use of the \QFT in some capacity.

Before beginning our exploration of quantum algorithms for algebraic problems, we briefly summarize the development of quantum algorithms more generally.
It has sometimes been said that there are really only two quantum algorithms: Shor's and Grover's.  We hope that this article will, in some small way, help to dispel this pernicious myth.  While it is difficult to compete with the impact of Shor's algorithm (a dramatic speedup for a problem profoundly relevant to modern electronic commerce) or the broad applicability of Grover's algorithm (a modest yet surprising speedup for the most basic of search problems), recent years have seen a steady stream of new quantum algorithms, both for artificial problems that shed light on the power of quantum computation, and for problems of genuine practical interest.

In 1996, \textcite{Gro97} gave an algorithm achieving quadratic speedup\footnote{Prior to Grover's result it was already shown by \textcite{BBBV97} that a quadratic speedup for the unstructured search problem is optimal.  More generally, for \emph{any} total Boolean function, there can be be at most a polynomial separation (in general, at most degree $6$) between classical and quantum query complexity \cite{BBCMW01}.} for the unstructured search problem, the problem of deciding whether a black-box Boolean function has any input that evaluates to $1$.  Grover's algorithm was subsequently generalized to the framework of amplitude amplification and to counting the number of solutions \cite{BHMT00}.  The unstructured search problem is extremely basic, and Grover's algorithm has found application to a wide variety of related problems (e.g., \textcite{BHT97b,DHHM04,AS06}).

The concept of quantum walk, developed by analogy to the classical notion of random walk, has proven to be another broadly useful tool for quantum algorithms.  Continuous-time quantum walk was introduced by \textcite{FG98}, and discrete-time quantum walk was introduced by \textcite{Wat01a}. 
The continuous-time formulation has been used to demonstrate exponential speedup of quantum over classical computation \cite{CCDFGS03,CSV07}, though it remains to be seen whether these ideas can be applied to a problem of practical interest.  However, both continuous- and discrete-time quantum walk have been applied to achieve polynomial speedup for a variety of search problems.  Following related work on spatial search \cite{AA05,SKW02,CG03,CG04,AKR04}, \textcite{Amb07} gave an optimal quantum algorithm for the element distinctness problem.  This approach was subsequently generalized \cite{Sze04c,MNRS06} and applied to other problems in query complexity, namely triangle finding \cite{MSS05}, checking matrix multiplication \cite{BS06}, and testing group commutativity \cite{MN07}.  Recently, quantum walk has also been applied to give optimal quantum algorithms for evaluating balanced binary game trees \cite{FGG07} and, more generally, Boolean formulas \cite{ACRSZ07,RS07}.

A related technique for quantum algorithms is the concept of adiabatic evolution.  The quantum adiabatic theorem guarantees that a quantum system in its ground state will remain close to its ground state as the Hamiltonian is changed, provided the change is sufficiently slow, depending on spectral properties of the Hamiltonian (see for example \textcite{BF28,JRS06}).  \textcite{FGGS00} proposed using adiabatic evolution as an approach to optimization problems.  Unfortunately, analyzing this approach is challenging.  While it is possible to construct specific cost functions for which specific formulations of adiabatic optimization fail \cite{DMV01,DV03,Fis92,Rei04}, the performance in general remains poorly understood.  Going beyond the setting of optimization problems, note that adiabatic evolution can simulate a universal quantum computer \cite{ADKLLR04}.

Finally, returning to the original motivation for quantum computation, Manin and Feynman's vision of quantum computers as quantum simulators has been considerably developed (e.g., \textcite{Llo96,Wie96,Zal98,ADLH05}).  However, it has proven difficult to identify a concrete computational task involving quantum simulation for which the speedup over classical computers can be understood precisely.  While it is widely expected that quantum simulation will be one of the major applications of quantum computers, much work remains to be done.

The main body of this article is organized as follows.  In \sec{complexity}, we give a brief introduction to the model of quantum computation and the complexity of quantum algorithms.  In \sec{abelQFT}, we introduce the Abelian quantum Fourier transform, and in \sec{abelHSP}, we show how this transform can be applied to solve the Abelian hidden subgroup problem, with various applications.  In \sec{numberfield}, we describe quantum algorithms for problems involving number fields, including the efficient quantum algorithm for solving Pell's equation.  In \sec{nonabelQFT}, we introduce the non-Abelian version of the quantum Fourier transform, and in \sec{nonabelHSP} we discuss the status of the non-Abelian version of the hidden subgroup problem.  In Sections~\ref{sec:hiddenshift} and \ref{sec:nonlinear}, we describe two approaches to going beyond the hidden subgroup framework, namely hidden shift problems and hidden nonlinear structure problems, respectively.  Finally, in \sec{jones}, we briefly discuss quantum algorithms for approximating the Jones polynomial and other \cclass{\#P}-complete problems.

\section{Complexity of Quantum Computation}
\label{sec:complexity}

In this section we give a brief introduction to quantum computers, with particular emphasis on characterizing computational efficiency.  For more detailed background, the reader is encouraged to consult \textcite{PreskillNotes,NC00,KSV02,KLM07}.

\subsection{Quantum data}

A quantum computer is a device for performing calculations using a quantum mechanical representation of information.  Data are stored using quantum bits, or \emph{qubits}, the states of which can be represented by $\ell_2$-normalized vectors in a complex vector space.  For example, we can write the state of $n$ qubits as
\be
  |\psi\> = \sum_{x \in \{0,1\}^n} a_x |x\>
\ee
where the $a_x \in \C$ satisfy $\sum_{x \in \{0,1\}^n} |a_x|^2 = 1$.  We refer to the basis of states $|x\>$ as the \emph{computational basis}.

Although we can always suppose that our data is represented using qubits, it is often useful to think of quantum states as storing data more abstractly.  For example, given a group $G$, we write $|g\>$ for a computational basis state corresponding to the group element $g \in G$, and
\be
  |\phi\> = \sum_{g \in G} b_g |g\>
\ee
(where $b_g \in \C$ with $\sum_{g \in G} |b_g|^2 = 1$) for an arbitrary superposition over the group.  We often implicitly assume that there is some canonical way of concisely representing group elements using bit strings; it is usually unnecessary to make this representation explicit.
We use the convention that for any finite set $S$, the state $|S\>$ denotes the normalized uniform superposition of its elements, i.e.,
\be
  |S\> := \frac{1}{\sqrt{|S|}} \sum_{s \in S} |s\>.
\label{eq:setket}
\ee

If a quantum computer stores the state $|\psi\>$ in one register and the state $|\phi\>$ in another, the overall state is given by the tensor product of those two states.  This may variously be denoted $|\psi\> \otimes |\phi\>$, $|\psi\> |\phi\>$, or $|\psi,\phi\>$.

It can be useful to consider statistical mixtures of pure quantum states, represented by \emph{density matrices}.  We refer the reader to the references above for further details.

\subsection{Quantum circuits}

The allowed operations on pure quantum states are those that map normalized states to normalized states, namely \emph{unitary operators} $U$, satisfying $U U^\dag = U^\dag U = 1$.  When viewed as an $N\times N$ matrix, the rows (and columns) of $U$ form an orthonormal basis of the space $\C^N$. 

To have a sensible notion of \emph{efficient} computation, we require that the unitary operators appearing in a quantum computation are realized by \emph{quantum circuits} \cite{Deu89,Yao93}.  We are given a set of gates, each of which acts on one or two qubits at a time, meaning that it is a tensor product of a nontrivial one- or two-qubit operator with the identity operator on the remaining qubits.  A quantum computation begins in the $|0 \ldots 0\>$ state, applies a sequence of one- and two-qubit gates chosen from the set of allowed gates, and finally reports an outcome obtained by measuring in the computational basis.  A circuit is called efficient if it contains a number of gates that is polynomial in the number of qubits the circuit acts on.

In principle, any unitary operator on $n$ qubits can be implemented using only $1$- and $2$-qubit gates \cite{Div94}.  Thus we say that the set of all $1$- and $2$-qubit gates is \emph{(exactly) universal}.  Of course, some unitary operators take many more $1$- and $2$-qubit gates to realize than others, and indeed, a simple counting argument shows that most unitary operators on $n$ qubits can only be realized using an exponentially large circuit \cite{Kni95}.

In general, we are content with circuits that give good approximations of our desired unitary transformations.  We say that a circuit with gates $U_1,U_2,\ldots,U_t$ approximates $U$ with precision $\epsilon$ if
$
  \|{U - U_t \cdots U_2 U_1}\| \le \epsilon
$,
where $\|\cdot\|$ denotes the operator norm, i.e., the largest singular value.  We call a set of elementary gates \emph{universal} if any unitary operator on a fixed number of qubits can be approximated to precision $\epsilon$ using $\poly(\log\frac{1}{\epsilon})$ elementary gates.  It turns out that there are finite sets of gates that are universal \cite{BMPRV99}: for example, the set $\{H,T,\ctrld{X}\}$ with 
\begin{gather}
  H := \frac{1}{\sqrt 2}\begin{pmatrix}1 & 1 \\ 1 & -1\end{pmatrix}
  \qquad
  T := \begin{pmatrix}e^{\ii \pi/8} & 0 \\ 0 & e^{-\ii  \pi/8}\end{pmatrix}
  \\
  \ctrld{X} := \begin{pmatrix}1 & 0 & 0 & 0 \\ 0 & 1 & 0 & 0 \\
                              0 & 0 & 0 & 1 \\ 0 & 0 & 1 & 0\end{pmatrix}.
\end{gather}

There are situations in which a set of gates is \emph{effectively} universal, even though it cannot actually approximate any unitary operator on $n$ qubits.  For example, the gate set $\{H,T^2,\ctrld{X},\cctrld{X}\}$, where
$\cctrld{X}$ denotes the Toffoli gate
($\cctrld{X}|xyz\>=|xyz\>$ for $xy \in \{00,01,10\}$, and $\cctrld{X}|11z\>=|11\bar z\>$)
is universal \cite{Kit97a}, but only if we allow the use of ancilla qubits (qubits that start and end in the $|0\>$ state).  Similarly, the gate set $\{H,\cctrld{X}\}$ is universal in the sense that, with ancillas, it can approximate any \emph{orthogonal} transformation \cite{Shi02b,Aha03}.  It clearly cannot approximate complex unitary matrices, since the entries of $H$ and $\cctrld{X}$ are real; but the effect of arbitrary unitary transformations can be simulated using orthogonal ones by simulating the real and imaginary parts separately \cite{BV93,RG02}.

One might wonder whether some universal gate sets are better than others.  It turns out that the answer is essentially no: a unitary operator that can be realized efficiently with one set of $1$- and $2$-qubit gates can also be realized efficiently with another such set.  This is a consequence of the Solovay-Kitaev theorem \cite{Kit97a,Sol00,HRC02}:

\begin{theorem}
Fix two gate sets that allow universal quantum computation and that are closed under taking inverses.  Then any $t$-gate circuit using the first gate set can be implemented with error at most $\epsilon$ using a circuit of $t \cdot \poly(\log(t/\epsilon))$ gates from the second gate set.  Furthermore, there is an efficient classical algorithm for finding this circuit.
\end{theorem}

In particular, this means we can view a simple finite gate set, such as $\{H,T,\ctrld{X}\}$, as equivalent to an infinite gate set, such as the set of all two-qubit gates.  A finite gate set is needed both for fault tolerance (\sec{faulttol}) and for the concept of uniformly generated circuits (Footnote \ref{foot:uniform}).

Note that to implement unitary operators \emph{exactly}, the notion of efficiency might depend on the allowed gates (see for example \textcite{MZ04}), so we usually restrict our attention to quantum computation with bounded error.

In principle, one can construct quantum circuits adaptively, basing the choices of gates on the outcomes of intermediate measurements.  We may also discard quantum data in the course of a circuit.  In general, the possible operations on mixed quantum states correspond to completely positive, trace preserving maps on density matrices.  Again, we refer the reader to the aforementioned references for more details. 

\subsection{Reversible computation}

Unitary matrices are invertible: in particular, $U^{-1} = U^\dag$.  Thus any unitary transformation is a reversible operation.  This may seem at odds with how we often define classical circuits, using irreversible gates such as $\textsc{and}$ and $\textsc{or}$.  But any classical computation can be made reversible by replacing each irreversible gate $x \mapsto g(x)$ by the reversible gate $(x,y) \mapsto (x,y \oplus g(x))$, where $\oplus$ denotes bitwise addition modulo $2$.  Applying this gate to the input $(x,0)$ produces $(x,g(x))$.  By storing all intermediate steps of the computation, we make it reversible \cite{Ben73}.

On a quantum computer, storing all intermediate computational steps could present a problem, since two identical results obtained via distinct computational histories would not be able to interfere.  However, there is an easy way to remove the accumulated information.  After performing the classical computation with reversible gates, we simply copy the answer into an ancilla register, and then perform the computation in reverse.  Thus we can implement the map $(x,y) \mapsto (x,y \oplus f(x))$ even when $f$ is a complicated circuit consisting of many gates.

Using this trick, any computation that can be performed efficiently on a classical computer can be performed efficiently on a quantum computer, even on a superposition of computational basis states.  In other words, if we can efficiently implement the map $x \mapsto f(x)$ on a classical computer, we can efficiently perform the transformation
\be
  \sum_{x} a_x |x,y\> \mapsto 
  \sum_{x} a_x |x,y \oplus f(x)\>
\label{eq:quantumreversible}
\ee
on a quantum computer.  Note that this does not necessarily mean we can efficiently perform the transformation
\be
  \sum_{x} a_x |x\> \mapsto 
  \sum_{x} a_x |f(x)\>,
\ee
even if the function $f$ is bijective.

\subsection{Quantum complexity theory}
\label{sec:qcomplexity}

We say that an algorithm for a problem is \emph{efficient} if the circuit describing it uses a number of gates that is polynomial in the input size, the number of bits needed to write down the input.\footnote{Strictly speaking, we would like the circuits for solving instances of a problem of different sizes to be related to one another in some simple way.  Given the ability to choose an arbitrary circuit for each input size, we could even have circuits computing uncomputable functions (i.e., functions that a Turing machine could not compute).  Thus we require our circuits to be \emph{uniformly generated}: say, that there exists a fixed (classical) Turing machine that, given a tape containing the symbol `$1$' $n$ times, outputs a description of the $n$th circuit in time $\poly(n)$. \label{foot:uniform}}
For example, if the input is an integer modulo $N$, the input size is $\ceil{\log_2 N}$.

With a quantum computer, as with a randomized (or noisy) classical computer, the final result of a computation may not be correct with certainty.  Instead, we are typically content with an algorithm that can produce the correct answer with high enough probability.  To solve a decision problem, it suffices to give an algorithm with success probability bounded above $1/2$ (say, at least $2/3$), since we can repeat the computation many times and take a majority vote to make the probability of outputting an incorrect answer arbitrarily small. Similarly, if we can check whether a given solution is correct, it suffices to output the correct answer with probability $\Omega(1)$.\footnote{In this article, we use standard big-$O$ notation, where $f=O(g)$ if there exist positive constants $c,y$ such that $|f(x)|\leq c|g(x)|$ for all $x\geq y$;   $f=\Omega(g)$ if $g=O(f)$; and $f=\Theta(g)$ if both $f=O(g)$ and $f=\Omega(g)$.  The expression $\Omega(1)$ thus represents a function lower bounded by an unspecified positive constant.  We write $f=o(g)$ to denote that $\lim_{x \to \infty} f(x)/g(x)=0$.  To convey that a function $f$ is bounded from above by a polynomial in the function $g$, we write $f=\poly(g)$, which could also be written as $f=g^{O(1)}$.}

It is common practice to characterize the difficulty of computational problems using \emph{complexity classes} (see for example \textcite{Pap94}).  Typically, these classes contain decision problems, problems with a `yes' or `no' answer.  (Such a problem is conventionally formulated as deciding whether a string over some finite alphabet is in a given \emph{language}; formally, a complexity class is a set of languages.)  For example, the problems that can be decided in polynomial time on a deterministic classical computer belong to the class \cclass{P}; on a probabilistic classical computer with error at most $1/3$, to the class \BPP; and on a quantum computer with error at most $1/3$, to the class \BQP.  Clearly, $\cclass{P} \subseteq \BPP \subseteq \BQP$.  The central problem of quantum algorithms can be viewed as trying to understand what problems are in \BQP, but not in \cclass{P} (or \BPP).

Whereas the classes \cclass{P}, \BPP, and \BQP all attempt to characterize modes of computation that could be carried out in practice, \emph{computational complexity theory} is also concerned with more abstract classes that characterize other aspects of computation.
For example, the class \NP corresponds to those decision problems for which a `yes' answer can be \emph{verified} in polynomial time on a classical computer, given a succinct proof.  It is widely believed that $\cclass{P} \ne \NP$, and indeed, that $\NP \not\subseteq \BQP$ (though it is also plausible that $\BQP \not\subseteq \NP$), but proving this appears to be an extremely challenging problem (see for example the excellent survey of quantum complexity by \textcite{Wat08}).  Indeed, it seems almost as difficult just to prove $\cclass{P} \ne \PSPACE$, where \PSPACE denotes the class of problems that can be decided by a deterministic classical computer running in polynomial \emph{space}.  Since $\BQP \subseteq \PSPACE$ \cite{BV97} (i.e., any computation that can be performed on a quantum computer in polynomial time can be performed on a classical computer with polynomial memory---indeed, even stronger such results are known \cite{BV97,ADH97,FR98}), we expect it will be hard to prove $\cclass{P} \ne \BQP$.  Instead, we try to find efficient quantum algorithms for problems that \emph{appear} to be hard for classical computers.

While most complexity classes contain decision problems, some classes describe the complexity of computing non-Boolean functions.  For example, the class \cclass{\#P} characterizes the complexity of counting the number of `yes' solutions to a problem in \NP.

Alternatively, instead of considering \emph{natural} computational problems (in which the input is a string), we sometimes work in the setting of \emph{query complexity}.  Here the input is a black-box transformation (or \emph{oracle})---which in the quantum setting is given as a unitary transformation as in \eq{quantumreversible}---and our goal is to discover some property of the transformation by querying it as few times as possible.  For example, in Simon's problem \cite{Sim97b}, we are given a black box for a transformation $f: \{0,1\}^n \to S$ satisfying $f(x)=f(y)$ iff $y\in\{x,x \oplus t\}$ for some unknown $t \in \{0,1\}^n$, and the goal is to learn $t$.

The query model facilitates proving lower bounds: it is often tractable to establish that many queries must be used to solve a given black-box problem, whereas it is generally hard to show that many gates are required to compute some explicit function.  Indeed, we will encounter numerous examples of black-box problems that can be solved in polynomial time on a quantum computer, but that \emph{provably} require exponentially many queries on a randomized classical computer.  Of course, if we find an efficient algorithm for a problem in query complexity, then if we are provided with an explicit, efficient circuit realizing the black-box transformation, we will have an efficient algorithm for a natural computational problem.  We stress, however, that lower bounds in the query model no longer apply when the black box is thus replaced by a transparent one.  For example, Shor's factoring algorithm (\sec{factor}) proceeds by solving a problem in query complexity which is provably hard for classical computers.  Nevertheless, it is an open question whether factoring is classically hard, since there might be a fast classical algorithm that does not work by solving the query problem.

\subsection{Fault tolerance}\label{sec:faulttol}

With any real computer, operations cannot be done perfectly.  Quantum gates and measurements may be performed imprecisely, and errors may happen even to stored data that is not being manipulated.  Fortunately, there are protocols for dealing with faults that occur during the execution of a quantum computation.  Specifically, the \emph{fault-tolerant threshold theorem} states that as long as the noise level is below some threshold (depending on the noise model and the architecture of the quantum computer, but typically in the range of $10^{-2}$ to $10^{-4}$), an arbitrarily long computation can be performed with arbitrarily small error \cite{Sho96,AB08,Kit97a,KLZ96,KLZ97,Pre97a}.  Throughout this article, we implicitly assume that fault-tolerant protocols have been applied, so that we effectively have a perfectly functioning quantum computer.

\section{Abelian Quantum Fourier Transform}
\label{sec:abelQFT}

\subsection{Fourier transforms over finite Abelian groups}
\label{sec:abelQFTdef}

For the group $\ZN$, the group of integers modulo $N$ under addition (see Appendix~\ref{app:nt}), the \emph{quantum Fourier transform} (\QFT) is a unitary operation $\Four{\ZN}$.
Its effect 
on a basis state $\ket{x}$ for any $x\in\ZN$ is 
\begin{equation}
\ket{x} \mapsto \frac{1}{\sqrt{N}}\sum_{y \in \ZN}\rouN^{xy}\ket{y},
\end{equation}
where $\rouN := \e^{2\pi\ii/N}$ denotes a primitive $N$th root of unity.

More generally, a finite Abelian group $G$ has $|G|$ distinct one-dimensional irreducible representations (or irreducible characters) $\psi\in\hat{G}$.  These are functions $\psi:G\rightarrow \C$ with $\psi(a+b)=\psi(a)\psi(b)$ for all $a,b\in G$, using additive notation for the group operation of $G$ (see Appendix~\ref{app:repr} for further details).  The quantum Fourier transform $\Four{G}$ over $G$ acts as
\begin{equation}
\ket{x} \mapsto \frac{1}{\sqrt{|G|}}\sum_{\psi\in\hat{G}}{\psi(x)\ket{\psi}}
\label{eq:abelQFT}
\end{equation}
for each $x \in G$.

For example, the group $(\ZN)\times(\ZN)$ has $N^2$ irreducible representations defined by $\psi_{y_1,y_2}:(x_1,x_2)\mapsto \rouN^{x_1y_1+x_2y_2}$ for all $y_1,y_2\in\ZN$; hence its quantum Fourier transform $\Four{(\ZN)\times(\ZN)}$ acts as
\begin{equation}
\ket{x_1,x_2} \mapsto \frac{1}{{N}}\sum_{y_1,y_2 \in \ZN} \rou{N}^{x_1y_1+x_2y_2}\ket{y_1,y_2}
\end{equation}
for all $x_1,x_2\in\ZN$.
In this example, $\Four{(\ZN)\times(\ZN)}$ can be written as the tensor product  $\Four{\ZN}\otimes\Four{\ZN}$. 
In general, according to the fundamental theorem of finite Abelian groups, any finite Abelian group $G$ can be expressed as a direct product of cyclic subgroups of prime power order, $G \cong (\ZZ{p_1^{r_1}}) \times \cdots \times (\ZZ{p_k^{r_k}})$, and the \QFT over $G$ can be written as the tensor product of \QFTs $F_{\ZZ{p_1^{r_1}}} \otimes \cdots \otimes F_{\ZZ{p_k^{r_k}}}$.

The Fourier transform $F_G$ is useful for exploiting symmetry with respect to $G$.  Consider the operator $P_s$ that adds $s \in G$, defined by $P_s|x\>=|x+s\>$ for any $x \in G$.  This operator is diagonal in the Fourier basis: we have
\begin{equation}
  \Four{G} P_s \Four{G}^\dag = \sum_{\psi \in \hat G} \psi(s) |\psi\>\<\psi|
.
\label{eq:fourier_shift}
\end{equation}
Thus, measurements in the Fourier basis produce the same statistics for a pure state $|\phi\>$ and its shift $P_s|\phi\>$.  Equivalently, a $G$-invariant mixed state is diagonalized by $F_G$.

\subsection{Efficient quantum circuit for the \QFT over \texorpdfstring{$\ZZ{2^n}$}{Z/(2\textasciicircum n)Z}}
\label{sec:abelQFTcircuits}

To use the Fourier transform over $G$ as part of an efficient quantum computation, we must implement it (approximately) by a quantum circuit of size $\poly(\log |G|)$.  This can indeed be done for any finite Abelian group \cite{Cop94,Cle94,Kit95,BEST96,Sho97,HH00}.
In this section we explain a construction for the case of the group $\ZZ{2^n}$, following the presentation of \textcite{CEMM98}.

Transforming from the basis of states $\{\ket{x}: x \in G\}$ to the basis  $\{\ket{\psi}: \psi \in \hat{G}\}$, the matrix representation of the Fourier transformation over $\ZN$ is
\begin{equation}
\Four{\ZN} = \frac{1}{\sqrt{N}}\begin{pmatrix}
1 & 1 & 1 & \cdots & 1 \\
1 & \rouN & \rouN^2 & \cdots & \rouN^{N-1} \\
1 & \rouN^2 & \rouN^4 & \cdots & \rouN^{2N-2} \\
\vdots & \vdots & \vdots &\ddots & \vdots \\
1 & \rouN^{N-1} & \rouN^{2N-2} & \cdots & \rouN^{(N-1)(N-1)}
\end{pmatrix}.
\end{equation}
More succinctly,
\begin{equation}
\Four{\ZN} = \frac{1}{\sqrt{N}}\sum_{x,y\in\ZN}{\rouN^{xy}\ketbra{y}{x}},
\label{eq:cyclicfourier}
\end{equation} 
where $\ket{y}$ represents the basis state corresponding to the character $\psi_y$ with $\psi_y(x) = \rouN^{xy}$.  It is straightforward to verify that $\Four{\ZN}$ is indeed a unitary transformation, i.e., that
$\Four{\ZN} \Four{\ZN}^\dagger=\Four{\ZN}^\dagger \Four{\ZN} = 1$. 

Assume now that $N=2^n$, and let us represent the integer $x\in\ZN$  by $n$ bits 
$x_0,x_1,\dots,x_{n-1}$ where $x=\sum_{j=0}^{n-1}2^j x_j$.
The Fourier transform of $|x\>$ can then be written as the tensor product 
of $n$ qubits, since
\begin{align}
\Four{\ZZ{2^n}}\ket{x} &=
\frac{1}{\sqrt{2^n}}\sum_{y\in\{0,1\}^n}{\rou{2^n}^{x(\sum_{j=0}^{n-1}{2^j y_j})}\ket{y_0,\dots,y_{n-1}}} \\
& = \frac{1}{\sqrt{2^n}}\bigotimes_{j=0}^{n-1}\sum_{y_j\in\{0,1\}}{\e^{2\pi\ii\,x y_j/2^{n-j}}\ket{y_j}} \\
& = \bigotimes_{j=0}^{n-1}{\frac{\ket{0}+\e^{2\pi\ii \sum_{k=0}^{n-1}{2^{j+k-n} x_k}}\ket{1}}{\sqrt{2}}}\\
& =: \bigotimes_{j=0}^{n-1} \ket{z_j}.
\end{align} 
Now, because $\exp(2\pi\ii \, 2^s x_k)=1$ for all integers $s\geq 0$, we see that the $j$th output qubit is 
\begin{align}
\ket{z_j} & = 
\frac{1}{\sqrt{2}}(\ket{0} + 
\e^{2\pi\ii (2^{j-n} x_0 + 2^{j+1-n} x_1 +\cdots+ 2^{-1} x_{n-1-j})}
\ket{1}),
\end{align} 
and hence only depends on the $n-j$ input bits $x_0,\dots,x_{n-1-j}$.

To describe a quantum circuit that implements the Fourier transform, we define the single-qubit phase rotation
\begin{equation}
R_r := 
\begin{pmatrix}
1 & 0 \\
0 & \e^{2\pi\ii/2^r}
\end{pmatrix}
\quad \simeq\quad 
\Qcircuit @C=1em @R=1em { &  \measure{R_r} & \qw} 
\end{equation} 
and the two-qubit controlled rotation
\begin{equation}
\ctrld{R_r} := 
\begin{pmatrix}
1 & 0 & 0 & 0 \\
0 & 1 & 0 & 0 \\
0 & 0 & 1 & 0 \\
0 & 0 & 0 & \e^{2\pi\ii/2^r}
\end{pmatrix}
\quad \simeq\quad 
\begin{array}{c}\Qcircuit @C=1em @R=1em { & \ctrl{1} & \qw \\
&  \measure{R_r} & \qw}
\end{array} 
\end{equation}
acting symmetrically on $a$ and $b \in \{0,1\}$ as $\ctrld{R_r} \ket{a,b} = \e^{2\pi\ii ab/2^r}\ket{a,b}$. The circuit shown in Figure~\ref{fig:QFT} uses $\binom{n}{2}$ of these gates together with $n$ Hadamard gates to exactly implement the quantum Fourier transform over $\ZZ{2^n}$.

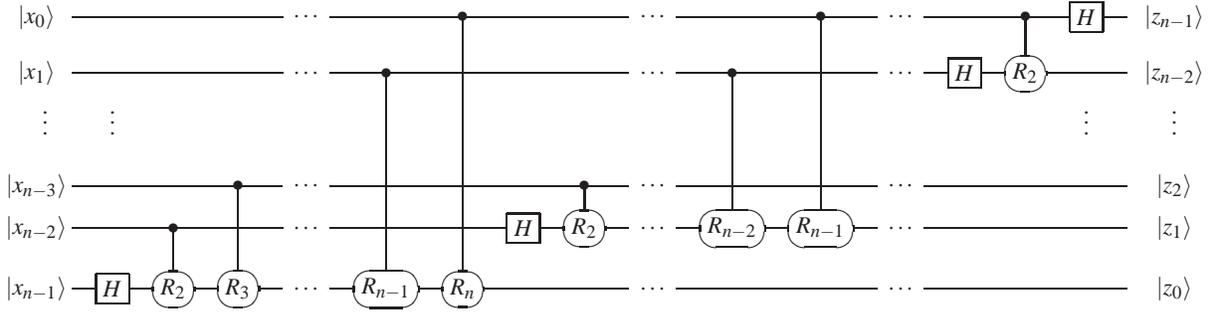
\begin{figure*}
\mbox{\Qcircuit @C=1em @R=1em { 
\ket{x_0}\quad & &  \qw & \qw & \qw &\qw& \cdots  & & \qw & \ctrl{7}& \qw &\qw&\qw& \cdots & & \qw & \ctrl{6} &\qw& \cdots & & \qw & \ctrl{1} & \gate{H}& \qw & & \ket{z_{n-1}}\\
\ket{x_1}\quad & &  \qw & \qw & \qw &\qw& \cdots & & \ctrl{6} & \qw & \qw &\qw&\qw& \cdots & & \ctrl{5} & \qw &\qw& \cdots & & \gate{H} & \measure{R_2} & \qw & \qw & & \ket{z_{n-2}} \\
\vdots~ ~ ~ & &  \vdots & & & & & & & & & & & & & & & & & & & & \vdots & & & \vdots\\\\
 & &   & \\
\ket{x_{n-3}}\quad & & \qw & \qw & \ctrl{2}& \qw & \cdots & & \qw & \qw & \qw & \ctrl{1}&\qw & \cdots & & \qw & \qw & \qw & \cdots &   & \qw & \qw & \qw & \qw & & \ket{z_2}\\
\ket{x_{n-2}}\quad & &  \qw &  \ctrl{1}&\qw&\qw& \cdots & & \qw & \qw & \gate{H} & \measure{R_2} & \qw& \cdots & & \measure{R_{n-2}} & \measure{R_{n-1}} & \qw &\cdots & & \qw & \qw & \qw & \qw & & \ket{z_1}\\
\ket{x_{n-1}}\quad & & \gate{H}& \measure{R_2} &\measure{R_3} & \qw & \cdots & &\measure{R_{n-1}} & \measure{R_{n}} & \qw &\qw&\qw& \cdots & & \qw & \qw & \qw &\cdots & & \qw & \qw & \qw & \qw  & & \ket{z_0}
}}
\caption{An efficient (size $O(n^2)$) quantum circuit for the quantum Fourier transform over $\ZZ{2^n}$.
Note that the order of the $n$ output bits $z_0,\dots,z_{n-1}$ is reversed, as compared with the order of the $n$ input bits $x_0,\dots,x_{n-1}$.}
\label{fig:QFT}
\end{figure*}

In this circuit, there are many rotations by small angles that do not significantly affect the final result.  By simply omitting the gates $\ctrld{R_r}$ with $r=\Omega(\log n)$, we obtain a circuit of size $O(n \log n)$ (instead of $O(n^2)$ for the original circuit) that implements the \QFT with precision $1/\poly(n)$ \cite{Cop94}.

\subsection{Phase estimation and the \QFT over any finite Abelian group}
\label{sec:phaseest}

Aside from being directly applicable to quantum algorithms, such as Shor's algorithm, the \QFT over $\ZZ{2^n}$ provides a useful quantum computing primitive called \emph{phase estimation} \cite{Kit95,CEMM98}.  In the phase estimation problem, we are given a unitary operator $U$ (either as an explicit circuit, or as a black box that lets us apply a controlled-$U^x$ operation for integer values of $x$).  We are also given a state $|\phi\>$ that is promised to be an eigenvector of $U$, namely $U |\phi\> = \e^{\ii\phi} |\phi\>$ for some $\phi \in \R$.  The goal is to output an estimate of $\phi$ to some desired precision.  (Of course, we can also apply the procedure to a general state $|\psi\>$; by linearity, we obtain each value $\phi$ with probability $|\<\phi|\psi\>|^2.$)

The procedure for phase estimation is straightforward:
\begin{algorithm}[Phase estimation] \label{alg:phaseest} \ \\
\emph{Input:} Eigenstate $|\phi\>$ (with eigenvalue $\e^{\ii\phi}$) of a given unitary operator $U$. \\
\emph{Problem:} Produce an $n$-bit estimate of $\phi$. \\
\begin{enumerate}
\item Prepare the quantum computer in the state
\be
  \frac{1}{\sqrt{2^n}}\sum_{x \in \ZZ{2^n}} |x\> \otimes |\phi\>.
\ee
\item
Apply the unitary operator
\be
  \sum_{x \in \ZZ{2^n}} |x\>\<x| \otimes U^x,
\label{eq:controlledux}
\ee
giving the state
\be
  \frac{1}{\sqrt{2^n}}\sum_{x \in \ZZ{2^n}} \e^{\ii\phi x} |x\> \otimes  |\phi\>.
\ee
\item
Apply an inverse Fourier transform on the first register, giving
\be
  \frac{1}{2^n}\sum_{x,y \in \ZZ{2^n}} \rou{2^n}^{x(\frac{2^n}{2\pi}\phi-y)} |y\> \otimes  |\phi\>.
\ee
\item
Measure the first register of the resulting state in the computational basis.
\end{enumerate}
\end{algorithm}
If the binary expansion of $\phi/2\pi$ terminates after at most $n$ bits, then the result is guaranteed to be the binary expansion of $\phi/2\pi$.  In general, we obtain a good approximation with high probability \cite{CEMM98}. (The relevant calculation appears in \sec{periodz} for the case where $\phi \in \Q$; that same calculation works for any $\phi \in \R$.) The optimal way of estimating the unknown phase is analyzed in \cite{DDEMM07}, but the above method is sufficient for our purposes.

The complexity of \alg{phaseest} can depend on the form of the unitary operator $U$.  If we are only given a black box for the controlled-$U$ gate, then there may be no better way to implement the controlled-$U^x$ operation than by performing a controlled-$U$ gate $x$ times, so that the running time is $\Theta(2^n)$ (i.e., approximately the inverse of the desired precision).  On the other hand, if it is possible to implement \eq{controlledux} in $\poly(n)$ time---say, using repeated squaring---then phase estimation can be performed in $\poly(n)$ time.

One useful application of phase estimation is to implement the \QFT \eq{cyclicfourier} over an arbitrary cyclic group $\ZN$ \cite{Kit95}.  The circuit presented in the previous section only works when $N$ is a power of two (or, with a slight generalization, a power of some other fixed integer).  But the following simple technique can be used to realize $\Four{\ZN}$ (approximately) using phase estimation.  (While this approach is conceptually simple, it is possible to implement the \QFT over a cyclic group more efficiently; see \textcite{HH00}.)

We would like to perform the transformation that maps $|x\> \mapsto |\hat x\>$, where $|\hat x\> := \Four{\ZN}|x\>$ denotes a Fourier basis state.  By linearity, if the transformation acts correctly on a basis, it acts correctly on all states.  It is straightforward to perform the transformation 
$|x,0\> \mapsto |x,\hat x\>$ (create a uniform superposition $\sum_{y \in \ZN} \ket{y}/\sqrt{N}$ in the second register and apply the controlled phase shift $\ket{x,y} \mapsto \rouN^{xy} \ket{x,y}$), but it remains to erase the first register.

Consider the unitary operator $P_1$ that adds $1$ modulo $N$, i.e., $P_1|x\>=|x+1\>$ for any $x \in \ZN$.  According to \eq{fourier_shift}, the eigenstates of this operator are precisely the Fourier basis states $|\hat x\>$, with eigenvalues $\rouN^x$.  Thus, using phase estimation on $P_1$ (with $n = O(\log N)$ bits of precision), we can approximate the transformation
$
  |\hat x,0\> \mapsto |\hat x, x\>
$.
Reversing this operation, we can erase $|x\>$, giving the desired \QFT.  Note that we can perform $P_1^x$ in $\poly(\log N)$ steps even when $x$ is exponentially large in $\log N$, so the resulting procedure is indeed efficient.

Given the Fourier transform over $\ZN$, it is straightforward to implement the \QFT over an arbitrary finite Abelian group using the decomposition of the group into cyclic factors, as discussed at the end of \sec{abelQFTdef}.

If gates can be performed in parallel, it is possible to perform the \QFT much more quickly, using only $O(\log\log N)$ time steps \cite{CW00,Hal02a}.

\subsection{The \QFT over a finite field}\label{sec:fieldqft}

The elements of the finite field $\Fq$, where $q=p^m$ is a power of a prime number $p$, form an Abelian group under addition (see Appendix~\ref{app:nt}), and the \QFT over this group has many applications.  If $q$ is prime, then $\Fq = \ZZ{q}$, so the \QFT over $\Fq$ is straightforward.  More generally, as an additive group, $\Fq \cong (\Zp)^m$, so in principle, the \QFT over $\Fq$ could be defined using an explicit isomorphism to $(\Zp)^m$.  However, it is often more convenient to define $F_{\Fq}$ in terms of the \emph{(absolute) trace}, the linear function $\Tr:\Fq \to \Fp$ defined by
\be
  \Tr(x) := x + x^p + x^{p^2} + \cdots + x^{p^{m-1}}.
\ee
One can show that the functions $\psi_y:\Fq \to \C$ defined by
\begin{equation}
 \psi_y(x) =   \rou{p}^{\Tr(xy)} 
\end{equation}
for each $y \in \Fq$
form a complete set of additive characters of $\Fq$.  Thus, the \QFT over $\Fq$ can be written 
\begin{equation} 
F_{\Fq} = \frac{1}{\sqrt q} \sum_{x,y \in \Fq} \rou{p}^{\Tr(xy)} |y\>\<x|.
\label{eq:fieldqft}
\end{equation}
This definition is preferred over other possible choices because it commutes with the permutation $|z\> \mapsto |z^p\>$ implementing the Frobenius automorphism, and hence respects the multiplicative structure of $\Fq$.

\section{Abelian Hidden Subgroup Problem}
\label{sec:abelHSP}

\subsection{Period finding over \texorpdfstring{$\ZN$}{Z/NZ}} \label{sec:period}

Suppose we are given a function over the integers $0,1,\ldots,N-1$ that is periodic with period $r$.  Further, suppose that this function never takes the same value twice within the fundamental period (i.e., it is \emph{injective} within each period).  In other words, the function $f:\ZN\rightarrow S$ satisfies
\begin{equation}\label{eq:rperiod}
  f(x)=f(y) \text{~if and only if~} \frac{x-y}{r} \in \Z
\end{equation}
for all $x,y\in\ZN$.  Notice that this can only be the case if $r$ divides $N$, so that $f$ can have exactly $N/r$ periods.

If we know $N$, then we can find the period $r$ efficiently using the quantum Fourier transform over the additive group $\ZN$.  We represent each element $x \in \ZN$ uniquely as an integer $x\in\{0,\dots,N-1\}$. Similarly, the irreducible representations $\psi:\ZN\rightarrow\C$ can be labeled by integers $y\in\{0,\dots,N-1\}$, namely with $\psi_y(x)=\e^{2\pi\ii xy/N}$.  The following algorithm solves the period finding problem.

\begin{algorithm}[Period finding over $\ZN$] \label{alg:period}\ \\
\emph{Input:} A black box $f: \ZN \to S$ satisfying \eq{rperiod} for some unknown $r \in \ZN$, where $r$ divides $N$. \\
\emph{Problem:} Determine $r$.

\begin{enumerate}
\item
Create the uniform superposition
\begin{equation}
  \ket{\ZN} = \frac{1}{\sqrt N} \sum_{x \in \ZN} |x\>
\end{equation}
of all elements of $\ZN$ (recall the notation \eq{setket}).  For example, this can be done by applying the Fourier transform over $\ZN$ to the state $\ket{0}$.
\item
Query the function $f$ in an ancilla register, giving
\begin{equation}
\frac{1}{\sqrt{N}}\sum_{x \in \ZN}\ket{x,f(x)}.
\label{eq:period_fvalue}
\end{equation}
\item
At this point, if we were to measure the ancilla register, the first register would be left in a superposition of those $x \in \ZN$ consistent with the observed function value. By the periodicity of $f$, this state would be of the form 
\begin{equation}
\sqrt{\frac{r}{N}}\sum_{j={0}}^{\frac{N}{r}-1}\ket{s+jr}
\label{eq:period_discard}
\end{equation}
for some unknown offset $s\in\{0,\dots,r-1\}$ occurring uniformly at random, corresponding to the uniformly random observed function value $f(s)$.   Since we will not use this function value, there is no need to explicitly measure the ancilla; ignoring the second register results in the same statistical description.  Thus, we may simply discard the ancilla, giving a mixed quantum state, or equivalently, a random pure state. 
\item
Apply the Fourier transform over $\ZN$, giving
\begin{equation}
\sqrt{\frac{r}{N}}\sum_{y\in\ZN}\sum_{j={0}}^{\frac{N}{r}-1}
\rouN^{(s+jr)y}\ket{y}.
\label{eq:period_afterf}
\end{equation}
By the identity
\begin{equation}
  \sum_{j=0}^{M-1} \rou{M}^{jy}=M\,\delta_{j,y\bmod{M}}
\label{eq:sumdelta}
\end{equation}
(applied with $M=N/r$, so $\rou{N}^{jry}=\rou{M}^{jy}$),
only the values $y\in \{0, N/r, 2N/r,\ldots,(r-1)N/r\}$ experience constructive interference, and \eq{period_afterf} equals
\begin{equation}
\frac{1}{\sqrt{r}}\sum_{k=0}^{r-1}\rou{r}^{sk}\ket{kN/r}.
\label{eq:periodfinalstate}
\end{equation}
\item
Measure this state in the computational basis, giving some integer multiple $kN/r$ of $N/r$.  
Dividing this integer by $N$ gives the fraction $k/r$, which, when reduced to lowest terms, has $r/\gcd(r,k)$ as its denominator.
\item Repeating the above gives a second denominator $r/\gcd(r,k')$.
If $k$ and $k'$ are relatively prime, the least common multiple of $r/\gcd(r,k)$ and $r/\gcd(r,k')$ is $r$. 
The probability of this happening is at least 
$\prod_{p \text{~prime}}(1-\frac{1}{p^2}) = 6/\pi^2 \approx 0.61$, 
so the algorithm succeeds with constant probability.  
\end{enumerate}
\end{algorithm}

\subsection{Computing discrete logarithms} \label{sec:dlog}

Let $C=\<g\>$ be a cyclic group generated by an element $g$, with the group operation written multiplicatively.    Given an element $x \in C$, the \emph{discrete logarithm of $x$ in $C$ with respect to $g$}, denoted $\log_g x$, is the smallest non-negative integer $\ell$ such that $g^\ell = x$.  The \emph{discrete logarithm problem} is the problem of calculating $\log_g x$ given $g$ and $x$.  (Notice that for \emph{additive} groups such as $G=\Zp$, the discrete log represents division: $\log_g x = x/g \bmod p$.)

\subsubsection{Discrete logarithms and cryptography}
\label{sec:diffiehellman}

Classically, the discrete logarithm seems like a good candidate for a one-way function. We can efficiently compute $g^\ell$, even if $\ell$ is exponentially large (in $\log|C|$), by repeated squaring.  But given $x$, it is not immediately clear how to compute $\log_g x$ without checking exponentially many possibilities.

The apparent hardness of the discrete logarithm problem is the basis of the \emph{Diffie-Hellman key exchange protocol} \cite{DH76}, the earliest published public-key cryptographic protocol.  The goal of key exchange is for two distant parties, Alice and Bob, to agree on a secret key using only an insecure public channel.  The Diffie-Hellman protocol works as follows:
\begin{enumerate}
\item[1.] Alice and Bob publicly agree on a large prime $p$ and an integer $g$ of high order.  For simplicity, suppose they choose a $g$ for which $\<g\>=\Zpx$ (i.e., a primitive root modulo $p$).  (In general, finding such a $g$ might be hard, but it can be done efficiently given certain restrictions on $p$.)
\item[2a.] Alice chooses some $a \in \ZZ{(p-1)}$ uniformly at random.  She computes $A:=g^a \bmod p$ and sends the result to Bob (keeping $a$ secret).
\item[2b.] Bob chooses some $b \in \ZZ{(p-1)}$ uniformly at random.  He computes $B:=g^b \bmod p$ and sends the result to Alice (keeping $b$ secret).
\item[3a.] Alice computes $K := B^a = g^{ab} \bmod p$.
\item[3b.] Bob computes $K = A^b = g^{ab} \bmod p$.
\end{enumerate}
At the end of the protocol, Alice and Bob share a key $K$, and an eavesdropper Eve has only seen $p$, $g$, $A$, and $B$.

The security of the Diffie-Hellman protocol relies on the assumption that discrete log is hard.  Clearly, if Eve can compute discrete logarithms, she can recover $a$ and $b$, and hence the key.  But it is widely believed that the discrete logarithm problem is difficult for classical computers.  The best known algorithms for general groups, such as Pollard's rho algorithm and the baby-step giant-step algorithm, run in time $O(\sqrt{|C|})$.  For particular groups, it may be possible to do better: for example, over $\Zpx$ with $p$ prime, the number field sieve is conjectured to compute discrete logarithms in time $2^{O((\log p)^{1/3} (\log\log p)^{2/3})}$ \cite{Gor93} (whereas the best known rigorously analyzed algorithms run in time $2^{O(\sqrt{\log p \log\log p})}$ \cite{Pom87}); but this is still superpolynomial in $\log p$.  It is suspected that breaking the Diffie-Hellman protocol is essentially as hard as computing the discrete logarithm.\footnote{It is nevertheless an open question whether, given the ability to break the protocol, Eve can calculate discrete logarithms.  Some partial results on this question are known \cite{Boe90,MW99}.}

This protocol by itself only provides a means of exchanging a secret key, not of sending private messages.  However, Alice and Bob can subsequently use their shared key in a symmetric encryption protocol to communicate securely.  The ideas behind the Diffie-Hellman protocol can also be used to directly create public-key cryptosystems (similar in spirit to the widely used RSA cryptosystem), such as the ElGamal protocol; see for example \cite{Buc04,MOV96}.

\subsubsection{Shor's algorithm for discrete log}

Although the problem appears to be difficult for classical computers, quantum computers can calculate discrete logarithms efficiently.  Recall that we are given some element $x$ of a cyclic group $C=\<g\>$ and we would like to calculate $\log_g x$, the smallest non-negative integer $\ell$ such that $g^\ell = x$.

For simplicity, assume that the order of the group, $N:=|C|$, is known.  For example, if $C=\Zpx$, then we know $N=p-1$.  If we do not know $N$, we can determine it efficiently using Shor's algorithm for period finding over $\Z$, which we discuss in \sec{periodz}.  We also assume that $x \ne g$ (i.e., $\log_g x \ne 1$), since it is easy to check this.

The algorithm of \textcite{Sho97} for computing discrete logarithms works as follows:
\begin{algorithm}[Discrete logarithm] \label{alg:dlog}\ \\
\emph{Input:} A cyclic group $C=\<g\>$ and an element $x\in C$.\\
\emph{Problem:} Calculate $\log_g x$.  
\begin{enumerate}
\item
If necessary, using the period finding algorithm of \sec{periodz}, determine the order $N=|C|$.
\item
Create the uniform superposition
\begin{equation}
  \ket{\ZN \times \ZN} = \frac{1}{N} \sum_{\alpha,\beta \in \ZN} |\alpha,\beta\>
\end{equation}
over all elements of the additive Abelian group $\ZN \times \ZN$.
\item
Define a function $f: \ZN \times \ZN \to C$ as follows:
\begin{equation}
  f(\alpha,\beta)=x^\alpha g^{\beta}.
\label{eq:dloghidingf}
\end{equation}
Compute this function in an ancilla register, giving
\begin{equation}
  \frac{1}{N} \sum_{\alpha,\beta \in \ZN} |\alpha,\beta,f(\alpha,\beta)\>.
\end{equation}
\item\label{item:dloglines}
Discard the ancilla register.\footnote{Note that if we were to measure the ancilla register instead of discarding it, the outcome would be unhelpful: each possible value $g^\gamma$ occurs with equal probability, and we cannot obtain $\gamma$ from $g^\gamma$ unless we know how to compute discrete logarithms.}
Since $f(\alpha,\beta)=g^{\alpha \log_g x + \beta}$, $f$ is constant on the lines
\begin{equation}
  L_\gamma :=
  \{(\alpha,\beta) \in (\ZN)^2: \alpha \log_g x + \beta = \gamma\},
\end{equation}
so the remaining state is a uniform superposition over group elements consistent with a uniformly random, unknown $\gamma \in \ZN$, namely
\ba
  |L_\gamma\> = \frac{1}{\sqrt N} \sum_{\alpha \in \ZN} |\alpha,\gamma-\alpha \log_g x \>.
\ea
\item
Now we can exploit the symmetry of the quantum state by performing a \QFT over $\ZN \times \ZN$, giving
\ba
  &\frac{1}{N^{3/2}} \sum_{\alpha,\mu,\nu \in \ZN}
  \rouN^{\mu\alpha+\nu (\gamma - \alpha \log_g x)} |\mu,\nu\> \\
  &\quad= \frac{1}{\sqrt N} \sum_{\nu \in \ZN} \rouN^{\nu\gamma}
  |\nu \log_g x,\nu\>
\ea
where we used the identity \eq{sumdelta}.
\item
Measure this state in the computational basis.  We obtain some pair $(\nu\log_g x,\nu)$ for a uniformly random $\nu \in \ZN$.  
\item Repeating the above gives a second pair $(\nu'\log_g x,\nu')$ with a uniformly random $\nu'\in\ZN$, independent of $\nu$. With constant probability (at least $6/\pi^2\approx 0.61$), $\nu$ and $\nu'$ are coprime, in which case we can find integers $\lambda$ and $\lambda'$ such that $\lambda\nu+\lambda'\nu'=1$.  Thus we can determine $\lambda\nu\log_g x + \lambda'\nu'\log_g x = \log_g x$.
\end{enumerate}
\end{algorithm}

This algorithm can be carried out for any cyclic group $C$, given a unique representation of its elements and the ability to efficiently compute products and inverses in $C$.  To efficiently compute $f(\alpha,\beta)$, we must  compute high powers of a group element, which can be done quickly by repeated squaring.

In particular, Shor's algorithm for discrete log breaks the Diffie-Hellman key exchange protocol described above, in which $C = \Zpx$.  In \sec{elliptic} we discuss further applications to cryptography, in which $C$ is the group corresponding to an \emph{elliptic curve}.

\subsection{Hidden subgroup problem for finite Abelian groups}
\label{sec:finiteabelHSP}

Algorithms \ref{alg:period} and \ref{alg:dlog} solve particular instances of a more general problem, the \emph{Abelian hidden subgroup problem} (or \emph{Abelian \HSP{}}).  We now describe this problem and show how it can be solved efficiently on a quantum computer.

Let $G$ be a finite Abelian group with group operations written additively, and consider a function $f:G\rightarrow S$, where $S$ is some finite set.  We say that $f$ \emph{hides} the subgroup $H \le G$ if 
\begin{equation}
f(x)=f(y) \text{~if and only if~} x-y \in H
\label{eq:abelhides}
\end{equation}
for all $x,y\in G$.  In the Abelian hidden subgroup problem, we are asked to find a generating set for $H$ given the ability to query the function $f$.

It is clear that $H$ can in principle be reconstructed from the entire truth table of $f$.  Notice in particular that $f(0)=f(x)$ if and only if $x \in H$: the hiding function is constant on the hidden subgroup, and does not take that value anywhere else.  Furthermore, fixing any $y \in G$, we see that $f(y)=f(x)$ if and only if $x \in y+H := \{y+h : h\in H\}$, a \emph{coset} of $H$ in $G$ with coset representative $y$.  So $f$ is constant on the cosets of $H$ in $G$, and distinct on different cosets.

The simplest example of the Abelian hidden subgroup problem is \emph{Simon's problem}, in which $G=(\ZZ{2})^n$ and $H=\{0,x\}$ for some unknown $x \in (\ZZ{2})^n$.  Simon's efficient quantum algorithm for this problem \cite{Sim97b} led the way to Shor's algorithms for other instances of the Abelian \HSP.

The period finding problem discussed in \sec{period} is the Abelian \HSP with $G=\ZN$.  The subgroups of $G$ are of the form $H=\{0,r,2r,\dots,N-r\}$ (of order $|H|=N/r$), where $r$ is a divisor of $N$.  Thus a function hides $H$ according to \eq{abelhides} precisely when it is $r$-periodic, as in \eq{rperiod}.  We have already seen that such a subgroup can be found efficiently.

The quantum algorithm for discrete log, as discussed in \sec{dlog}, solves an Abelian hidden subgroup problem in the group $\ZN \times \ZN$.  The function defined in \eq{dloghidingf} hides the subgroup
\begin{equation}
  H=\{(\alpha,\alpha\log_g x): \alpha \in \ZN\}.
\end{equation}
Shor's algorithm computes $\log_g x$ by finding this hidden subgroup.

More generally, there is an efficient quantum algorithm to identify any hidden subgroup $H \le G$ of a known finite Abelian group $G$.  (In \sec{wfsample} we relax the commutativity restriction to the requirement that $H$ is a normal subgroup of $G$, which is always the case if $G$ is Abelian.)  The algorithm for the general Abelian hidden subgroup problem is as follows:

\begin{algorithm}[Abelian hidden subgroup problem] \label{alg:abelhsp}\ \\
\emph{Input:} A black-box function $f:G\to S$ hiding some $H \le G$. \\
\emph{Problem:} Find a generating set for $H$.
\begin{enumerate}
\item
Create a uniform superposition $\ket{G}$ over the elements of the group.
\item
Query the function $f$ in an ancilla register, giving the state
\begin{equation}
\frac{1}{\sqrt{|G|}}\sum_{x\in G}\ket{x,f(x)}.
\end{equation}
\item
Discard the ancilla register, giving the \emph{coset state}
\begin{equation}
\ket{s+H} = \frac{1}{\sqrt{|H|}}\sum_{y\in H}{\ket{s+y}}
\end{equation} 
for some unknown, uniformly random $s \in G$.  Equivalently, the state can be described by the density matrix
\begin{align}
\rho_H &:= \frac{1}{|G|}\sum_{s\in G}\ketbra{s+H}{s+H}. 
\end{align}
\item
Apply the \QFT over $G$ to this state.  According to the definition of the \QFT in \eq{abelQFT}, the result is
\ba
  &\frac{1}{\sqrt{|H|\cdot|G|}} \sum_{\psi \in \hat G} \sum_{y \in H} \psi(s+y) |\psi\> \\
  &\quad= \sqrt{\frac{|H|}{|G|}} \sum_{\psi \in \hat G} \psi(s) \psi(H) |\psi\>
\ea
where
\begin{equation}
  \psi(H):=\frac{1}{|H|}\sum_{y\in H}\psi(y).
\end{equation}
If $\psi(y)=1$ for all $y \in H$, then clearly $\psi(H)=1$.  On the other hand, if there is any $y \in H$ with $\psi(y) \ne 1$ (i.e., if the restriction of $\psi$ to $H$ is not the trivial character of $H$), then by the orthogonality of distinct irreducible characters (\thm{orthochar} in \app{repr}), $\psi(H)=0$.  Thus we have the state
\begin{equation}
  |\widehat{s+H}\> := \sqrt{\frac{|H|}{|G|}} \sum_{\psi \in \hat G, \Res^G_H\psi = 1} \psi(s) |\psi\>
\end{equation}
or, equivalently, the mixed quantum state
\begin{equation}
\hat{\rho}_H := \frac{|H|}{|G|}\sum_{\psi \in \hat G, \Res^G_H\psi = 1}{\ketbra{\psi}{\psi}}, 
\end{equation}
where $\Res^G_H\psi=1$ means that $\psi(h)=1$ for all $h \in H$. 
\item \label{item:abelhspmeasure}
Measure in the computational basis.  Then we obtain one of the $|G|/|H|$ characters $\psi\in \hat{G}$ that is trivial on the hidden subgroup $H$, with every such character occurring with equal probability $|H|/|G|$.  Letting $\ker\psi := \{g \in G: \psi(g)=1\}$ denote the \emph{kernel} of the character $\psi$ (which is a subgroup of $G$), we learn that $H \le \ker\psi$.
\item
Repeat the entire process $T$ times, obtaining characters $\psi_1,\ldots,\psi_T$, and output a generating set for $K_T$, where $K_t := \bigcap_{j=1}^t \ker\psi_j$.  We are guaranteed that $H \le K_t$ for any $t$.  A simple calculation shows that if $K_t \ne H$, then $|K_{t+1}|/|K_t| \le 1/2$ with probability at least $1/2$.  Thus, we can choose $T=O(\log |G|)$ such that $K_T = H$ with high probability.
\label{item:abelhspintersect}
\end{enumerate}
\end{algorithm}

In summary, given a black-box function $f$ hiding a subgroup $H$ of a known finite Abelian group $G$, a quantum computer can determine $H$ in time $\poly(\log|G|)$, and in particular, using only $\poly(\log|G|)$ queries to the function $f$.  Of course, this assumes that we can efficiently implement group operations in $G$ using some unique representation of its elements.

In contrast, the Abelian hidden subgroup problem is typically hard for classical computers.  For example, an argument based on the birthday problem shows that even the simple case of Simon's problem (where $G=(\ZZ{2})^n$) has classical query complexity $\Omega(\sqrt{2^n})$ \cite{Sim97b}.  While certain special cases are easy---for example, since the only subgroups of $\Zp$ with $p$ prime are itself and the trivial subgroup, period finding over $\Zp$ is trivial---the classical query complexity of the Abelian \HSP is usually exponential.  In particular, one can show that if $G$ has a set of $N$ subgroups with trivial pairwise intersection, then the classical query complexity of the \HSP in $G$ is $\Omega(\sqrt N)$.  (For a proof in the case where $G=\Fq \times \Fq$, see \textcite{BCW02}.)

\subsection{Period finding over \texorpdfstring{$\Z$}{Z}} \label{sec:periodz}

In the previous section, we saw that the Abelian \HSP can be solved efficiently over any known finite Abelian group.  In this section we consider the \HSP over an infinite Abelian group, namely $\Z$ \cite{Sho97}.  Similar ideas can be used to solve the \HSP over any finitely generated Abelian group \cite{ME99}.  (For an Abelian group that is not finitely generated, new ideas are required, as we discuss in \sec{periodr}.)

The \HSP in $\Z$ is of interest when we are faced with a periodic function $f$ over an unknown domain.  For example, Shor's factoring algorithm (\sec{factor}) works by finding the period of a function defined over $\Z$.  Without knowing the factorization, it is unclear how to choose a finite domain whose size is a multiple of the unknown period, so we cannot immediately apply the period finding algorithm from \sec{period}.

Of course, we cannot represent arbitrary integers on a computer with finitely many bits.  Instead, we can restrict the function to the inputs $\{0,1,\ldots,N-1\}$ for some chosen $N$ and perform Fourier sampling over $\ZN$.  This can work even when the function is not precisely periodic over $\ZN$, provided $N$ is sufficiently large.  To simplify the implementation of the \QFT, we can choose $N=2^n$ to be a power of $2$.

This approach can only work if the period is sufficiently small, since otherwise we could miss the period entirely.  We will see how to choose $N$ if given an a priori upper bound on the period.  If we do not initially have such a bound, we can simply start with $N=2$ and repeatedly double $N$ until the period finding algorithm succeeds.  The overhead incurred by this procedure is only $\poly(\log r)$.

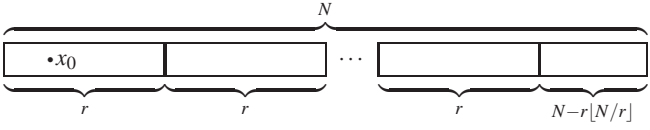
\begin{figure}
\setlength{\unitlength}{1ex}
\begin{picture}(60,10.5)
\put(0,7.5){\makebox{$\overbrace{\makebox[60ex]{}}^N$}}
\put(0,4){\framebox(15,3){}}
\put(4,5){\makebox{\raisebox{.2ex}{\tiny$\bullet$}$\,x_0\qquad$}}
\put(15,4){\framebox(15,3){}}
\put(30,4){\makebox(5,3){$\cdots$}}
\put(35,4){\framebox(15,3){}}
\put(50,4){\framebox(10,3){}}
\put(0,3.5){\makebox{$\underbrace{\makebox[15ex]{}}_r$}}
\put(15,3.5){\makebox{$\underbrace{\makebox[15ex]{}}_r$}}
\put(35,3.5){\makebox{$\underbrace{\makebox[15ex]{}}_r$}}
\put(50,3.5){\makebox{$\underbrace{\makebox[10ex]{}}_{N-r\floor{N/r}}$}}
\end{picture}
\caption{Sampling a $\Z$-periodic function over $\ZN$.}\label{fig:periodsample}\end{figure}

\begin{algorithm}[Period finding over $\Z$] \label{alg:periodz}\ \\
\emph{Input:} A black box $f: \ZN \to S$ satisfying \eq{rperiod} for some $r \in \Z$ with $r^2<N$, where $r$ does not necessarily divide $N$. \\
\emph{Problem:} Determine $r$.

\begin{enumerate}
\item
Prepare the uniform superposition $|\ZN\>$.
\item
Query the function in an ancilla register, giving
\ba
  \frac{1}{\sqrt N} \sum_{x \in \ZN} |x,f(x)\>.
\ea
\item \label{item:periodzsample}
Discard the ancilla register, leaving the first register in a uniform superposition over those $x \in \ZN$ consistent with some particular function value.  Since $f$ is periodic with minimum period $r$, we obtain a superposition over points separated by $r$.  The number of such points, $n$, depends on where the first point, $x_0 \in \{0,1,\ldots,r-1\}$, appears.  When restricted to $\ZN$, the function has $\floor{N/r}$ full periods and $N-r\floor{N/r}$ remaining points, as depicted in \fig{periodsample}.  Thus
\ba
  n=\begin{cases}
  \floor{N/r}+1 & x_0 < N - r\floor{N/r} \\
  \floor{N/r} & \text{otherwise}.
  \end{cases}
\ea
In other words, we are left with the quantum state
\ba
\frac{1}{\sqrt n} \sum_{j = 0}^{n-1} |x_0 + jr\>
\ea
where $x_0$ occurs nearly uniformly at random (specifically, it appears with probability $n/N$) and is unknown.
\item
Apply the Fourier transform over $\ZN$, giving
\ba
  \frac{1}{\sqrt{nN}} \sum_{k \in \ZN} \rouN^{k x_0} \sum_{j=0}^{n-1}  \rouN^{jkr} |k\> \label{eq:afterfourier}.
\ea
If we were lucky enough to choose a value of $N$ for which $r|N$, then $n=N/r$ regardless of the value of $x_0$, and the sum over $j$ gives $n\delta_{k \bmod n,0}$ by \eq{sumdelta}, so this state is identical to \eq{periodfinalstate}.  But more generally, the sum over $j$ in \eq{afterfourier} is the geometric series
\ba
  \sum_{j=0}^{n-1} \rouN^{jkr}
  = \frac{\rouN^{k r n}-1}{\rouN^{k r}-1} \label{eq:periodzsum}
  = \rouN^{(n-1)kr/2} \frac{\sin(\frac{\pi k r n}{N})}{\sin(\frac{\pi k r}{N})}.
\ea
\item \label{item:periodzmeasure}
Measure in the computational basis.
The probability of seeing a particular value $k$ is
\begin{equation}
  \Pr(k) = \frac{\sin^2(\frac{\pi k r n}{N})}{nN\sin^2(\frac{\pi k r}{N})}.
\end{equation}
From the case where $n=N/r$, we expect this distribution to be strongly peaked around values of $k$ that are close to integer multiples of $N/r$.  The probability of seeing $k = \nint{jN/r} = jN/r + \epsilon$ for some $j \in \Z$, where $\nint{x}$ denotes the nearest integer to $x$, is
\begin{align}
  \Pr(k=\nint{jN/r})
  &= \frac{\sin^2(\pi j n + \frac{\pi\epsilon r n}{N})}{nN\sin^2(\pi j + \frac{\pi \epsilon r}{N})} \\
  &=
  \frac{\sin^2 (\frac{\pi\epsilon r n}{N})}{nN \sin^2 (\frac{\pi \epsilon r}{N})}.
\end{align}
Using the inequalities $4x^2/\pi^2 \le \sin^2 x \le x^2$ (where the lower bound holds for $|x| \le \pi/2$, and can be applied since $|\epsilon| \le 1/2$), we find
\begin{align}
  \Pr(k=\nint{jN/r})
  &\ge  
  \frac{4}{\pi^2 r}.
\label{eq:periodzprobbound}
\end{align}
This bound shows that Fourier sampling produces a value of $k$ that is the closest integer to one of the $r$ integer multiples of $N/r$ with probability $\Omega(1)$.
\item \label{item:periodzcfe}
To discover $r$ given one of the values $\nint{jN/r}$, divide by $N$ to obtain a rational approximation to $j/r$ that deviates by at most $1/2N$, and compute the positive integers $a_i$ in the \emph{continued fraction expansion} (CFE)
\begin{equation}
  \frac{\nint{jN/r}}{N} = \frac{1}{a_1 + \displaystyle\frac{1}{a_2 +       \displaystyle\frac{1}{a_3 + \cdots}}}.
\end{equation}
This expansion gives a sequence of successively better approximations to $\nint{jN/r}/N$ by fractions, called the \emph{convergents} of the CFE. By \cite[Theorem 184]{HW79}, any fraction $p/q$ with $|p/q - \nint{jN/r}/N| < 1/2q^2$ will appear as one of the convergents.  Since $j/r$ differs by at most $1/2N$ from $\nint{jN/r}/N$, the fraction $j/r$ will appear as a convergent provided $r^2 < N$.  Thus, we carry out the CFE until we obtain the closest convergent to $\nint{jN/r}/N$ whose denominator is smaller than our a priori upper bound on the period; this denominator must equal $r$.  These calculations can be done in polynomial time using standard techniques; see for example \cite[Chapter X]{HW79}.  
\end{enumerate}
\end{algorithm}

Notice that period finding can efficiently determine the \emph{order} of a given group element $g \in G$, the smallest $r \in \{1,2,\ldots\}$ such that $g^r=1$.  This follows because the function $f:\Z \to G$ defined by $f(j)=g^j$ is periodic, with period equal to the order of $g$ in $G$.  In particular, this allows us to find the order of a cyclic group $C=\<g\>$, as needed in \alg{dlog}.  In contrast, the classical query complexity of computing the order of a permutation of $2^n$ elements is $\Omega(2^{n/3}/\sqrt n)$ \cite{Cle04}.

\subsection{Factoring integers} \label{sec:factor}

Perhaps the best-known application of quantum computers is to the problem of factoring integers \cite{Sho97}.  At present, the mostly widely used public-key cryptosystem, RSA \cite{RSA78}, is based on the presumed difficulty of this problem.\footnote{The RSA protocol uses similar ideas to the Diffie-Hellman protocol (\sec{dlog}), but relies on a different assumption and achieves secure communication instead of key exchange.  Note that breaking RSA might be easier than factoring. For elementary discussions of the details of RSA and related protocols, see \cite{Buc04,MOV96}.}
The fastest rigorously analyzed classical algorithm for factoring an integer $N$ has running time $2^{O(\sqrt{\log N \log\log N})}$ (see for example \textcite{Pom87}), and the best known classical algorithm is believed to be the number field sieve \cite{BLP93}, which is conjectured to run in time $2^{O((\log N)^{1/3} (\log\log N)^{2/3})}$.  Both of these running times are superpolynomial in $\log N$.
In contrast, a quantum computer can factor $N$ in time $O(\log^3 N)$.  Thus, the development of a large-scale quantum computer could have dramatic implications for the practice of cryptography.

We have already discussed the core of Shor's quantum factoring algorithm, the ability to perform period finding over the integers.  It remains to see how factoring can be reduced to a particular instance of period finding.

To efficiently factor a given integer $N$, it suffices to efficiently produce some nontrivial factor of $N$ (i.e., a factor other than $1$ or $N$) with constant probability.  The repeated use of such a subroutine, combined with an efficient primality testing algorithm \cite{Mil76,Rab80,AKS04}, can be used to find all the prime factors of $N$.  It is easy to check whether $2$ divides $N$, so we can focus on the case of $N$ odd without loss of generality.  Furthermore, it is straightforward to check whether $N$ is a prime power, or indeed whether it is the $k$th power of any integer, simply by computing $\sqrt[k]{N}$ for $k=2,3,\dots,\log_2 N$, so we can assume that $N$ has at least two distinct prime factors.

The reduction from finding some nontrivial factor of an odd $N$ to order finding in the multiplicative group $\ZNx$ is due to \textcite{Mil76}.  Suppose we choose $a \in \{2,3,\dots,N-1\}$ uniformly at random from those values that are coprime to $N$.  Furthermore, assume for now that the order $r$ of $a$ is even.  Then since $a^r = 1 \bmod N$, we have $(a^{r/2})^2 - 1 = 0 \bmod N$, or equivalently,
\begin{equation}
  (a^{r/2}-1)(a^{r/2}+1) = 0 \bmod N.
\end{equation}
Since $N$ divides the product $(a^{r/2}-1)(a^{r/2}+1)$, we might hope for $\gcd(a^{r/2}-1,N)$ to be a nontrivial factor of $N$.  Notice that $\gcd(a^{r/2}-1,N) \ne N$, since if it were, the order of $a$ would be at most $r/2$.  Thus it suffices to ensure that $\gcd(a^{r/2}-1,N) \ne 1$, which holds if $a^{r/2} \ne -1 \bmod{N}$.  In \lem{factorsuccessprob} below, we show that a random value of $a$ satisfies these properties with probability at least $1/2$, provided $N$ has at least two distinct prime factors.  Thus the following quantum algorithm can be used to factor $N$:

\begin{algorithm}[Integer factorization] \label{alg:factor}\ \\
\emph{Input:} An odd integer $N$ with at least two distinct prime factors. \\
\emph{Problem:} Determine some nontrivial factor of $N$.
\begin{enumerate}
\item Choose a random $a\in\{2,3,\dots,N-1\}$. \label{item:factorstart}
\item Compute $\gcd(a,N)$ using the Euclidean algorithm.  If the result is different from $1$, then it is a nontrivial factor of $N$, and we are done.  More likely, $\gcd(a,N)=1$, and we continue.
\item Using \alg{periodz}, determine the order of $a$ modulo $N$.  If $r$ is odd, the algorithm has failed, and we return to step \ref{item:factorstart}.  If $r$ is even, we continue.
\item Compute $\gcd(a^{r/2}-1,N)$.  If the result is different from $1$, then it is a nontrivial factor of $N$.  Otherwise, return to step \ref{item:factorstart}.
\end{enumerate}
\end{algorithm}

\begin{lemma}\label{lem:factorsuccessprob}
Suppose $a$ is chosen uniformly at random from $\ZNx$, where $N$ is an odd integer with at least two distinct prime factors.  Then with probability at least $1/2$, the multiplicative order $r$ of $a$ modulo $N$ is even, and $a^{r/2} \ne -1 \bmod{N}$.
\end{lemma}

\begin{proof}
  Suppose $N=p_1^{m_1} \cdots p_k^{m_k}$ is the factorization of $N$
  into powers of $k \ge 2$ distinct odd primes.  By the Chinese
  remainder theorem, there are unique values $a_i \in \ZZ{p_i^{m_i}}$
  such that $a = a_i \bmod{p_i^{m_i}}$.  Let $r_i$ be the
  multiplicative order of $a_i$ modulo $p_i^{m_i}$, and let $2^{c_i}$
  be the largest power of $2$ that divides $r_i$.  We claim that if
  $r$ is odd or if $a^{r/2} = -1 \bmod{N}$, then
  $c_1=\cdots=c_k$.  Since $r=\lcm(r_1,\ldots,r_k)$, we have
  $c_1=\cdots=c_k=0$ when $r$ is odd.  On the other hand, if $r$ is
  even and $a^{r/2} = -1 \bmod{N}$, then for each $i$ we have
  $a^{r/2} = -1 \bmod{p_i^{m_i}}$, so $r_i$ does not divide
  $r/2$; but we know that $r/r_i$ is an integer, so it must be odd,
  which implies that each $r_i$ has the same number of powers of $2$
  in its prime factorization.

Now we claim that the probability of any given $c_i$ taking on any
particular value is at most $1/2$, which implies that $\Pr(c_1=c_2)
\le 1/2$, and the desired conclusion follows.  To see this, consider
$a$ chosen uniformly at random from $\ZNx$---or equivalently, each
$a_i$ chosen uniformly at random from $\ZZx{p_i^{m_i}}$.  The order of
the latter group is $\varphi(p_i^{m_i})=(p_i-1)p_i^{m_i}=2^{d_i} q_i$
for some positive integer $d_i$ and some odd integer $q_i$.  The
number of $a_i \in \ZZx{p_i^{m_i}}$ of odd order is $q_i$, and the
number of $a_i$'s with any particular $c_i \in \{1,\ldots,d_i\}$ is
$2^{c_i-1} q_i$.  In particular, the highest-probability event is
$c_i=d_i$, which happens with probability only $1/2$.
\end{proof}

\subsection{Breaking elliptic curve cryptography}\label{sec:elliptic}

As discussed in \sec{dlog}, Shor's algorithm allows quantum computers
to break cryptographic protocols based on the presumed hardness of the
discrete logarithm problem in $\ZNx$, such as the Diffie-Hellman key
exchange protocol.  However, Shor's algorithm works equally well for
calculating discrete logarithms in any finite group, provided only
that group elements can be represented uniquely and operated on
efficiently.  In particular, quantum computers can also efficiently
calculate discrete logarithms over the group corresponding to an
\emph{elliptic curve}, thereby breaking elliptic curve cryptography.

An elliptic curve is a cubic, nonsingular, planar curve 
over some field.
(The terminology ``\emph{elliptic} curve'' has to do with a connection to elliptic functions.)
For simplicity, suppose we choose
a field with characteristic not equal to $2$ or $3$.  (Cryptographic
applications often use the field $\F_{2^n}$ of characteristic $2$, but
the definition of an elliptic curve is slightly more complicated in
this case.)  Then, by suitable linear transformations, any elliptic curve can be rewritten in the form of
the \emph{Weierstra{\ss} equation},
\begin{equation} y^2 = x^3 + ax + b,
\label{eq:ec}
\end{equation}
where $a,b$ are parameters.  The set of points $(x,y)$ satisfying this equation
form an elliptic curve.  To be nonsingular, the \emph{discriminant} $\Delta := -16(4a^3 + 27b^2)$ must be nonzero.  
Typically, one considers elliptic curves in the projective plane $\P^2$ rather than the affine plane, which means that one point at infinity must be included in the set of solutions.
(For further details on the concepts of projective curves, points at infinity, and nonsingularity, see \app{curves}.)

An example of an elliptic curve over the field $\R$ (namely, the curve $y^2=x^3-x+1$) is shown in \fig{elliptic}.  Although such pictures are helpful for developing intuition about elliptic curves, it is useful in cryptographic applications to have a curve whose points can be represented exactly with a finite number of bits, so we use curves over finite fields.  For simplicity, we will only consider the field $\F_p$ where $p$ is a prime larger than $3$.

\begin{figure}
  \includegraphics[width=.7\columnwidth]{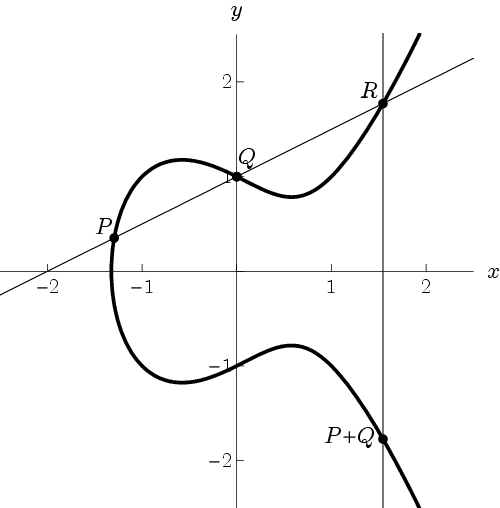}
  \caption{The group law for an elliptic curve: $P+Q=-R$.  The points $P$ and $Q$ sum to the point $-R$, where $R$ is the intersection between the elliptic curve and the line through $P$ and $Q$, and $-R$ is obtained by the reflection of $R$ about the $x$ axis.\label{fig:elliptic}}
\end{figure}
 
\begin{example}
Consider the curve
\begin{equation}
  E = \{(x,y) \in \F_7^2: y^2 = x^3 - x + 1 \}
\end{equation}
over $\F_7$.  It has $4 a^3 + 27 b^2 = 2 \bmod 7$, so it is nonsingular.  It is straightforward to check that the points on this curve are
\begin{equation}
  \begin{aligned}
  E =
  \{&\pai, (0,1), (0,6), (1,1), (1,6), (2,0), \\
&(3,2), (3,5), (5,3), (5, 4), (6, 1), (6,6)\},
  \end{aligned}
\end{equation}
where $\pai$ denotes the point at infinity.
\end{example}
In general, the number of points on an elliptic curve depends on the parameters $a$ and $b$.  However, a theorem of Hasse says that 
$\big| |E|-(p+1) \big| \le 2 \sqrt{p}$, so for large $p$ 
the number of points is close to $p$. 

An elliptic curve can be used to define an Abelian group by designating
one point of the curve as the additive identity.  Here, we use the common convention that $\pai$, the point at infinity, is this special element (although in principle, it is possible to let any point play this role).
It remains to define a binary operation `$+$' that maps a pair of points on the curve to a new point on the curve in a way that satisfies the group axioms.  To motivate the definition, consider the case of the field $\R$.  Given two points $P,Q \in E$, their sum $P+Q$ is defined geometrically, as follows.  First, assume that neither point is $\pai$.  Draw a line through the points $P$ and $Q$ (or, if $P=Q$, draw the tangent to the curve at $P$), and let $R$ denote the third point of intersection, defined to be $\pai$ if the line is vertical.  Then $P+Q$ is the reflection of $R$ about the $x$ axis, where the reflection of $\pai$ is $\pai$.  If one of $P$ or $Q$ is $\pai$, we draw a vertical line through the other point, so that $P + \pai = P$ as desired.  Since $\pai$ is the additive identity, we define $\pai + \pai = \pai$.  Reflection about the $x$ axis corresponds to negation, so we can think of the rule as saying that the three points of intersection of a line with the curve sum to $\pai$, as depicted in \fig{elliptic}.

It can be shown that $(E,+)$ is an Abelian group, where the inverse of $P=(x,y)$ is $-P = (x,-y)$.  From the geometric definition, it is clear that this group is Abelian (the line through $P$ and $Q$ does not depend on which point is chosen first) and closed (we always choose $P+Q$ to be some point on the curve).  The only remaining group axiom to check is associativity: we must show that $(P+Q)+T = P+(Q+T)$.

To define the group operation for a general field, it is useful to have an algebraic description of elliptic curve point addition.  Let $P=(x_P,y_P)$ and $Q=(x_Q,y_Q)$.  Provided $x_P \ne x_Q$, the slope of the line through $P$ and $Q$ is
\begin{equation}
  \lambda = \frac{y_Q-y_P}{x_Q-x_P}.
\end{equation}
Computing the intersection of this line with \eq{ec}, we find
\ba
  x_{P+Q} 
          &= \lambda^2 - x_P - x_Q \label{eq:ecaddx}\\          
  y_{P+Q} 
          &= \lambda (x_P - x_{P+Q}) - y_P
\label{eq:ecaddy}
.
\ea
If $x_P= x_Q$, there are two possibilities for $Q$: either $Q=(x_Q,y_Q)= (x_P,y_P)=P$ or $Q=(x_Q,y_Q) = (x_P,-y_P)=-P$. 
If $Q=-P$, then $P+Q=\pai$.  On the other hand, if $P=Q$ (i.e., if we are computing $2P$), then \eqs{ecaddx}{ecaddy} hold with $\lambda$ replaced by the slope of the tangent to the curve at $P$, namely
\begin{equation}
  \lambda = \frac{3 x_P^2 + a}{2 y_P}
\end{equation}
(unless $y_P=0$, in which case the slope is infinite, so $2P=\pai$).  

While the geometric picture does not necessarily make sense for the case of a finite field, we can take its algebraic description as a definition of the group operation.  It is again obvious that addition of points, defined by these algebraic expressions, is commutative and closed.  Associativity of the group operation can be verified by a direct calculation.  This shows that $(E,+)$ is indeed an Abelian group.

Suppose we fix an elliptic curve group $(E,+)$ and choose a point $g \in E$.  Then we can consider the subgroup $\<g\>$, which is possibly the entire group if it happens to be cyclic.  Using exponentiation in this group (which is multiplication in our additive notation), we can define analogs of Diffie-Hellman key exchange and related cryptosystems such as ElGamal.  The security of these cryptosystems then relies on the assumption that the discrete log problem on $\<g\>$ is hard.

In practice, there are many details to consider when choosing an elliptic curve group for cryptographic purposes \cite{MOV96,Buc04}.  Algorithms are known for calculating discrete logarithms on ``supersingular'' and ``anomolous'' curves that run faster than algorithms for the general case, so such curves should be avoided.  At the same time, $g$ should be chosen to be a point of high order.  Curves with the desired hardness properties can be found efficiently, and in the general case it is not known how to solve the discrete log problem over an elliptic curve group classically any faster than by general methods (see \sec{dlog}), which run in time $O(\sqrt p)$.

However, using Shor's algorithm, a quantum computer can solve the discrete log problem for an elliptic curve group over $\F_p$ in time $\poly(\log p)$.  Points on the curve can be represented uniqely by their coordinates, with a special symbol used to denote $\pai$.  Addition of points on the curve can be computed using \eqs{ecaddx}{ecaddy}, which involve only elementary arithmetic operations in the field.  The most complex of these operations is the calculation of modular inverses, which can easily be done using Euclid's algorithm.  For more details on the implementation of Shor's algorithm over elliptic curves, see \cite{PZ03,Kay05,CMMP08}.

Elliptic curve cryptosystems are commonly viewed as being more secure than RSA for a given key size, since the best classical algorithms for factoring run faster than the best classical algorithms for calculating discrete logarithms in elliptic curve groups.  Thus in practice, much smaller key sizes are used in elliptic curve cryptography than in factoring-based cryptography.  Ironically, Shor's algorithm takes a comparable number of steps for both factoring and discrete log,\footnote{Naively, computing the group operations for an elliptic curve using \eqs{ecaddx}{ecaddy} requires slightly more operations than performing ordinary integer multiplication.  However, there are ways to improve the running time of Shor's algorithm for discrete log over elliptic curve groups, at least in certain cases \cite{CMMP08}.} so it could actually be \emph{easier} for quantum computers to break present-day elliptic curve cryptosystems than to break RSA.

One can also define an Abelian group corresponding to a \emph{hyperelliptic curve}, a curve of the form $y^2 = f(x)$ for some suitable polynomial $f$ of degree higher than $3$.  These groups are also candidates for cryptographic applications (see for example \cite[Chapter 6]{Kob98}).  In general, such a group is referred to as the \emph{Jacobian} of the curve; it is no longer isomorphic to the curve itself in the non-elliptic case. The elements of a general Jacobian can be represented uniquely and added efficiently, so that Shor's algorithm can also efficiently compute discrete logarithms over the Jacobian of a hyperelliptic curve.

\subsection{Decomposing Abelian and solvable groups}
\label{sec:decomposing}

Recall from \sec{periodz} that Shor's period-finding algorithm can be used to compute the order of a cyclic group $C=\<g\>$, given the ability to efficiently represent and multiply elements of the group.  More generally, given a black-box representation of some group, it would be useful to have a way of identifying the structure of that group.  For certain kinds of groups, such decompositions can be obtained efficiently by a quantum computer.

These algorithms operate in the framework of \emph{black-box groups} \cite{BS84}.  In this framework, the elements of a group $G$ are represented uniquely by strings of length $\poly(\log|G|)$, and we are given a black box that can compute products or inverses in $G$ as desired.  Of course, any algorithm that works in the black-box setting also works when the group is represented explicitly, say as a matrix group or as some known group.  Note that computing the order of $G$ in the black-box group setting is hard even when $G$ is promised to be Abelian \cite{BS84}.

Suppose we are given a generating set for a finite Abelian black-box group.  Recall that by the fundamental theorem of finite Abelian groups, any such group can be decomposed as a direct product $G \cong \ZZ{p_1^{r_1}} \times \cdots \times \ZZ{p_k^{r_k}}$ of cyclic subgroups of prime power order.  By combining the solution of the Abelian \HSP with classical techniques from computational group theory, there is an efficient quantum algorithm for determining the structure of the group (i.e., the values $p_i^{r_i}$), and furthermore, for obtaining generators for each of the cyclic factors \cite{Mos99,CM01}.
Note that this provides an alternative approach to factoring an integer $N$: by decomposing the multiplicative group $\ZNx$, we learn its size $\varphi(N)$, which is sufficient to determine the factors of $N$ \cite{Mil76,Sho05}.

More generally, a similar decomposition can be obtained for any solvable group \cite{Wat01b}.  A finite group $G$ is called \emph{solvable} if there exist elements $g_1,\ldots,g_m \in G$ such that
\begin{equation}
  \{1\} = H_0 \normalin H_1 \normalin \cdots \normalin H_m = G,
\label{eq:solvabledef}
\end{equation}
where $H_j := \<g_1,\ldots,g_j\>$ for each $j=0,1,\ldots,m$, and where the notation $H_i \normalin H_{j+1}$ indicates 
that $H_j$ is a \emph{normal subgroup} of $H_{j+1}$, i.e., that $xH_j = H_j x$ for every $x\in H_{j+1}$.  (Equivalently, $G$ is solvable if its derived series contains the trivial subgroup.)
Every Abelian group is solvable, but the converse does not hold; for example, $S_3 \cong D_3$ is non-Abelian but solvable.
Given a generating set for a black-box solvable group, there is an efficient probabilistic classical algorithm to find $g_1,\ldots,g_m$ satisfying \eq{solvabledef} for some $m=\poly(\log|G|)$ \cite{BCFLS95}. To compute the order of $G$, it suffices to compute the orders of the quotient groups $H_j/H_{j-1}$ for $j=1,\ldots,m$, which are necessarily cyclic. We cannot directly compute the orders of these groups using Shor's algorithm since we do not have unique encodings of their elements. However, Watrous shows that if we are given the uniform superposition $|H_{j-1}\>$, we can (probabilistically) compute $|H_j/H_{j-1}|$ using a modified version of Shor's algorithm, and also (probabilistically) prepare the state $|H_j\>$. By recursing this procedure along the normal series \eq{solvabledef} (starting with enough copies of $|H_0\>$, and maintaining enough copies of the intermediate states $|H_j\>$, to handle the cases where the algorithm fails), a quantum computer can calculate $|G|$ in polynomial time. By straightforward reductions, this also gives efficient quantum algorithms for testing membership in solvable groups and for deciding whether a subgroup of a solvable group is normal. Similar ideas give a method for determining the structure of any Abelian factor group $G/H$, where $H \normalin G$ \cite{Wat01b}; see also \cite{IMS03} for related work. 

\subsection{Counting points on curves}
\label{sec:kedlaya}

Suppose we are given a polynomial $f \in \Fq[x_1,\ldots,x_n]$ in $n$ variables over the finite field $\Fq$.  The set $H_f := \{x \in \FFsup{q}{n}: f(x)=0\}$ of solutions to the equation $f(x)=0$ is called a \emph{hypersurface.}  Counting the number of solutions $|H_f|$ of this equation is a fundamental computational problem.  More generally, given $m$ polynomials $f_1,\ldots,f_m \in \Fq[x_1,\ldots,x_n]$, we may be interested in the number of solutions to the system of equations $f_1(x)=\cdots = f_m(x) = 0$.  The complexity of such counting problems can be characterized in terms of at least five parameters: the number $m$ of  polynomials, the number $n$ of variables, the degrees $\deg(f_i)$ of the polynomials, the size $q$ of the finite field $\Fq$, and the characteristic $p$ of the field, where $q=p^r$ and $p$ is prime.

The complexity class \sharpP characterizes the difficulty of counting the number of values $x$ such that $f(x)=0$, where $f$ is an efficiently computable function. One can show that for quadratic polynomials over $\FF{2}$, with no restrictions on the number $n$ of variables and the number $m$ of polynomials, the corresponding counting problem is \sharpP-complete. As \sharpP problems are at least as hard as \NP problems (see Section~\ref{sec:qcomplexity}), we do not expect quantum computers to solve such counting problems in time $\poly(n,m)$.  In fact, the counting problem is \sharpP-hard even for a single polynomial in two variables \cite{GKK97} provided we use a sparse representation that only lists the nonzero coefficients of the polynomial, which allows its degree to be exponential in the size of its representation.  Using a non-sparse representation, so that we aim for a running time polynomial in the degree, the computational complexity of such counting problems is a more subtle issue.

Here we are concerned with the counting problem for planar curves, meaning that we have $m=1$ polynomial in $n=2$ variables.  (\app{curves} contains some crucial background information about curves over finite fields for readers unfamiliar with this topic.)  A key parameter characterizing the complexity of this counting problem is the \emph{genus} $g$ of the curve.  For a nonsingular, projective, planar curve $f$, the genus is $g=\frac{1}{2}(d-1)(d-2)$, where $d=\deg(f)$.

\textcite{Sch85} gave an algorithm to count the number of points on an elliptic curve (for which $g=1$) over $\Fq$ in time $\poly(\log q)$.  Following results by \textcite{Pil90}, \textcite{AH01} generalized this result to hyperelliptic curves, giving an algorithm with running time $(\log q)^{O(g^2\log g)}$, where $g$ is the genus of the curve.  For fields $\Fpr$ with characteristic $p$, \textcite{LW02} showed the existence of a deterministic algorithm for counting points with time complexity $\poly(p,r,\deg f)$.
While the former algorithm is efficient for $g=O(1)$, and the latter is efficient for $p=\poly(\log q)$, neither is efficient without some restriction on the genus or the field characteristic.

On the other hand, \textcite{Ked06} explained how the quantum algorithm for determining the structure of an unknown finite Abelian group (\sec{decomposing}) can be used to count the number of points on a planar curve of genus $g$ over $\Fq$ in time $\poly(g, \log q)$.
It is probably fair to say that this constitutes not so much a new quantum algorithm, but rather a novel application of known quantum algorithms to algebraic geometry.

In brief, Kedlaya's algorithm counts the solutions of a smooth, projective curve $C_f$ of genus $g$ by determining the $2g$ nontrivial roots of the corresponding Zeta function $Z_f(T)$, which are determined from the orders of the class groups $\Cl_s(C_f)$ over the different base fields $\FF{p^s}$ for $s=1,\ldots,16g$.  As the class groups are all finite Abelian groups, $|{\Cl}_s(C_f)|$ can be computed in time $\poly(g,s, \log p)$ by a quantum computer, thus giving an efficient quantum algorithm for the point counting problem.
We explain some of the details below.  For further information, see 
\cite{Hul03}, \cite{Lor96}, and the original article by Kedlaya (in increasing order of sophistication).

\paragraph*{The Zeta function of a curve.}
Let the polynomial $f\in\Fp[X,Y]$ define a smooth, planar, projective curve $C_f$.  To count the number of points on this curve in the projective plane $\P^2(\Fp)$, it is useful to consider extensions of the base field.  For any positive integer $r$, we define
\begin{equation}
  N_r := |C_f(\GF{p^r})|,
\end{equation}
where
\be
  C_f(\GF{p^r}) := \{x\in\P^2(\GF{p^r}) : f(x)=0 \}
\ee
denotes the projective curve defined by $f$ when viewed as a polynomial over $\GF{p^r}$.

In terms of these values, we can define the \emph{Zeta function}
$Z_f(T)$ of the curve $C_f$, namely
\begin{equation}
Z_f(T) := \exp\bigg(
\sum_{r=1}^{\infty}\frac{N_r}{r}T^r
\bigg),
\end{equation}
with $T$ a formal variable, and the exponential function defined by the Taylor series $\exp(x)=\sum_{j=0}^\infty x^j/j!$.  Whereas the Riemann zeta function is used to study the elements and primes of the ring $\Z$, the Zeta function of a curve captures the ideals and prime ideals of the ring $\Fp[X,Y]/(f)$, where $(f)$ denotes the ideal generated by $f$.

From the proof of Weil's Riemann hypothesis for curves (see for example \cite[Chap.\ X]{Lor96}),  the Zeta function of a smooth, projective curve $C_f$ of genus $g$ has the form
\begin{equation}
Z_f(T) = \frac{Q_f(T)}{(1-pT)(1-T)},
\end{equation}
where $Q_f(T)$ is a polynomial of degree $2g$ with integer coefficients.  Moreover, $Q_f(T)$ has the factorization
\begin{equation}\label{eq:Nr}
Q_f(T) = \prod_{j=1}^{2g}(1-\alpha_j T)
\end{equation}
with $\alpha_{g+j}={\alpha}^*_j$ and $|\alpha_j|=\sqrt{p}$ for all $j$. By considering the $r$th derivative of $Z_f(T)$ at $T=0$, it is easy to see that the values $\alpha_j$ determine the numbers $N_1,N_2,\dots$, and in particular
\begin{equation}
  N_r = p^r+1-\sum_{j=1}^{2g}{\alpha_j^r}
\end{equation}
for all $r$.
Thus, if we know the integer coefficients of the degree $2g$ polynomial $Q_f(T)$, we can infer the number of points on the curve $f(x)=0$ over $\P^2(\Fpr)$.  Kedlaya's algorithm calculates $Z_f(T)$, and hence $Q_f(T)$, by relating it to the \emph{class group} of the curve, a finite Abelian group.

\paragraph*{The class group of a function field.}
A \emph{divisor} $D$ on a curve $C$ over $\Fp$ is a finite, formal sum over points on the curve extended to the algebraic closure ${\bar\F}_p$ of $\Fp$, namely
\begin{equation}
  D = \sum_{P\in C(\bar{\F}_p)} c_P \cdot P.
\end{equation}
To be a divisor, $D$ must satisfy three conditions: (1) $c_P\in\Z$ for all $P$, (2) $\sum_P |c_P|$ is finite, and (3) $D$ is invariant under the Frobenius automorphism $\phi:x\mapsto x^p$, i.e., $c_P = c_{\phi(P)}$ for all $P$.  The \emph{degree} of $D$ is the integer $\deg(D)=\sum_P c_P$.  Under point-wise addition of the coefficients $c_P$, the divisors of degree $0$ form the group
\begin{equation}
  \Div(C) := \{ D :\deg(D)=0 \}. 
\end{equation}

As explained in \app{curves}, for any curve $C_f$ one can define the \emph{function   field} $\Fp(C_f)$, the field of rational functions $\{ g = g_1/g_2 : g_2\neq 0\}$, where $g_1$ and $g_2$ are homogeneous polynomials of equal degree modulo $(f)$, such that $g$ is a function on the projective curve $C_f(\bar{\F}_p)$.  For each such nonzero rational function $g\in \Fpx(C_f)$ we define the corresponding \emph{principal divisor}
\begin{align}
\div(g) & := \sum_{P\in C_f(\bar{\F}_p)} \ord_P(g)\cdot P \\
 & = \sum_{P\in C_f(\bar{\F}_p)} \ord_P(g_1)\cdot P 
- \sum_{P\in C_f(\bar{\F}_p)} \ord_P(g_2)\cdot P,
\end{align}
where the nonnegative integer $\ord_P(g_i)$ is the multiplicity of $P$ as a solution to $g_i=0$. 
In particular, $\ord_P(g_i)\geq 1$ if and only if $g_i(P)=0$, and $\ord_P(g_i)=0$ when $g_i(P)\neq 0$.

For each principal divisor we have $\deg(\div(g))=0$.  For a rational curve such as the straight line $C=\P^1$, the converse holds as well: the only divisors of degree $0$ are the principal divisors of the curve.  But it is an important fact that for general curves the converse does \emph{not} hold.  For curves that are not rational, i.e., curves of positive genus such as elliptic curves, the class group captures the relationship between the group $\Div(C_f)$ and its subgroup of principal divisors.

There is a crucial equivalence relation $\sim$ among divisors defined by
\begin{equation}
  D_1 \sim D_2 \text{~if and only if~} D_1-D_2 \text{~is a principal divisor}.
\end{equation}
Finally, the \emph{(divisor) class group} $\Cl(C)$ of a curve $C$ is defined as the group of degree $0$ divisors modulo this equivalence relation:
\begin{equation}
  {\Cl}(C) := \Div(C)/\sim.
\end{equation}

Returning to the theory of Zeta functions, it is known that the order of $\Cl(C_f)$ can be expressed in terms of the roots $\alpha_j$ of $Z_f(T)$ as 
\begin{equation} 
 |\Cl(C_f)| = \prod_{j=1}^{2g}(1-\alpha_j).
\end{equation}
This fact establishes a close connection between the number of points $N_1 = |C_f(\Fp)|$ on a curve and the size $|\Cl(C_f)|$ of its class group. 

All of the above can repeated while interpreting the polynomial $f$ as an element of the extended ring $\GF{p^s}[X,Y]$.  Indicating this change of the base field $\Fp$ to its degree $s$ extension $\GF{p^s}$ with a parenthesized superscript, the Zeta function $Z^{(s)}_f(T)$ has $2g$ nontrivial roots $\alpha^{(s)}_j = \alpha_j^s$ for $j=1,\ldots,2g$, and the class group $\Cl^{(s)}(C_f)$ has order
\begin{equation}
|\Cl^{(s)}(C_f)| 
= \prod_{j=1}^{2g}(1-\alpha^s_j).
\end{equation}
Observe that the change of base field affects the class group since the new divisors must be invariant under the Frobenius automorphism $\phi^s:x\mapsto x^{p^s}$ (which is a weaker restriction than the corresponding condition over $\Fp$, making $\Div^{(s)}(C_f)$ larger than $\Div(C_f)$), while the group of principal divisors now allows all rational functions $g\in\GF{p^s}(C_f)$.

To illustrate the above definitions, we present the following extensive example of the class group of an elliptic curve. 

\begin{example}[Point counting and the class group of an elliptic curve]
Consider the elliptic curve $E$ over $\F_2$ defined by the equation $Y^2+XY+X^3+1=0$.  The projective version of $E$ is defined by the homogeneous equation $Y^2Z+XYZ+X^3+Z^3=0$.  We want to consider the number of points $N_r$ in the projective space $\P^2(\F_{2^r})$ for various $r$.

It is not hard to see that $N_1 = 4$ with the four solutions
\begin{equation*} 
\begin{array}{c|cccc}
 & P_0 & P_1 & P_2 & P_3 \\\hline
(X:Y:Z) & (0:1:0) & (1:0:1) & (0:1:1) & (1:1:1) 
\end{array}
\end{equation*}

For the first extension field, there are $N_2=8$ elements in 
$E(\GF{4})$: in addition to the previous four points, we now also have 
the solutions 
\begin{equation*}
\begin{array}{c|cccc}
 & P_4 & P_5 & P_6 & P_7  \\\hline
(X:Y:Z) & (\omega:0:1) & (\omega:\omega:1) 
& (\omega^2:0:1)  & (\omega^2:\omega^2:1) 
\end{array}
\end{equation*}
with $\omega$ an element of the field $\GF{4}$ satisfying $\omega^2=\omega+1$.

In general, it can be shown that the number of points on $E(\GF{2^r})$ is
\begin{align} \label{eq:Nralpha}
N_r & = 2^r+1^r - \alpha^r - \bar{\alpha}^r
\end{align}
for any $r$, where $\alpha := -\smfrac{1}{2}+\smfrac{1}{2}\sqrt{-7}$.

To explore the class group $\Cl(E)$ of this curve, we start by considering
some principal divisors.  For the linear functions in $X,Y,Z$ we find 
the following (degree $3$) divisors:
\begin{equation*}
\begin{array}{c|cccccccc}
\ord_P(f)  & P_0 & P_1 & P_2 & P_3 & P_4 & P_5 & P_6 & P_7 \\\hline
X          &   1 &     &  2  &     \\ 
Y          &     &   1 &     &     &   1 &     &   1 \\
Z          &   3 &     &     &     \\ 
X+Y        &     &     &     &   1 &     &   1 &     & 1 \\
X+Z        &   1 &   1 &     &   1 \\
Y+Z        &     &     &   1 &   2 \\
X+Y+Z      &     &   2 &   1 
\end{array}
\end{equation*}
From this table we see, for example, that the principal divisor of $X/Z$ equals $-2P_0 + 2P_2$, and that $\div((X+Y+Z)/(X+Z)) = -P_0 + P_1 + P_2 - P_3$.  (Note also that in this function field we have equalities such as $X^2+YZ = (X+Z)^2(Y+Z)/(X+Y+Z)$, which confirms that $\div(X^2+YZ) = 2\div(X+Z)+\div(Y+Z)-\div(X+Y+Z) = 2P_0 + 4P_3$.)

One can also show that $P_0-P_1$ is \emph{not} a principal divisor, and hence that $\Cl(E)$ is nontrivial.  In fact, there are four different elements $C_j$ of the class group, which we can indicate by the representatives 
\begin{equation*} 
\begin{array}{c|c|c|c}
 C_0 & C_1 & C_2 & C_3 \\\hline
  0  &  P_0-P_1 & P_0 - P_2 & P_0-P_3  
\end{array}
\end{equation*}
(Note however that these representatives are far from unique, as for example $0 \sim -P_0+P_1+P_2-P_3 = \div((X+Y+Z)/(X+Z))$.)  One can verify that the elements of $\Cl(E)$ act as the group $\ZZ{4}$, with $C_x+C_y \sim C_{x+y}$ for any $x,y\in\ZZ{4}$.

Performing similar calculations over the extension field $\GF{2^s}$, one can show that in general,
\begin{align}
|\Cl^{(s)}(E)| = 
(1-\alpha^s)(1-\bar{\alpha}^s)
= 2^s+1 - \alpha^s-\bar{\alpha}^s,
\end{align}
where $\alpha$ is as in \eq{Nralpha}. 
This concludes our example. 
\end{example}

While for elliptic curves the number of points on the curve equals the number of elements of the corresponding class group, this coincidence does not persist for general curves with genus different from $1$.  However, the class group is nevertheless always a finite Abelian group, which can be explored using the quantum algorithm of \sec{decomposing}. 

 \paragraph*{Kedlaya's algorithm.}
Finally, we describe the quantum algorithm of \textcite{Ked06} for counting the points on a curve over a finite field.
 
\begin{algorithm}[Point counting] \ \\
\emph{Input:} A nonsingular, planar, projective curve $C_f$ defined by a polynomial $f\in\GF{q}[X,Y]$. \\
\emph{Problem:} Determine the number of solutions $|C_f(\GF{q^r})|$ of the equation 
$f=0$ in the projective plane $\P^2(\Fqr)$. \\
\begin{enumerate}
\item Let $g=\frac{1}{2}(d-1)(d-2)$ be the genus of the curve, where $d=\deg(f)$.
\item For $s=1,2,\dots,16g$:
\begin{enumerate}
\item Construct the class group ${\Cl}^{(s)}(C_f)$. 
\item Using the algorithm of \sec{decomposing}, determine $|{\Cl}^{(s)}(C_f)|$. 
\end{enumerate}
\item Using the calculated group sizes and the equalities 
\begin{equation}
 |{\Cl}^{(s)}(C_f)| = \prod_{j=1}^{2g}(1-\alpha_j^s)
\end{equation}
for $s=1,2,\dots,16g$, determine the roots $\alpha_j$. 
\item Compute $N_r = |C_f(\GF{q^r})| = 
q^r+1 - \sum_{j=1}^{2g}{\alpha_j^r}$. 
\end{enumerate}
\end{algorithm}
Several aspects of this algorithm are beyond the scope of this article, most notably the issue of uniquely representing and manipulating the elements of the class group $\Cl(C_f)$  in such a way that they can be sampled (nearly) uniformly, facilitating finding a set of generators.  For an explanation of this and other issues, we refer the reader to the original article and references therein.

In conclusion, note that the above quantum algorithm has running time polynomial in the parameters $\log p^r$ and $g$, whereas the best known classical algorithms are either exponential in $g$ \cite{AH01}, or exponential in $\log p$ \cite{LW02}. 
Whether it is possible to generalize Kedlaya's algorithm for curves to more general surfaces, i.e., to polynomials $f$ with more than $2$ variables,
remains an open question. The best known classical result for this problem is that of \textcite{LW02}, who described an algorithm with
running time $\poly(p^n,r^n,\deg(f)^{n^2})$.

\section{Quantum Algorithms for Number Fields}
\label{sec:numberfield}

\newcommand{\K}{\mathbb{K}}

\subsection{Pell's equation}

Given a squarefree integer $d$ (i.e., an integer not divisible by any perfect square), the Diophantine equation
\begin{equation}
   x^2 - d y^2 = 1
\label{eq:pell}
\end{equation}
is known as \emph{Pell's equation}.  This appellation provides a nice example of Stigler's Law of Eponymy \cite{Sti80} in action, as Pell had nothing whatsoever to do with the equation.  The misattribution is apparently due to Euler, who confused Pell with a contemporary, Brouncker, who had actually worked on the equation.  In fact, Pell's equation was studied in ancient India, where (inefficient) methods for solving it were developed hundreds of years before Pell \cite{Len02}.  (Indeed, Lenstra has suggested that most likely, Pell was named after the equation.)

The left hand side of Pell's equation can be factored as
\begin{equation}
  x^2 - d y^2 = (x + y \sqrt{d})(x - y \sqrt{d}).
\end{equation}
Note that a solution of the equation $(x,y) \in \Z^2$ can be encoded uniquely as the real number $x + y \sqrt{d}$: since $\sqrt{d}$ is irrational, $x + y \sqrt{d} = w + z \sqrt{d}$ if and only if $(x,y)=(w,z)$.  Thus we can also refer to the number $x + y \sqrt{d}$ as a solution of Pell's equation.

There is clearly no loss of generality in restricting our attention to \emph{positive} solutions of the equation, namely those for which $x>0$ and $y>0$.  It is straightforward to show that if $x_1 + y_1 \sqrt{d}$ is a positive solution, then $(x_1 + y_1 \sqrt{d})^n$ is also a positive solution for any $n \in \N$.  In fact, with $x_1 + y_1 \sqrt{d}$ the smallest positive solution of the equation, called the \emph{fundamental solution,} one can show that \emph{all} positive solutions equal $(x_1+y_1\sqrt{d})^n$ for some $n\in\N$. Thus, even though Pell's equation has an infinite number of solutions, we can in a sense find them all by finding the fundamental solution.

Some examples of fundamental solutions for various values of $d$ are shown in \tab{pellsolutions}.  Notice that while the size of the fundamental solution generally increases with increasing $d$, the behavior is far from monotonic: for example, $x_1$ has $44$ decimal digits when $d=6009$, but only $11$ decimal digits when $d=6013$.  In general, though, it is possible for the solutions to be very large: the size of $x_1 + y_1 \sqrt{d}$ is only upper bounded by $2^{O(\sqrt{d} \log d)}$.  Thus it is not even possible to \emph{write down} the fundamental solution with $\poly(\log d)$ bits.

\begin{table}
\begin{tabular}{l@{\hspace{2ex}}p{21ex}@{\hspace{2ex}}p{21ex}}
$d$ & $x_1$ & $y_1$ \\ \hline
2 & 3 & 2 \\
3 & 2 & 1 \\
5 & 9 & 4 \\
$\vdots$ & $\vdots$ & $\vdots$ \\
13 & 649 & 180 \\
14 & 15 & 4 \\
$\vdots$ & $\vdots$ & $\vdots$ \\
6009 & 1316340106327253158\allowbreak
       9259446951059947388\allowbreak
       4013975 $\approx 1.3 \times 10^{44}$ & 
       1698114661157803451\allowbreak
       6889492378831465766\allowbreak
       81644  $\approx 1.6 \times 10^{42}$ \\
6013 & 40929908599 & 527831340 \\
$\vdots$ & $\vdots$ & $\vdots$
\end{tabular}
\caption{Some examples of fundamental solutions of Pell's equation $x^2-dy^2=1$
for different input values $d$ \cite{Joz03}.
}
\label{tab:pellsolutions}
\end{table}

To get around this difficulty, we define the \emph{regulator} of the fundamental solution,
\begin{equation}
  R := \log(x_1 + y_1 \sqrt{d}).
\end{equation}
Since $R = O(\sqrt{d} \log d)$, we can write down $\nint{R}$, the nearest integer to $R$, using $O(\log d)$ bits.  Since $R$ is an irrational number, determining only its integer part may seem unsatisfactory, but in fact, given  $\nint{R}$, there is a classical algorithm to compute $n$ digits of $R$ in time $\poly(\log d, n)$.  Thus we will be satisfied with an algorithm that finds the integer part of $R$ in time $\poly(\log d)$.  The best known classical algorithm for this problem runs in superpolynomial time (for more details, see \sec{principalideal}).  In contrast, \textcite{Hal07} gave a polynomial-time quantum algorithm for computing $\nint{R}$.  For a self-contained review of Hallgren's algorithm, see \cite{Joz03}.

\subsection{From Pell's equation to the unit group}
\label{sec:pellunit}

Given a squarefree positive integer $d$, the \emph{quadratic number field} $\Q[\sqrt{d}]$ is defined as
\begin{equation}
  \Q[\sqrt{d}] := \{ x + y \sqrt{d}: x,y \in \Q\}.
\end{equation}
It is easy to check that $\Q[\sqrt{d}]$ is a field with the usual addition and multiplication operations.  We also define an operation called \emph{conjugation} as
\begin{equation}
  \overline{x + y \sqrt{d}} := x - y \sqrt{d}
.
\end{equation}
One can easily check that conjugation of elements of $\Q[\sqrt{d}]$ has many of the same properties as complex conjugation, and indeed $\Q[\sqrt d]$ behaves in many respects like $\C$, with $\sqrt d$ taking the place of the imaginary unit $\ii = \sqrt{-1}$.  Defining the ring $\Z[\sqrt{d}] \subset \Q[\sqrt{d}]$ as
\begin{equation}
  \Z[\sqrt{d}] := \{ x + y \sqrt{d}: x,y \in \Z\},
\end{equation}
we see that solutions of Pell's equation correspond to those $\xi \in \Z[\sqrt d]$ satisfying $\xi \bar \xi = 1$.

Notice that any solution of Pell's equation, $\xi \in \Z[\sqrt d]$, has the property that its multiplicative inverse over $\Q[\sqrt d]$, $\xi^{-1} = \bar \xi/\xi \bar \xi = \bar \xi$, is also an element of $\Z[\sqrt d]$.  In general, an element of a ring with an inverse that is also an element of the ring is called a \emph{unit}.  In $\Z$, the only units are $\pm 1$, but in other rings it is possible to have more units.

It should not be a surprise that the units of $\Z[\sqrt d]$ are closely related to the solutions of Pell's equation.  In particular,
$\xi = x+y\sqrt{d}$ is a unit in $\Z[\sqrt d]$ if and only if $\xi \bar \xi = x^2 - d y^2 = \pm 1$.
To see this, we note that 
\begin{equation}
  \xi^{-1}
  = \frac{\bar \xi}{\xi \bar \xi}
  = \frac{x - y \sqrt{d}}{x^2 - d y^2},
\end{equation}
and if $x^2 - d y^2 = \pm1$, then clearly $\xi^{-1} = \pm \bar\xi \in \Z[\sqrt d]$.  Conversely, if $\xi^{-1} \in \Z[\sqrt d]$, then so is
\begin{equation}
  \xi^{-1} \overline{\xi^{-1}}
  = \frac{(x - y \sqrt{d})(x + y \sqrt{d})}{(x^2 - d y^2)^2}
  = \frac{1}{x^2 - d y^2},
\end{equation}
which shows that $x^2 - d y^2 = \pm 1$.

The set of units in $\Z[\sqrt{d}]$ forms a group under multiplication called the \emph{unit group}.  This group is given by $\{\pm\epsilon_1^n: n \in \Z\}$, where $\epsilon_1$ is the aforementioned fundamental unit, the smallest unit greater than $1$.  The proof of this fact is essentially the same as the proof that all solutions of Pell's equation are powers of the fundamental solution.

If we can find $\epsilon_1$, then it is straightforward to find all the solutions of Pell's equation.  If $\epsilon_1 = x + y \sqrt{d}$ has $x^2 - d y^2 = +1$, then the units are precisely the solutions of Pell's equation.  On the other hand, if $x^2 - d y^2 = -1$, then $\epsilon_2:=\epsilon_1^2$ satisfies $\epsilon_2 \bar \epsilon_2 = \epsilon_1^2 \bar \epsilon_1^2 = (-1)^2 = 1$; in this case the solutions of Pell's equation are $\{\pm\epsilon_1^{2n}: n \in \Z\}$.  Thus our goal is to find $\epsilon_1$.  Just as in our discussion of the solutions to Pell's equation, $\epsilon_1$ is too large to write down, so instead we compute the \emph{regulator of the fundamental unit}, $\Reg := \log \epsilon_1$.

\begin{example}
Consider the quadratic number field $\Q[\sqrt{5}]$ and the corresponding ring $\Z[\sqrt 5]$.
The unit group of $\Z[\sqrt 5]$ has the fundamental unit $\epsilon_1 = 2 + \sqrt{5}$, whose regulator is $\Reg=\log (2+\sqrt{5}) \approx 1.44$.  Here $\epsilon_1 \bar \epsilon_1 = -1$, so the fundamental solution of Pell's equation is $x_1 + y_1 \sqrt{5} = \epsilon_1^2 = 9 + 4\sqrt{5}$.  Thus the set of positive solutions to Pell's equation $x^2-5y^2=1$ is
\begin{equation}
  \{(x_k,y_k): x_k + y_k \sqrt{5} = (9+4\sqrt{5})^k,\, k \in \N\}.
\end{equation}
\end{example}

\subsection{Periodic function for Pell's equation}
\label{sec:pellperiodic}

To define a periodic function that encodes $\Reg$, we need to introduce the concept of an \emph{ideal} of a ring (and more specifically, a \emph{principal ideal}).  For any ring $R$, we say that $I \subseteq R$ is an ideal if it is closed under integer linear combinations and under multiplication by arbitrary elements of $R$.  For example, $2\Z$ is an ideal of $\Z$.  We say that an ideal is \emph{principal} if it is generated by a single element of the ring, i.e., if it is of the form $\alpha R$ for some $\alpha \in R$; thus $2\Z$ is a principal ideal.

Principal ideals are useful because the function mapping the ring element $\xi \in \Z[\sqrt d]$ to the principal ideal $\xi R$ is periodic, and its periodicity corresponds to the units of $\Z[\sqrt d]$.  Specifically,
$\xi \Z[\sqrt d] = \zeta \Z[\sqrt d]$ if and only if $\xi = \zeta \epsilon$ where $\epsilon$ is a unit in $\Z[\sqrt d]$.
To see this, note that if
$\epsilon$ is a unit, then $\xi \Z[\sqrt d] = \zeta \epsilon \Z[\sqrt d] = \zeta \Z[\sqrt d]$ since $\epsilon \Z[\sqrt d] = \Z[\sqrt d]$ by the definition of a unit.  Conversely, suppose that $\xi \Z[\sqrt d] = \zeta \Z[\sqrt d]$; then, since $1 \in \Z[\sqrt d]$, we have $\xi \in \xi \Z[\sqrt d] = \zeta \Z[\sqrt d]$, so there is some $\mu \in \Z[\sqrt d]$ satisfying $\xi = \zeta \mu$.  Similarly, $\zeta \in \zeta \Z[\sqrt d] = \xi \Z[\sqrt d]$, so there is some $\nu \in \Z[\sqrt d]$ satisfying $\zeta = \xi \nu$.  Thus we have $\xi = \zeta \mu = \xi \nu \mu$.  This shows that $\nu \mu = 1$, so $\mu$ and $\nu$ are units (indeed, $\nu = \mu^{-1}$).

As a result, the function $g(\xi)=\xi \Z[\sqrt d]$ is (multiplicatively) periodic with period $\epsilon_1$.  In other words, letting $\xi = \e^z$, the function
\begin{equation}
  h(z)=\e^z \Z[\sqrt d]
\label{eq:periodic1}
\end{equation}
is (additively) periodic with period $\Reg$.  However, we cannot simply use this function since it is not possible to succinctly represent the values it takes.

To define a more suitable periodic function, one can use the concept of a \emph{reduced} ideal.  We will not describe the details here.  However, one can show that there are only finitely many reduced principal ideals, and indeed only $O(d)$ of them, so that we can represent a reduced principal ideal using $\poly(\log d)$ bits.

It is also helpful to have a way of measuring the distance $\delta$ of any principal ideal from the \emph{unit ideal}, $1\Z[\sqrt d]=\Z[\sqrt d]$.  Such a function can be defined by
\begin{equation}
  \delta(\xi \Z[\sqrt d])
  := \log \left| \frac{\xi}{\bar \xi} \right| \bmod \Reg .
\end{equation}
Notice that the unit ideal has distance $\delta(1\Z[\sqrt d]) = \log |1/1| \bmod \Reg = 0$, as desired.  Furthermore, the distance function does not depend on which generator we choose to represent an ideal, since
two equivalent ideals have generators that differ by some unit $\epsilon=\epsilon_1^n$, and
\begin{align}
  \delta(\epsilon \Z[\sqrt d])
  &= 2 \log |\epsilon| \bmod \Reg 
   = 0.
\end{align}
With this definition of distance, one can show  that there is a reduced ideal close to any non-reduced ideal.

The periodic function $f(z)$ used in Hallgren's algorithm is defined as the reduced principal ideal whose distance from the unit ideal is maximal among all reduced principal ideals of distance at most $z$ (together with the distance from $z$, to ensure that the function is injective within each period).  In other words, we select the reduced principal ideal ``to the left of, or at, $z$.''

This function $f$ is periodic with period $\Reg$, and one can show that it can be computed in time $\poly(\log d)$.  However, since $\Reg$ is in general irrational, it remains to see how to perform period finding for such a function.

\subsection{Period finding over \texorpdfstring{$\R$}{R}} \label{sec:periodr}

Suppose we are given a function $f: \R \to S$ satisfying
\begin{equation}
  f(x) = f(y) \text{~if and only if~} \frac{x-y}{r} \in \Z
\end{equation}
for some $r \in \R$, for all $x,y \in \R$.  Here we consider how Shor's period-finding algorithm (\sec{periodz}) can be adapted to find an approximation to $r$, even if it happens to be irrational \cite{Hal07}.

Of course, to perform period finding on a digital computer, we must discretize the function.  We must be careful about how we perform this discretization.  For example, suppose that $S=\R$.  If we simply evaluate $f$ at equally spaced points and round the resulting values to obtain integers, there is no reason for the function values corresponding to inputs separated by an amount close to the period to be related in any way whatsoever.  It could be that the discretized function is injective, carrying absolutely no information about the period.

Instead we will discretize in such a way that the resulting function is \emph{pseudoperiodic}.  We say that $f: \Z \to S$ is \emph{pseudoperiodic at $k \in \Z$ with period $r \in \R$} if for each $\ell \in \Z$, either $f(k)=f(k + \floor{\ell r})$ or $f(k)=f(k - \ceil{\ell r})$.  We say that $f$ is \emph{$\epsilon$-pseudoperiodic} if it is pseudoperiodic for at least an $\epsilon$ fraction of the values $k=0,1,\ldots,\floor{r}$.  We will require that the discretized function is $\epsilon$-pseudoperiodic for some constant $\epsilon$, and that it is injective on the subset of inputs where it is pseudoperiodic.  The periodic function encoding the regulator of Pell's equation can be constructed so that it satisfies these conditions.

The algorithm for period finding over $\R$ closely follows \alg{periodz}.  Again the basic approach is Fourier sampling over $\ZN$, with $N$ depending on some a priori upper bound on the period.

\begin{algorithm}[Period finding for a pseudoperiodic function] \label{alg:periodr}\ \\
\emph{Input:} Black box $f: \Z \to S$ that is $\epsilon$-pseudoperiodic (for some $\epsilon=\Omega(1)$) with period $r \in \R$. \\
\emph{Problem:} Approximate $r$.

\begin{enumerate}
\item
Prepare the uniform superposition $|\ZN\>$.
\item
Query the pseudoperiodic function in an ancilla register, giving
\begin{equation}
\frac{1}{\sqrt{N}} \sum_{x \in \ZN} |x,f(x)\>.
\end{equation}
\item
Discard the ancilla register, so that the first register is left in a uniform superposition over those $x$ for which $f(x)$ takes some particular value.  With constant probability, this is a value at which $f$ is pseudoperiodic.  Suppose that this value is $f(x_0)$ where $0 \le x_0 \le r$.  As in step \ref{item:periodzsample} of \alg{periodz}, the first register is a superposition over $n \approx N/r$ points, with the rounding depending on the particular value of $x_0$.  Let us write $[\ell]$ to denote an integer that could be either $\floor{\ell}$ or $\ceil{\ell}$.  With this notation, we obtain the state
\begin{equation}
  \frac{1}{\sqrt n} \sum_{j=0}^{n-1} |x_0 + [jr]\>.
\end{equation}
\item
Perform the Fourier transform over $\ZN$, giving
\begin{align} 
  \frac{1}{\sqrt{nN}} \sum_{k \in \ZN} \rouN^{k x_0} \sum_{j=0}^{n-1} \rouN^{k[jr]} |k\>.
\label{eq:step4}
\end{align}
We have $[jr] = jr+\delta_j$ where $-1 < \delta_j < 1$, so the sum over $j$ above is
\begin{equation}
  \sum_{j=0}^{n-1} \rouN^{k[jr]} = \sum_{j=0}^{n-1} \rouN^{kjr} \rouN^{k\delta_j}.
\end{equation}
When the offsets $\delta_j$ are zero, this is simply \eq{periodzsum}, which we have already shown is strongly peaked around values of $k$ close to integer multiples of $N/r$.  To compare with this case, we compute the deviation
\begin{align}
  \bigg|\sum_{j=0}^{n-1} \rouN^{kjr} \rouN^{k\delta_j} 
      - \sum_{j=0}^{n-1} \rouN^{kjr}\bigg|
  &\le \sum_{j=0}^{n-1} |\rouN^{k \delta_j} - 1| \\
  &\le \frac{1}{2} \sum_{j=0}^{n-1} \Big|\frac{\pi k \delta_j}{N}\Big| \\
  &\le \frac{\pi k n}{2N}.
\end{align}
This bound does not show that the amplitudes are close for all values of $k$.  However, suppose we restrict our attention to those values of $k$ less than $N/\log r$.  (We obtain such a $k$ with probability about $1/\log r$, so we can condition on such a value with only a polynomial increase in the overall running time.)  Then if $k=\nint{jN/r}$ for some $j \in \Z$, we find (using \eq{periodzprobbound})
\begin{align}
  \bigg| \frac{1}{\sqrt{nN}} \sum_{j=0}^{n-1} \rouN^{k[jr]} \bigg|
  &= \Omega\Big(\frac{1}{\sqrt r}\Big).
\end{align}
\item
Measure the state of \eq{step4} in the computational basis.  As in step \ref{item:periodzmeasure} of \alg{periodz}, we sample from a distribution in which some value $k=\nint{jN/r}$ (with $j \in \Z$) appears with reasonably large probability (now $\Omega(1/\poly(\log r))$ instead of $\Omega(1)$).
\item
Finally, we must obtain an approximation to $r$ using these samples.  Since $r$ is not an integer, the procedure from step \ref{item:periodzcfe} of \alg{periodz} does not suffice.  However, we can perform Fourier sampling sufficiently many times that we obtain two values $\nint{jN/r},\nint{j'N/r}$ where $j$ and $j'$ are relatively prime, again with only polynomial overhead.  
It can be shown that 
if $N \ge 3r^2$, then $j/j'$ is guaranteed to be one of the convergents in the continued fraction expansion of $\nint{jN/r}/\nint{j'N/r}$.  Thus we can learn $j$, and hence compute $jN/\nint{jN/r}$, which gives a good approximation to $r$: in particular, $|r-\nint{jN/\nint{jN/r}}| \le 1$.
\end{enumerate}
\end{algorithm}

\subsection{The principal ideal problem and number field cryptography}
\label{sec:principalideal}

Pell's equation is closely related to another problem in algebraic number theory called the \emph{principal ideal problem}.  Fix a quadratic number field $\Q[\sqrt d]$, and suppose we are given an invertible ideal $I$, an ideal for which there exists some $J \subseteq \Q[\sqrt d]$ with $IJ = \Z[\sqrt d]$.  In the principal ideal problem, we are asked to decide whether there is some $\alpha \in \Z[\sqrt d]$ such that $I=\alpha\Z[\sqrt d]$ (i.e., whether $I$ is principal), and if so, to find that $\alpha$ (or more precisely, $\nint{\log\alpha}$).  Notice that computing $\alpha$ can be viewed as an analog of the discrete logarithm problem in $\Z[\sqrt d]$.  Using similar ideas as in the algorithm for solving Pell's equation, and proceeding along similar lines to \alg{dlog}, \textcite{Hal07} also gave an efficient quantum algorithm for the principal ideal problem.

The integer factoring problem reduces to solving Pell's equation, and Pell's equation reduces to the principal ideal problem \cite{BW89}; but no reductions in the other direction are known.  Indeed, whereas factoring is conjectured to be possible with a classical computer in time $2^{O((\log d)^{1/3} (\log\log d)^{1/3})}$, the best known classical algorithms for Pell's equation and the principal ideal problem both take time $2^{O(\sqrt{\log d \log\log d})}$ assuming the generalized Riemann hypothesis, or time $O(d^{1/4} \poly(\log d))$ with no such assumption \cite{Buc89,Vol00}.  Motivated by the possibility that the principal ideal problem is indeed harder than factoring, Buchmann and Williams proposed a key exchange protocol based on it \cite{BW89}.  This system is analogous to the Diffie-Hellman protocol discussed in \sec{diffiehellman}, but instead of exchanging integers, Alice and Bob exchange reduced ideals.  Hallgren's algorithm shows that quantum computers can efficiently break the Buchmann-Williams cryptosystem.

\subsection{Computing the unit group of a general number field}
\label{sec:unitgroup}

Recall from \sec{pellunit} that the quantum algorithm for solving Pell's equation proceeds by computing the fundamental unit of the unit group of $\Z[\sqrt d]$.  More generally, there is an efficient quantum algorithm to compute the unit group of an arbitrary number field of fixed degree \cite{Hal05,SV05}, which we briefly summarize.

In general, an \emph{algebraic number field} (or simply \emph{number field}) $\K=\Q[\theta]$ is a finite extension of the field $\Q$ of rational numbers.  Here $\theta$ is a root of some monic irreducible polynomial over $\Q$ called the \emph{minimal polynomial}.  If the minimal polynomial has degree $n$, we say that $\K$ is a number field of degree $n$.  For example, the quadratic number field $\Q[\sqrt d]$ has the minimal polynomial $x^2-d$, and hence is of degree $2$.

For a general number field, the units are defined as the algebraic integers of that field whose inverses are also algebraic integers.  In general, the units form a group under multiplication.  Just as the units of a quadratic number field are powers of some fundamental unit, it can be shown that the unit group $U(\K)$ of any number field $\K$ consists of elements of the form $\zeta \epsilon_1^{n_1} \cdots \epsilon_r^{n_r}$ for $n_1,\ldots,n_r \in \Z$, where $\zeta$ is a root of unity and $
\epsilon_1,\ldots,\epsilon_r$ are called the \emph{fundamental units} (with $r$ defined below).  Given a number field of constant degree (say, in terms of its minimal polynomial), $\zeta$ can be computed efficiently by a classical computer.  The \emph{unit group problem} asks us to compute (the regulators of) the fundamental units $\epsilon_1,\ldots,\epsilon_r$.

As in the quantum algorithm for solving Pell's equation, we can reduce this computation to a period-finding problem.  To see how the periodic function is defined, suppose the minimal polynomial of $\K$ has $s$ real roots and $t$ pairs of complex roots; then the number of fundamental units is $r=s+t-1$.
Let $\theta_1,\dots,\theta_{s}$ be the $s$ real roots, and let $\theta_{s+1},\dots,\theta_{s+t}$ be $t$ complex roots that, together with their complex conjugates $\theta_{s+1}^*,\dots,\theta_{s+t}^*$, constitute all $2t$ complex roots.  For each $j=1,\ldots,s+t$, we can embed $\K$ in $\C$ with the map $\sigma_j:\K \to \C$ that  replaces $\theta$ by $\theta_j$.  Then we define a function $L:\K^\times \to \R^{s+t}$ as
\begin{equation}
  \begin{aligned}
    L(x) :=
    (&\log|\sigma_1(x)|,\ldots,\log|\sigma_s(x)|, \\
     &2\log|\sigma_{s+1}(x)|,\ldots, 2\log|\sigma_{s+t}(x)|).
  \end{aligned}
\end{equation}
By Dirichlet's theorem, 
$L(U(\K))$ is an $r$-dimensional lattice in $\R^{r+1}$ whose coordinates $(y_1,\dots,y_{r+1})$ obey $\sum_j y_j = 0$ \cite[Thm. 4.9.7]{Coh93}.  The unit group problem is essentially equivalent to finding a basis for this lattice, i.e., determining the periodicity of $L(x)$.  Note that since the lattice has dimension $r$, we can restrict our attention to any $r$ components of $L(x)$, thereby giving a period finding problem over $\R^r$.

There are two main parts to the quantum algorithm for computing the unit group, again paralleling the algorithm for Pell's equation.  First, one must show how to efficiently compute the function $L(x)$, or more precisely, a related function that hides the same lattice, analogous to the function discussed in \sec{pellperiodic} (and again based on the concept of a reduced ideal).  Second, one must generalize period finding over $\R$ (\sec{periodr}) to period finding over $\R^r$.  All relevant computations can be performed efficiently provided the degree of $\K$ is constant, giving an efficient quantum algorithm for the unit group problem in this case.

\subsection{The principal ideal problem and the class group}
\label{sec:classgroup}

We conclude our discussion of quantum algorithms for number fields by mentioning two additional problems with efficient quantum algorithms.

In \sec{principalideal}, we saw that the efficient quantum algorithm for Pell's equation can be adapted to efficiently decide whether a given ideal is principal, and if so, to compute (the regulator of) its generator.  More generally, the principal ideal problem can be defined for any number field, and the techniques discussed in \sec{unitgroup} can be applied to give an efficient quantum algorithm for it whenever the number field has constant degree \cite{Hal05}.

A related problem is the task of computing the \emph{class group} $\Cl(\K)$ of a number field $\K$.  The class group is defined as the set of ideals of $\K$ modulo the set of principal ideals of $\K$; it is a finite Abelian group.  The \emph{class group problem} asks us to decompose $\Cl(\K)$ in the sense of \sec{decomposing}.  Assuming the generalized Riemann hypothesis (GRH), there is a polynomial-time algorithm to find generators of $\Cl(\K)$ \cite{Thi95}.  If $\K=\Q[\sqrt{-d}]$ is an imaginary quadratic number field, then its elements have unique representatives that can be computed efficiently, and $\Cl(\K)$ can be decomposed using the procedure of \cite{Mos99,CM01}.  More generally, it is not known how to uniquely represent the elements of $\Cl(\K)$ in an efficiently computable way.  However, we can take advantage of the technique introduced in \cite{Wat01b} for computing over quotient groups, namely to represent a coset by the uniform superposition of its elements.  Using this idea, it can be shown that there is an efficient quantum algorithm for decomposing $\Cl(\K)$, provided $\K$ has constant degree and assuming the GRH \cite{Hal05} (see \cite{Hal07} for the special case of a real quadratic number field).  In particular, we can efficiently compute $|\Cl(\K)|$, the \emph{class number} of the number field $\K$, just as Kedlaya's algorithm does for curves (\sec{kedlaya}).
\section{Non-Abelian Quantum Fourier Transform}\label{sec:nonabelQFT}

In \sec{abelHSP}, we saw that the Abelian Fourier transform can be used to exploit the symmetry of an Abelian HSP, and that this essentially gave a complete solution.  In the non-Abelian version of the HSP, we will see that a non-Abelian version of the Fourier transform can similarly be used to exploit the symmetry of the problem.  However, in general, this will only take us part of the way to a solution of the non-Abelian HSP.

\subsection{The Fourier transform over a non-Abelian group}

We begin by discussing the definition of the non-Abelian Fourier transform.  For a more extensive review of Fourier analysis on finite groups, we recommend the books by \textcite{Dia88,Ter99,Ser77}.  Here we assume knowledge of group representation theory; see \app{repr} for a summary of the requisite background.

The \emph{Fourier transform} of the state $|x\>$ corresponding to the group element $x\in G$ is a weighted superposition over a complete set of irreducible representations $\hat G$, namely
\begin{equation}
  \ket{\hat{x}}
  := \frac{1}{\sqrt{|G|}}\sum_{\sigma\in\hat{G}}{
     d_\sigma\ket{\sigma,\sigma(x)}},
\end{equation}
where $d_\sigma$ is the dimension of the representation $\sigma$, $\ket{\sigma}$ is a state that labels the irreducible representation (or \emph{irrep}), and $\ket{\sigma(x)}$ is a normalized, $d_\sigma^2$-dimensional state whose amplitudes are given by the entries of the $d_\sigma \times d_\sigma$ matrix $\sigma(x)/\sqrt{d_\sigma}$:
\begin{align}
  \ket{\sigma(x)}
  &:= (\sigma(x)\otimes \I_{d_\sigma})
      \sum_{j=1}^{d_\sigma}{\frac{\ket{j,j}}{\sqrt{d_\sigma}}} \\
  &=  \sum_{j,k=1}^{d_\sigma} \frac{\sigma(x)_{j,k}}{\sqrt{d_\sigma}} |j,k\>.
\end{align}
Here $\sigma(x)$ is a unitary matrix representing the group element $x \in G$; we have $\sigma(x)\sigma(y)=\sigma(xy)$ for all $x,y \in G$.
(If $\sigma$ is one dimensional, then $\ket{\sigma(x)}$ is simply a phase factor $\sigma(x)\in\C$ with $|\sigma(x)|=1$.)
In other words, the Fourier transform over $G$ is the unitary matrix
\begin{align}
  F_G
  &:= \sum_{x \in G} |\hat x\>\<x| \\
  &=  \sum_{x \in G} \sum_{\sigma \in \hat G}
      \sqrt{\frac{d_\sigma}{|G|}} \sum_{j,k=1}^{d_\sigma}
      \sigma(x)_{j,k} \, |\sigma,j,k\>\<x|.
\end{align}
Note that the Fourier transform over a non-Abelian $G$ is not uniquely defined, but rather, depends on a choice of basis for each irrep of dimension greater than $1$.

It is straightforward to check that $F_G$ is indeed a unitary transformation.
Using the identity
\begin{align}
  \braket{\sigma(y)}{\sigma(x)}
  = \Tr \big(\sigma^\dag(y) \sigma(x)\big) / d_\sigma
  = \chi_{\sigma}(y^{-1}x)/d_\sigma,
\end{align}
we have
\begin{align}
  \braket{\hat{y}}{\hat{x}}
  &= \sum_{\sigma \in \hat G}{\frac{d_\sigma^2}{|G|}
     \braket{\sigma(y)}{\sigma(x)}} \\
  &= \sum_{\sigma \in \hat G}{\frac{d_\sigma}{|G|} \chi_\sigma(y^{-1} x)}.
\end{align}
Hence by \eqs{dim2sum}{regrepsum} in \app{repr}, we see that $\braket{\hat{y}}{\hat{x}} = \delta_{x,y}$.

As noted in \app{repr}, $F_G$ is precisely the transformation that simultaneously block-diagonalizes the actions of left and right multiplication, or equivalently, that decomposes both the left and right regular representations of $G$ into their irreducible components.  Let us check this explicitly for the left regular representation $L$ of $G$.  This representation satisfies $L(x) |y\> = |xy\>$ for all $x,y \in G$, so 
\begin{align}
  \hat L(x)
  &:= F_G \, L(x) \, F_G^\dag \\
  &= \sum_{y \in G} |\widehat{xy}\>\<\hat y| \\
  &= \sum_{y \in G} \sum_{\sigma,\sigma' \in \hat G}
     \sum_{j,k=1}^{d_\sigma} \sum_{j',k'=1}^{d_{\sigma'}}
     \frac{\sqrt{d_\sigma d_{\sigma'}}}{|G|} \nn
  &\qquad
     \sigma(xy)_{j,k} \, \sigma'(y)^*_{j',k'} \,
     |\sigma,j,k\>\<\sigma',j',k'| \\
  &= \sum_{y \in G} \sum_{\sigma,\sigma' \in \hat G}
     \sum_{j,k,\ell=1}^{d_\sigma} \sum_{j',k'=1}^{d_{\sigma'}}
     \frac{\sqrt{d_\sigma d_{\sigma'}}}{|G|} \nn
  &\qquad
     \sigma(x)_{j,\ell} \, \sigma(y)_{\ell,k} \, \sigma'(y)^*_{j',k'} \,
     |\sigma,j,k\>\<\sigma',j',k'| \\
  &= \sum_{\sigma \in \hat G}
     \sum_{j,k,\ell=1}^{d_\sigma}
     \sigma(x)_{j,\ell} \,
     |\sigma,j,k\>\<\sigma,\ell,k| \\
  &= \bigoplus_{\sigma \in \hat G}
     \big(\sigma(x) \otimes \I_{d_\sigma}\big),
\end{align}
where in the fourth line we have used the orthogonality relation for irreps (\thm{orthorep} in \app{repr}).

A similar calculation can be done for the right regular representation defined by $R(x) |y\> = |y x^{-1}\>$, giving
\begin{align}
  \hat R(x)
  &:= F_G \, R(x) \, F_G^\dag \\
  &= \bigoplus_{\sigma \in \hat G}
     \big(\I_{d_\sigma} \otimes \sigma(x)^* \big).
\label{eq:rightreghat}
\end{align}
This identity will be useful when analyzing the application of the quantum Fourier transform to the hidden subgroup problem in \sec{nonabelHSP}.

\subsection{Efficient quantum circuits}
\label{sec:nonabelQFTcircuits}

In \sec{abelQFTcircuits}, we described efficient quantum circuits for implementing the quantum Fourier transform over any finite Abelian group.  Analogous circuits are known for many, but not all, non-Abelian groups.
Just as the circuit for the \QFT over a cyclic group parallels the usual classical fast Fourier transform (FFT), many of these circuits build on classical implementations of FFTs over non-Abelian groups \cite{Bet87,Cla89,Roc90,DR90,MR95}.
Here we briefly summarize the groups for which efficient {\QFT}s are known.

\textcite{Hoy97} gave efficient circuits for the quantum Fourier transform over metacyclic groups (i.e., semidirect products of cyclic groups), including the dihedral group, and over the Pauli group on $n$ qubits.  An alternative construction for certain metacyclic $2$-groups is given in \cite{PRB99}.  Beals gave an efficient implementation of the \QFT over the symmetric group \cite{Bea97}.  Finally, \textcite{MRR06} gave a general construction of {\QFT}s, systematically quantizing classical FFTs.  For example, this approach yields polynomial-time quantum circuits for Clifford groups, the symmetric group, the wreath product of a polynomial-sized group, and metabelian groups.

There are a few important groups for which efficient quantum Fourier transforms are \emph{not} known.  These include the classical groups, such as the group $\GL_n(\Fq)$ of $n \times n$ invertible matrices over a finite field with $q$ elements.  However, it is possible to implement these transforms in subexponential time \cite{MRR06}.

\section{Non-Abelian Hidden Subgroup Problem}\label{sec:nonabelHSP}

We now turn to the general, non-Abelian version of the hidden subgroup problem.  We begin by stating the problem and describing some of its potential applications.  Then we describe the standard way of approaching the problem on a quantum computer, and explain how the non-Abelian Fourier transform can be used to simplify the resulting hidden subgroup states.  This leads to the notions of weak and strong Fourier sampling; we describe some of their applications and limitations.  Then we discuss how multi-register measurements on hidden subgroup states can potentially avoid some of those limitations.  Finally, we describe two specific algorithmic techniques for hidden subgroup problems: the Kuperberg sieve and the pretty good measurement.  Note that some of the results presented in \sec{hiddenshift}, on the hidden shift problem, also give algorithms for the non-Abelian \HSP.

\subsection{The problem and its applications}\label{sec:hspandapps}

The non-Abelian hidden subgroup problem naturally generalizes the Abelian \HSP considered in \sec{abelHSP}.  In the hidden subgroup problem for a group $G$, we are given a black box function $f: G \to S$, where $S$ is a finite set.  We say that $f$ hides a subgroup $H \le G$ provided
\be
  f(x) = f(y) \text{~if and only if~} x^{-1} y \in H
\label{eq:hsphides}
\ee
(where we use multiplicative notation for non-Abelian groups).
In other words, $f$ is constant on left cosets $H, g_1 H, g_2 H, \ldots$ of $H$ in $G$, and distinct on different left cosets.  We say that an algorithm for the \HSP in $G$ is efficient if it runs in time $\poly(\log|G|)$.

The choice of left cosets is an arbitrary one; we could just as well define the \HSP in terms of right cosets $H, H g_1', H g_2', \ldots$, by promising that $f(x)=f(y)$ if and only if $x y^{-1} \in H$.  But here we will use the definition in terms of left cosets.

The non-Abelian \HSP is of interest not only because it generalizes the Abelian case in a natural way, but because a solution of certain non-Abelian {\HSP}s would have particularly useful applications.  The most well-known (and also the most straightforward) applications are to the \emph{graph automorphism problem} and the \emph{graph isomorphism problem} \cite{BL95,Bea97,Hoy97,EH99}.

In the graph automorphism problem, we are given a graph $\Gamma$ on $n$ vertices, and our goal is to determine its automorphism group.  We say that $\pi \in S_n$ is an automorphism of $\Gamma$ if $\pi(\Gamma)=\Gamma$.  The automorphisms of $\Gamma$ form a group $\Aut\Gamma \le S_n$; if $\Aut\Gamma$ is trivial then we say $\Gamma$ is \emph{rigid}.  We may cast the graph automorphism problem as an \HSP over $S_n$ by considering the function $f(\pi):=\pi(\Gamma)$, which hides $\Aut\Gamma$.

In the graph isomorphism problem, we are given two connected graphs $\Gamma,\Gamma'$, each on $n$ vertices, and our goal is to determine whether there is any permutation $\pi \in S_n$ such that $\pi(\Gamma)=\Gamma'$, in which case we say that $\Gamma$ and $\Gamma'$ are \emph{isomorphic}.  We can cast graph isomorphism as an \HSP in the wreath product $S_n \wr S_2 \le S_{2n}$.  (The wreath product group $G \wr T$, where $T \le S_m$, is the semidirect product $G^m \semidirect T$, where $T$ acts to permute the elements of $G^m$.)  Writing the elements of $S_n \wr S_2$ in the form $(\sigma,\tau,b)$ where $\sigma,\tau \in S_n$ represent permutations of $\Gamma,\Gamma'$, respectively, and $b \in \{0,1\}$ denotes whether to swap the two graphs, by defining
\be
  f(\sigma,\tau,b) :=
  \begin{cases}
    (\sigma(\Gamma),\tau(\Gamma')) & b=0 \\
    (\sigma(\Gamma'),\tau(\Gamma)) & b=1.
  \end{cases}
\ee
The function $f$ hides the automorphism group of the disjoint union of $\Gamma$ and $\Gamma'$.  This group contains an element that swaps the two graphs, and hence is at least twice as large as $|{\Aut\Gamma}|\cdot|{\Aut\Gamma'}|$, if and only if the graphs are isomorphic.
In particular, if $\Gamma$ and $\Gamma'$ are rigid (which seems to be a hard case for the \HSP approach to graph isomorphism, and in fact is equivalent to the problem of deciding rigidity \cite[Sec.\ VI.6]{Hof82}), the hidden subgroup is trivial when $\Gamma,\Gamma'$ are non-isomorphic, and has order two, with its nontrivial element the involution $(\pi,\pi^{-1},1)$, when $\Gamma=\pi(\Gamma')$.

The graph automorphism and graph isomorphism problems are closely related.  The decision version of graph isomorphism is polynomial-time equivalent to the problems of finding an isomorphism between two graphs provided one exists, counting the number of such isomorphisms, finding the automorphism group of a single graph, and computing the size of this automorphsim group \cite{Hof82}.  Deciding whether a graph is rigid (i.e., whether the automorphsim group is trivial) can be reduced to general graph isomorphism, but the other direction is unknown, so deciding rigidity could be an easier problem \cite{KST93}.

We should point out the possibility that graph isomorphism is not a hard problem, even for classical computers.  There are polynomial-time classical algorithms for many special cases of graph isomorphism, such as when the maximum degree is bounded \cite{Luk82}, the genus is bounded \cite{Mil80,FM80}, or the eigenvalue multiplicity is bounded \cite{BGM82}.  Furthermore, there are classical algorithms that run in time $2^{O(\sqrt{n \log n})}$ for general graphs \cite{BKL83}; and in time $2^{O(\sqrt[3]{n})}$ for strongly regular graphs \cite{Spi96}, which are suspected to be some of the hardest graphs for the problem.  Even if there is a polynomial-time quantum algorithm for graph isomorphism, it is plausible that the \HSP in the symmetric group might be substantially harder, since the graph structure is lost in the reduction to the \HSP.  Indeed, solving the \HSP in the symmetric group would equally well solve other isomorphism problems, such as the problem of code equivalence \cite{EH99}, which is at least as hard as graph isomorphism, and possibly harder \cite{PR97}.

The second major potential application of the hidden subgroup problem is to lattice problems.  An $n$-dimensional lattice is the set of all integer linear combinations of $n$ linearly independent vectors in $\R^n$ (a \emph{basis} for the lattice).  In the \emph{shortest vector problem}, we are asked to find a shortest nonzero vector in the lattice (see for example \textcite{MG02}).  In particular, in the \emph{g(n)-unique shortest vector problem}, we are promised that the shortest nonzero vector is unique (up to its sign), and is shorter than any other non-parallel vector by a factor $g(n)$.  This problem can be solved in polynomial time on a classical computer if $g(n)=(1+\epsilon)^{\Omega(n)}$ \cite{LLL82}, and indeed even if $g(n)=2^{\Omega(n \log\log n/\log n)}$ \cite{Sch87,AKS01}.  The problem is \NP-hard if $g(n)=O(1)$ \cite{Emd81,Ajt98,Mic01}; in fact, even stronger hardness results are known
\cite{Kho05}.  Even for $g(n)=\poly(n)$, the problem is suspected to be hard, at least for a classical computer.  In particular, the presumed hardness of the $O(n^8)$-unique shortest vector problem is the basis for a cryptosystem proposed by \textcite{Ajt96,AD97,MG02}, and a subsequent improvement by \textcite{Reg03} requires quantum hardness of the $O(n^{1.5})$-shortest vector problem.

Regev showed that an efficient quantum algorithm for the dihedral hidden subgroup problem based on the standard method (described below) could be used to solve the $\poly(n)$-unique shortest vector problem \cite{Reg04}.  Such an algorithm would be significant since it would break these lattice cryptosystems, which are some of the few proposed cryptosystems that are not compromised by Shor's algorithm.

So far, only the symmetric and dihedral hidden subgroup problems are known to have applications to natural problems.  Nevertheless, there has been considerable interest in understanding the complexity of the \HSP for general groups.  There are at least three reasons for this.  First, the problem is simply of fundamental interest: it appears to be a natural setting for exploring the extent of the advantage of quantum computers over classical ones.  Second, techniques developed for other {\HSP}s may eventually find application to the symmetric or dihedral groups.  Finally, exploring the limitations of quantum computers for {\HSP}s may suggest cryptosystems that could be robust even to quantum attacks \cite{OTU00,Reg04,KKNY05,HKK08,MRV07}.

\subsection{The standard method}\label{sec:standardmeth}

Nearly all known algorithms for the non-Abelian hidden subgroup problem use the black box for $f$ in essentially the same way as in the Abelian \HSP (\sec{finiteabelHSP}).  This approach has therefore come to be known as the \emph{standard method}.

In the standard method, we begin by preparing a uniform superposition over group elements:
\be
  |G\> := \frac{1}{\sqrt{|G|}} \sum_{x \in G} |x\>.
\ee
We then compute the value $f(x)$ in an ancilla register, giving 
\be
  \frac{1}{\sqrt{|G|}} \sum_{x \in G} |x,f(x)\>.
\ee
Finally, we 
discard the second register.
If we were to measure the second register, obtaining the outcome $y \in S$, then the state would be projected onto the uniform superposition of those $x \in G$ such that $f(x)=y$, which 
is simply some left coset of $H$.  Since every coset contains the same number of elements, each left coset occurs with equal probability.  Thus discarding the second register yields the \emph{coset state}
\be
  |xH\> := \frac{1}{\sqrt{|H|}} \sum_{h \in H} |xh\>
    \quad\parbox{1.25in}{with $x \in G$ uniformly \\ random and unknown.}
\label{eq:cosetstate}
\ee
Depending on context, it may be more convenient to view the outcome either as a random pure state, or equivalently, as the mixed quantum state
\be
  \rho_H := \frac{1}{|G|} \sum_{x \in G} |xH\>\<xH|
\label{eq:hsstate}
\ee
which we refer to as a \emph{hidden subgroup state}.  In the standard approach to the hidden subgroup problem, we attempt to determine $H$ using samples of this hidden subgroup state.

Historically, work on the hidden subgroup problem has focused almost exclusively on the standard method.  However, while this method seems quite natural, there is no general proof that it is necessarily the best way to approach the \HSP.  \textcite{KNP05} showed that the quantum query complexity of Simon's problem is linear, so that Simon's algorithm (using the standard method) is within a constant factor of optimal.  This immediately implies an $\Omega(n)$ lower bound for the \HSP in any group that contains the subgroup $\ZZ{2^n}$.  It would be interesting to prove similar results for more general groups, or to find other ways of evaluating the effectiveness of the standard method as compared with more general strategies.

\subsection{Weak Fourier sampling}\label{sec:wfsample}

The symmetry of the coset state \eq{cosetstate} (and equivalently, the hidden subgroup state \eq{hsstate}) can be exploited using the quantum Fourier transform.  In particular, we have 
\be
  |xH\> = \frac{1}{\sqrt{|H|}} \sum_{h \in H} R(h) |x\>
\ee
where $R$ is the right regular representation of $G$.  Thus the hidden subgroup state can be written
\ba
  \rho_H
  &= \frac{1}{|G|\cdot|H|} \sum_{x \in G} \sum_{h,h' \in H}
     R(h) |x\>\<x| R(h')^\dag \\
  &= \frac{1}{|G|\cdot|H|} \sum_{h,h' \in H} R(h h'^{-1}) \\
  &= \frac{1}{|G|} \sum_{h \in H} R(h).
\ea
Since the right regular representation is block-diagonal in the Fourier basis, the same is true of $\rho_H$.  In particular, using \eq{rightreghat}, we have
\ba
  \hat \rho_H
  &:= F_G \, \rho_H \, F_G^\dag \\
  &= \frac{1}{|G|} \bigoplus_{\sigma \in \hat G}
     \big( I_{d_\sigma} \otimes \sigma(H)^* \big)
\label{eq:hsstate_blockdiag}
\ea
where
\be
  \sigma(H) := \sum_{h \in H} \sigma(h).
\ee

Since $\hat\rho_H$ is block diagonal, with blocks labeled by irreducible representations, we may now measure the irrep label without loss of information.  This procedure is referred to as \emph{weak Fourier sampling}.  The probability of observing representation $\sigma \in \hat G$ under weak Fourier sampling is
\ba
  \Pr(\sigma)
  &= \frac{1}{|G|}
     \Tr \big( I_{d_\sigma} \otimes \sigma(H)^* \big) \\
  &= \frac{d_\sigma}{|G|} \sum_{h \in H} \chi_\sigma(h)^*,
\label{eq:wfsdist}
\ea
which is precisely $d_\sigma |H|/|G|$ times the number of times the trivial representation appears in $\Res^G_H\sigma$, the restriction of $\sigma$ to $H$ \cite[Theorem 1.2]{HRT03}.
We may now ask whether polynomially many samples from this distribution are sufficient to determine $H$, and if so, whether $H$ can be reconstructed from this information efficiently.

If $G$ is Abelian, then all of its representations are one-dimensional, so weak Fourier sampling reveals all of the available information about $\rho_H$.  This information can indeed be used to efficiently determine $H$, as discussed in \sec{finiteabelHSP}.

Weak Fourier sampling succeeds for a similar reason whenever $H$ is a \emph{normal subgroup} of $G$ (denoted $H \normalin G$), i.e., whenever $gHg^{-1} = H$ for all $g \in G$ \cite{HRT03}.
In this case, the hidden subgroup state within the irrep $\sigma \in \hat G$ is proportional to \begin{equation}
\sigma(H)^* = \frac{1}{|G|} \sum_{g \in G, h \in H} \sigma(g h g^{-1})^*,
\end{equation}
 and this commutes with $\sigma(x)^*$ for all $x \in G$, so by Schur's Lemma (\thm{schur}), it is a multiple of the identity.  Thus $\hat \rho_H$ is proportional to the identity within each block, and again weak Fourier sampling reveals all available information about $H$.
Indeed, the distribution under weak Fourier sampling is particularly simple: we have
\be
  \Pr(\sigma) = \begin{cases}
  d_\sigma^2 |H|/|G| & H \subseteq \ker \sigma \\
  0                  & \text{otherwise}
  \end{cases}
\label{eq:normalwfsdist}
\ee
(a straightforward generalization of the distribution seen in step \ref{item:abelhspmeasure} of \alg{abelhsp}), where $\ker\sigma := \{g \in G: \sigma(g)=1\}$ denotes the \emph{kernel} of the representation $\sigma$.
To see this, note that if $H \not\subseteq \ker\sigma$, then there is some $h' \in H$ with $\sigma(h') \ne 1$; but then $\sigma(h') \sigma(H) = \sum_{h \in H} \sigma(h' h) = \sigma(H)$, and since $\sigma(h')$ is unitary and $\sigma(H)$ is a scalar multiple of the identity, this can only be satisfied if in fact $\sigma(H) = 0$.  On the other hand, if $H \subseteq \ker\sigma$, then $\chi_\sigma(h) = d_\sigma$ for all $h \in H$, and the result is immediate.

To find $H$, we can simply proceed as in the Abelian case:
\begin{algorithm}[Finding a normal hidden subgroup] \label{alg:normalhsp}\ \\
\emph{Input:} Black box function hiding $H \normalin G$. \\
\emph{Problem:} Determine $H$.

\begin{enumerate}
\item Let $K_0 := G$.  For $t=1,\ldots,T$, where $T=O(\log|G|)$:
\begin{enumerate}
\item Perform weak Fourier sampling, obtaining an irrep $\sigma_t \in \hat G$.
\item Let $K_t := K_{t-1} \cap \ker\sigma_t$.
\end{enumerate}
\item Output $K_T$.
\end{enumerate}
\end{algorithm}

To see that this works, suppose that at the $t$th step, the intersection of the kernels is $K_{t-1} \normalin G$ with $K_{t-1} \ne H$ (so that, in particular, $|K_{t-1}| \ge 2|H|$); then the probability of obtaining an irrep $\sigma$ for which $K_{t-1} \subseteq \ker \sigma$ is (cf.\ step \ref{item:abelhspintersect} of \alg{abelhsp})
\be
  \frac{|H|}{|G|} \sum_{\sigma:\, K_{t-1} \subseteq \ker\sigma} d_\sigma^2
  = \frac{|H|}{|K_{t-1}|}
  \le \frac{1}{2}
\ee
where we have used the fact that the distribution \eq{normalwfsdist} remains normalized if $H$ is replaced by any normal subgroup of $G$.  Each repetition of weak Fourier sampling has a probability of at least $1/2$ of cutting the intersection of the kernels at least in half, so we converge to $H$ in $O(\log|G|)$ steps.
In fact, applying the same approach when $H$ is not necessarily normal in $G$ gives an algorithm to find the \emph{normal core} of $H$, the largest subgroup of $H$ that is normal in $G$ \cite{HRT03}.

This algorithm can be applied to find hidden subgroups in groups that are ``close to Abelian'' a certain sense.  In particular, \textcite{GSVV04} showed that if $\kappa(G)$, the intersection of the normalizers of all subgroups of $G$, is sufficiently large---specifically, if $|G|/|\kappa(G)| = 2^{O(\log^{1/2} n)}$, such as when $G=\ZZ{3} \semidirect \ZZ{2^n}$---then the \HSP in $G$ can be solved in polynomial time.  The idea is simply to apply the algorithm for normal subgroups to all subgroups containing $\kappa(G)$; the union of all subgroups obtained in this way gives the hidden subgroup with high probability.  This result was subsequently improved to give a polynomial-time quantum algorithm whenever $|G|/|\kappa(G)| = \poly(\log |G|)$ \cite{Gav04}.

\subsection{Strong Fourier sampling}\label{sec:sfsample}

Despite the examples given in the previous section, weak Fourier sampling does not provide sufficient information to recover the hidden subgroup in the majority of non-Abelian hidden subgroup problems.  For example, weak Fourier sampling fails to solve the \HSP in the symmetric group \cite{HRT03,GSVV04} and the dihedral group.

To obtain more information about the hidden subgroup, we can perform a measurement on the $d_\sigma^2$-dimensional state that results when weak Fourier sampling returns the outcome $\sigma$.  Such an approach is referred to as \emph{strong Fourier sampling}.
From \eq{hsstate_blockdiag}, this $d_\sigma^2$-dimensional state is the tensor product of a $d_\sigma$-dimensional maximally mixed state for the row register (as a consequence of the fact that the left and right regular representations commute) with some  $d_\sigma$-dimensional state $\hat\rho_{H,\sigma}$ for the column register.  Since the row register does not depend on $H$, we may discard this register without loss of information.  In other words, strong Fourier sampling is effectively faced with the state
\be
  \hat\rho_{H,\sigma} =
  \frac{\sigma(H)^*}{\sum_{h \in H} \chi_\sigma(h)^*}.
\ee

This state is proportional to a projector whose rank is simply the number of times the trivial representation appears in $\Res^G_H \sigma^*$.  This follows because
\be
  \sigma(H)^2 = \sum_{h,h' \in H} \sigma(h h') = |H| \, \sigma(H),
\ee
which gives
\be
  \hat\rho_{H,\sigma}^2
  = \frac{|H|}{\sum_{h \in H} \chi_\sigma(h)^*} \hat\rho_{H,\sigma},
\ee
so that $\hat\rho_{H,\sigma}$ is proportional to a projector with $\rank(\hat\rho_{H,\sigma}) = \sum_{h \in H} \chi_\sigma(h)^*/|H|$.

It is not immediately clear how to choose a good basis for strong Fourier sampling, so a natural first approach is to consider the effect of measuring in a random basis (i.e., a basis chosen uniformly with respect to the Haar measure over $\C^{d_\sigma}$).  There are a few cases in which such \emph{random strong Fourier sampling} is fruitful.  For example, \textcite{RRS05} showed that measuring in a random basis provides sufficient information to solve the \HSP in the Heisenberg group.  Subsequently, \textcite{Sen06} generalized this result to show that random strong Fourier sampling is information-theoretically sufficient whenever $\rank(\hat\rho_{H,\sigma}) = \poly(\log|G|)$ for all $\sigma \in \hat G$ (for example, when $G$ is the dihedral group), as a consequence of a more general result on the distinguishability of quantum states using random measurements.

However, in some cases random strong Fourier sampling is unhelpful.  For example, \textcite{GSVV04} showed that if $H$ is sufficiently small and $G$ is sufficiently non-Abelian (in a certain precise sense), then random strong Fourier sampling is not very informative.  In particular, they showed this for the problem of finding hidden involutions in the symmetric group.  Another example was provided by \textcite{MRRS07}, who showed that random strong Fourier sampling fails in the metacyclic groups $\Zp \semidirect \ZZ{q}$ (subgroups of the affine group $\Zp \semidirect \Zpx$) when $q<p^{1-\epsilon}$ for some $\epsilon > 0$.

Even when measuring in a random basis is information-theoretically sufficient, it does not give an efficient quantum algorithm; we must consider both the implementation of the measurement and the interpretation of its outcomes.  
We cannot efficiently measure in a random basis, but we can instead try to find explicit bases in which strong Fourier sampling can be performed efficiently, and for which the results solve the \HSP.  The first such algorithm was provided by \textcite{MRRS07}, for the metacyclic groups $\Zp \semidirect \ZZ{q}$ with $q=p/\poly(\log p)$.  Note that for these values of $p,q$, unlike the case $q<p^{1-\epsilon}$ mentioned above, measurement in a random basis is information-theoretically sufficient.  Indeed, we do not know of \emph{any} example of an \HSP for which strong Fourier sampling gives an efficient algorithm, yet random strong Fourier sampling fails information-theoretically; it would be interesting to find any such example (or to prove that none exists).

Of course, simply finding an informative basis is not sufficient; it is also important that the measurement results can be efficiently post-processed.  This issue arises not only in the context of measurement in a pseudo-random basis, but also in the context of certain explicit bases.  For example, \textcite{EH00} gave a basis for the dihedral \HSP in which a measurement gives sufficient classical information to infer the hidden subgroup, but no efficient means of post-processing this information is known (see \sec{abeliandihedral}).

For some groups, it turns out that strong Fourier sampling simply fails.  \textcite{MRS05} showed that, regardless of what basis is chosen, strong Fourier sampling provides insufficient information to solve the \HSP in the symmetric group.  Specifically, they showed that for any measurement basis (indeed, for any POVM on the hidden subgroup states), the distributions of outcomes in the cases where the hidden subgroup is trivial and where the hidden subgroup is a random involution are exponentially close.  

\subsection{Multi-register measurements and query complexity}\label{sec:multireg}

Even if we restrict our attention to the standard method, the failure of strong Fourier sampling does not necessarily mean that the \HSP cannot be solved.  In general, we need not restrict ourselves to measurements acting on a single hidden subgroup state $\rho_H$ at a time; rather, it may be advantageous to measure joint observables on $\rho_H^{\otimes k}$ for $k>1$.  Such an approach could conceivably be efficient provided $k=\poly(\log|G|)$.

By considering joint measurements of many hidden subgroup states at a time, \textcite{EHK99,EHK04} showed that the \emph{query complexity} of the \HSP is polynomial.  In other words, $\poly(\log |G|)$ queries of the black box function $f$ suffice to determine $H$.  Unfortunately, this does not necessarily mean that the (quantum) \emph{computational complexity} of the \HSP is polynomial, since it is not clear in general how to perform the quantum post-processing of $\rho_H^{\otimes \poly(\log |G|)}$ efficiently.  Nevertheless, this is an important observation since it already shows a difference between quantum and classical computation: recall that the classical query complexity of even the Abelian \HSP is typically exponential.  Furthermore, it offers some clues as to how we might design efficient algorithms.

To show that the query complexity of the \HSP is polynomial, it is sufficient to show that the (single-copy) hidden subgroup states are pairwise statistically distinguishable, as measured by the quantum fidelity
\ba
  F(\rho,\rho')
  &:= \Tr |\sqrt\rho \sqrt{\rho'}|. 
\ea
This follows from a result of \textcite{BK02}, who showed the following.

\begin{theorem}\label{thm:bk}
Suppose $\rho$ is drawn from an ensemble $\{\rho_1,\ldots,\rho_N\}$, where each $\rho_i$ occurs with some fixed prior probability $p_i$.  Then there exists a quantum measurement (the \emph{pretty good measurement}\footnote{To distinguish the states $\rho_i$ with prior probabilities $p_i$, the pretty good measurement (PGM) uses the measurement operators $E_i := p_i \varrho^{-1/2}\rho_i\varrho^{-1/2}$, where $\varrho:=\sum_i p_i\rho_i$ (see for example \textcite{HW94}).})
that identifies $\rho$ with probability at least
\begin{equation}
  1-N\sqrt{\max_{i \ne j}F(\rho_i,\rho_j)}.
\label{eq:bkbound}
\end{equation}
\end{theorem}
\noindent
In fact, by the minimax theorem, this holds even without assuming a prior distribution for the ensemble \cite{HW06}.

Given only one copy of the hidden subgroup state, \eq{bkbound} will typically give a trivial bound.  However, by taking multiple copies of the hidden subgroup states, we can ensure that the overall states are nearly orthogonal, and hence distinguishable.
In particular, since $F(\rho^{\otimes k},\rho'^{\otimes k}) = F(\rho,\rho')^k$, arbitrarily small error probability $\epsilon>0$ can be achieved using
\be
  k \ge \left\lceil \frac{2(\log N - \log \epsilon)}
                         {\log\left(1/\max_{i \ne j}F(\rho_i,\rho_j)\right)}
        \right\rceil
  \label{eq:bkcopies}
\ee
copies of $\rho$.

Provided that $G$ does not have too many subgroups, and that the fidelity between two distinct hidden subgroup states is not too close to $1$, this shows that polynomially many copies of $\rho_H$ suffice to solve the \HSP.  The total number of subgroups of $G$ is $2^{O(\log^2|G|)}$, which can be seen as follows.  Any group $K$ can be specified in terms of at most $\log_2 |K|$ generators, since every additional (non-redundant) generator increases the size of the group by at least a factor of $2$.  Since every subgroup of $G$ can be specified by a subset of at most $\log_2 |G|$ elements of $G$, the number of subgroups of $G$ is upper bounded by $|G|^{\log_2 |G|} = 2^{(\log_2|G|)^2}$.   Thus $k=\poly(\log |G|)$ copies of $\rho_H$ suffice to solve the \HSP provided the maximum fidelity is bounded away from $1$ by at least $1/\poly(\log |G|)$.

To upper bound the fidelity between two states $\rho,\rho'$, let $\Pi_\rho$ denote the projector onto the support of $\rho$. By considering the POVM with elements $\Pi_\rho,\I-\Pi_\rho$ and noting that the classical fidelity of the resulting distribution is an upper bound on the quantum fidelity, we have
\ba
  F(\rho,\rho')
  &\le 
       \sqrt{\Tr{\Pi_\rho \rho'}}.
\ea

Now consider the fidelity between $\rho_H$ and $\rho_{H'}$ for two distinct subgroups $H,H' \le G$.  Let $|H| \ge |H'|$ without loss of generality.  We can write \eq{hsstate} as
\be
  \rho_H = \frac{|H|}{|G|} \sum_{x \in T_H} |xH\>\<xH|
\label{eq:uniquecosets}
\ee
where $T_H$ is a left transversal of $H$ (i.e., a complete set of unique representatives for the left cosets of $H$ in $G$).  Since \eq{uniquecosets} is a spectral decomposition of $\rho_H$, we have
\ba
  \Pi_{\rho_H}
  = \sum_{x \in T_H} |xH\>\<xH|
  = \frac{1}{|H|} \sum_{x \in G} |xH\>\<xH|.
\ea
Then we have
\ba
  F(\rho_H,\rho_{H'})^2
  &\le \Tr \Pi_{\rho_H} \rho_{H'} \\
  &=   \frac{1}{|H|\cdot|G|} \sum_{x,x' \in G} |\<xH|x'H'\>|^2 \\
  &=   \frac{1}{|H|\cdot|G|} \sum_{x,x' \in G}
       \frac{|xH \cap x'H'|^2}{|H| \cdot |H'|} \label{eq:cosetintersect} \\
  &=   \frac{|H \cap H'|}{|H|} \\
  &\le \frac{1}{2},
\ea
where we have used the fact that
\be
  |xH \cap x'H'| =
  \begin{cases}
    |H \cap H'| & \text{if $x^{-1}x' \in HH'$} \\
    0           & \text{otherwise}
  \end{cases}
\ee
to evaluate
\ba
  \sum_{x,x' \in G} |xH \cap xH'|^2
  &= |G| \cdot |H \cap H'|^2 \cdot |HH'| \\
  &= |G| \cdot |H| \cdot |H'| \cdot |H \cap H'|.
\ea
This shows that $F(\rho_H,\rho_{H'}) \le 1/\sqrt{2}$, thereby establishing that the query complexity of the \HSP is $\poly(\log|G|)$.

It is possible to obtain tighter bounds on the number of hidden subgroup states needed to solve the \HSP.  For example, \textcite{BCD05a} showed that $(1+o(1)) \log_2 N$ hidden subgroup states are necessary and sufficient to find a hidden reflection in the dihedral group of order $2N$.  In a similar vein, \textcite{HKK08} gave asymptotically tight bounds on the number of hidden subgroup states needed to solve the \HSP in general groups, taking into account both the number of candidate subgroups and their sizes.

The measurements described in this section are highly multi-register: they observe correlated properties of all of $\poly(\log |G|)$ hidden subgroup states at once.  Thus they are quite far from strong Fourier sampling, in which measurements are made on only one hidden subgroup state at a time.  It is natural to ask whether some less entangled measurement might also be sufficient for general groups, perhaps measuring a smaller number of hidden subgroup states at a time, and adaptively using those measurement results to decide what measurements to make on successive hidden subgroup states.  However, \textcite{HMRRS06} have shown that this is not always the case: in the symmetric group (as well as a few other groups such as the general linear group), entangled measurements on $\Omega(\log |G|)$ registers at a time are required to solve the \HSP.

\subsection{The Kuperberg sieve}\label{sec:kuperberg}

In this section, we describe an approach developed by \textcite{Kup05} that gives a subexponential (though not polynomial) time algorithm for the dihedral hidden subgroup problem---specifically, it runs in time $2^{O(\sqrt{\log|G|})}$.

The dihedral group of order $2N$, denoted $D_N$, is the group of symmetries of a regular $N$-gon.  It has the presentation
\be
  D_N = \<r,s| r^2=s^N=1,~rsr=s^{-1}\>.
\ee
Here $r$ can be viewed as a reflection
about some fixed axis, and $s$ can be viewed as a rotation by an angle $2\pi/N$.

Using the defining relations, we can write any group element in the form $s^x r^a$ where  $x \in \ZN$ and $a \in \ZZ{2}$.  Thus we can equivalently think of the group as consisting of elements $(x,a) \in \ZN \times \ZZ{2}$.  Since
\ba
  (s^x r^a) (s^y r^b)
  &= s^x r^a s^y r^a r^{a+b} \\
  &= s^x s^{(-1)^a y} r^{a+b} \\
  &= s^{x+(-1)^a y} r^{a+b},
\ea
the group operation for such elements can be expressed as
\be
  (x,a)\cdot(y,b) = (x + (-1)^a y,a+b).
\ee
(In particular, this shows that the dihedral group is the semidirect product $\ZN \semidirect_\varphi \ZZ{2}$, where $\varphi: \ZZ{2} \to \Aut(\ZN)$ is defined by $\varphi(a)(y)=(-1)^a y$.)  It is also easy to see that the group inverse is
\be
  (x,a)^{-1} = (-(-1)^a x,a).
\ee

The subgroups of $D_N$ are either cyclic or dihedral.  The subgroups that are cyclic are of the form $\<(x,0)\>$ where $x \in \ZN$ is some divisor of $N$ (including $x=N$).  The subgroups that are dihedral are of the form $\<(x,0),(y,1)\>$ where $x \in \ZN$ is some divisor of $N$ and $y \in \ZZ{x}$; in particular, there are subgroups of the form $\<(y,1)\>$ where $y \in \ZN$.  A result of \textcite{EH00} reduces the general dihedral \HSP, in which the hidden subgroup could be any of these possibilities, to the dihedral \HSP with the promise that the hidden subgroup is of the form $\<(y,1)\> = \{(0,0),(y,1)\}$, i.e., a subgroup of order $2$ generated by the reflection $(y,1)$.\footnote{The basic idea of the Ettinger-H{\o}yer reduction is as follows.  Suppose that $f:D_N \to S$ hides a subgroup $H=\<(x,0),(y,1)\>$.  Then we can consider the function $f$ restricted to elements from the Abelian group $\ZN \times \{0\} \le D_N$.  This restricted function hides the subgroup $\<(x,0)\>$, and since the restricted group is Abelian, we can find $x$ efficiently using \alg{abelhsp}.  Now $\<(x,0)\> \normalin D_N$ (since $(z,a)(x,0)(z,a)^{-1} = (z+(-1)^a x,a)(-(-1)^a z,a) = ((-1)^a x,0) \in \ZN \times \{0\}$), so we can define the quotient group $D_N / \<(x,0)\>$.  But this is simply a dihedral group (of order $N/x$), and if we now define a function $f'$ as $f$ evaluated on some coset representative, it hides the subgroup $\<(y,1)\>$. \label{foot:ehreduction}}
Thus, from now on we will assume that the hidden subgroup is of the form $\<(y,1)\>$ for some $y \in \ZN$ without loss of generality.

When the hidden subgroup is $H=\<(y,1)\>$, one particular left transversal of $H$ in $G$ consists of the left coset representatives $(z,0)$ for all $z \in \ZN$.  The coset state \eq{cosetstate} corresponding to the coset $(z,0)H$ is
\be
  |(z,0)H\> =
  \frac{1}{\sqrt 2} (|z, 0\> + |y+z, 1\>).
\ee

We saw in \sec{wfsample} that to distinguish coset states in general, one should start with weak Fourier sampling: apply a Fourier transform over $G$ and then measure the irrep label.  Equivalently, we can simply Fourier transform the first register over $\ZN$, leaving the second register alone.  When the resulting measurement outcome $k$ is not $0$ or $N/2$, this procedure is effectively the same as performing weak Fourier sampling, obtaining a two-dimensional irrep labeled by either $k$ (for $k \in \{1,\ldots,\ceil{N/2}-1\}$) or $-k$ (for $k \in \{\floor{N/2}+1,\ldots,N-1\}$), with the uniformly random sign of $k$ corresponding to the maximally mixed row index, and the remaining qubit state corresponding to the column index.  For $k=0$ or $N/2$, the representation is reducible, corresponding to a pair of one-dimensional representations.

Fourier transforming the first register over $\ZN$, we obtain
\begin{align}
   &(F_{\ZN} \otimes I_2) |(z,0)H\> \nn
   &\quad= \frac{1}{\sqrt{2N}} \sum_{k \in \ZN} (\omega_N^{kz} |k,0\> + \omega_N^{k(y+z)} |k,1\>) \\
   &\quad= \frac{1}{\sqrt{N}} \sum_{k \in \ZN} \omega_N^{kz} |k\> 
   \otimes \frac{1}{\sqrt 2}(|0\> + \omega_N^{ky} |1\>).
\end{align}
If we then measure the first register, we obtain one of the $N$ values of $k$ uniformly at random, and we are left with the post-measurement state
\ba
  |\psi_k\> := \frac{1}{\sqrt 2} (|0\> + \omega_N^{yk} |1\>)
\label{eq:dihedralqubit}
\ea
(dropping an irrelevant global phase that depends on $z$).  Thus we are left with the problem of determining $y$ given the ability to produce single-qubit states $|\psi_k\>$ of this form (where $k$ is known).  Since this procedure is equivalent to dihedral weak Fourier sampling, there is no loss of information in processing the state to produce \eq{dihedralqubit}.

It would be useful if we could prepare states $|\psi_k\>$ with particular values of $k$.  For example, given the state $|\psi_{N/2}\> = \frac{1}{\sqrt 2}(|0\> + (-1)^y |1\>)$, we can learn the parity of $y$ (i.e., its least significant bit) by measuring in the basis of states $|\pm\> := (|0\> \pm |1\>)/\sqrt 2$.  The main idea of Kuperberg's algorithm is to combine states of the form \eq{dihedralqubit} to produce new states of the same form, but with more desirable values of $k$.

To combine states, we can use the following procedure.  Given two states $|\psi_p\>$ and $|\psi_q\>$, perform a controlled-not gate from the former to the latter, giving
\ba
  &|\psi_p,\psi_q\> \nn
  &\quad= \frac{1}{2}
    (|0,0\> + \omega_N^{yp}        |1,0\> 
            + \omega_N^{yq}     |0,1\>
            + \omega_N^{y(p+q)} |1,1\>) \\
  &\quad\mapsto \frac{1}{2}
    (|0,0\> + \omega_N^{yp}        |1,1\> 
            + \omega_N^{yq}     |0,1\>
            + \omega_N^{y(p+q)} |1,0\>) \\
  &\quad= \frac{1}{\sqrt 2}
    (               |\psi_{p+q},0\> 
    + \omega_N^{yq} |\psi_{p-q},1\>).
\ea
Then a measurement on the second qubit leaves the first qubit in the state $|\psi_{p \pm q}\>$ (up to an irrelevant global phase), with the $+$ sign occurring when the outcome is $0$ and the $-$ sign occurring when the outcome is $1$, each outcome occurring with probability $1/2$.

Note that this combination procedure can be viewed as implementing the \emph{Clebsch-Gordan decomposition}, the decomposition of a tensor product of representations into its irreducible constituents.  The state indices $p$ and $q$ can be interpreted as labels of irreps of $D_N$, and the extraction of $|\psi_{p \pm q}\>$ can be seen as transforming their tensor product (a reducible representation of $D_N$) into one of two irreducible components.

Now we are ready to describe the algorithm of \textcite{Kup05}.  For simplicity, we will assume from now on that $N=2^n$ is a power of $2$.  For such a dihedral group, it is actually sufficient to be able to determine the least significant bit of $y$, since such an algorithm could be used recursively to determine all the bits of $y$.\footnote{To see this, note that the group $D_N$ contains two subgroups isomorphic to $D_{N/2}$, namely $\{(2x,0),(2x,1): x \in \ZZ{(N/2)}\}$ and $\{(2x,0),(2x+1,1): x \in \ZZ{(N/2)}\}$.  The hidden subgroup is a subgroup of the former if $y$ has even parity, and of the latter if $y$ has odd parity.  Thus, once we learn the parity of $y$, we can restrict our attention to the appropriate $D_{N/2}$ subgroup.  The elements of either $D_{N/2}$ subgroup can be represented using only $n-1$ bits, and finding the least significant bit of the hidden reflection within this subgroup corresponds to finding the second least significant bit of $y$ in $D_N$.  Continuing in this way, we can learn all the bits of $y$ with only $n$ iterations of an algorithm for finding the least significant bit of the hidden reflection.}  Our strategy for doing this is to start with a large number of states, and collect them into pairs $|\psi_p\>,|\psi_q\>$ that share many of their least significant bits, such that $|\psi_{p-q}\>$ is likely to have many of its least significant bits equal to zero.  Trying to zero out all but the most significant bit in one shot would take exponentially long, so instead we proceed in stages, only trying to zero some of the least significant bits in each stage; this turns out to give an improvement.  (This approach is similar to previous classical sieve algorithms for learning \cite{BKW03} and lattice \cite{AKS01} problems, as well as a subsequent classical algorithm for average case instances of subset sum \cite{FP05}.)

\begin{algorithm}[Kuperberg sieve] \label{alg:kuperberg}\ \\
\emph{Input:} Black box function $f: D_{2^n} \to S$ hiding $\<(y,1)\> \le D_{2^n}$ for some $y \in \ZZ{2^n}$. \\
\emph{Problem:} Determine the least significant bit of $y$.

\begin{enumerate}
\item Prepare $\Theta(16^{\sqrt n})$ coset states of the form \eq{dihedralqubit}, where each copy has $k \in \ZZ{2^n}$ chosen independently and uniformly at random.
\item For each $j=0,1,\ldots,m-1$ where $m:=\ceil{\sqrt{n}\,}$, assume the current coset states have indices $k$ with at least $mj$ of the least significant bits equal to $0$.  Collect them into pairs $|\psi_p\>,|\psi_q\>$ that share at least $m$ of the next least significant bits, discarding any qubits that cannot be paired.  Create a state $|\psi_{p \pm q}\>$ from each pair, and discard it if the $+$ sign occurs.  Notice that the resulting states have at least $m(j+1)$ significant bits equal to $0$.
\item The remaining states are of the form $|\psi_0\>$ and $|\psi_{2^{n-1}}\>$.  Measure one of the latter states in the $|\pm\>$ basis to determine the least significant bit of $y$.
\end{enumerate}
\end{algorithm}

Since this algorithm requires $2^{O(\sqrt n)}$ initial queries and proceeds through $O(\sqrt n)$ stages, each of which takes at most $2^{O(\sqrt n)}$ steps,  the overall running time is $2^{O(\sqrt n)}$.

To show that the algorithm works, we need to prove that some qubits survive to the final stage of the process with non-negligible probability.  Let us analyze a more general version of the algorithm to see why we should try to zero out $\sqrt n$ bits at a time, starting with $2^{O(\sqrt n)}$ states.

Suppose we try to cancel $m$ bits in each stage, so that there are $n/m$ stages (not yet assuming any relationship between $m$ and $n$), starting with $2^\ell$ states.  Each combination operation succeeds with probability $1/2$, and turns two states into one, so at each step we retain only about $1/4$ of the states that can be paired.  Now when we pair states that allow us to cancel $m$ bits, there can be at most $2^m$ unpaired states, since that is the number of values of the $m$ bits to be canceled.  Thus if we ensure that there are at least $2 \cdot 2^m$ states at each stage, we expect to retain at least a $1/8$ fraction of the states for the next stage.  Since we begin with $2^\ell$ states,we expect to have at least $2^{\ell-3j}$ states left after the $j$th stage.  Thus, to have $2\cdot2^m$ states remaining at the last stage of the algorithm, we require $2^{\ell-3n/m} > 2^{m+1}$, or $\ell > m + 3n/m + 1$.  This is minimized by choosing $m \approx \sqrt{n}$, so $\ell \approx 4\sqrt{n}$ suffices.

This analysis is not quite correct because we do not obtain precisely a $1/8$ fraction of the paired states for use in the next stage.  For most of the stages, we have many more than $2 \cdot 2^m$ states, so nearly all of them can be paired, and the expected fraction remaining for the next stage is close to $1/4$.  Of course, the precise fraction will experience statistical fluctuations.  However, since we are working with a large number of states, the deviations from the expected values are very small, and a more careful analysis (using the Chernoff bound) shows that the procedure succeeds with high probability.  For a detailed argument, see \cite[Sec.\ 3.1]{Kup05}.  That paper also gives an improved algorithm that runs faster and that works for general $N$.

Note that this algorithm uses not only superpolynomial time, but also superpolynomial space, since all $2^{\Theta(\sqrt n)}$ coset states are present at the start.  However, by creating a smaller number of coset states at a time and combining them according to the solution of a subset sum problem, \textcite{Reg04b} showed how to make the space requirement polynomial in $n$ with only a slight increase in the running time.

Although Kuperberg's algorithm acts on pairs of coset states at a time, the overall algorithm effectively implements a highly entangled measurement on all $2^{\Theta(\sqrt n)}$ registers, since the procedure for producing $|\psi_{p \pm q}\>$ entangles the coset states $|\psi_p\>$ and $|\psi_q\>$.  The same is true of Regev's polynomial-space variant.

It is natural to ask whether a similar sieve could be applied to other {\HSP}s, such as in the symmetric group, for which highly entangled measurements are necessary.  \textcite{AMR07}
adapt Kuperberg's approach to give a subexponential-time algorithm for the \HSP in $G^n$, where $G$ is a fixed non-Abelian group.  (Note that the \HSP in $G^n$ can be much harder than solving $n$ instances of the \HSP in $G$, since $G^n$ has many subgroups that are not direct products of subgroups of $G$.)  Also, \textcite{Bac08} showed that an algorithm for the Heisenberg \HSP, similar to the one described in \sec{pgm} below, can be derived using the Clebsch-Gordan transform over the Heisenberg group.  It would be interesting to find further applications of the approach, especially ones that give new polynomial-time algorithms.

Unfortunately, this kind of sieve does not seem well-suited to the symmetric group.  In particular, \textcite{MRS07} gave the following negative result for the \HSP in $S_n \wr S_2$, where the hidden subgroup is promised to be either trivial or an involution.  Consider any algorithm that works by combining pairs of hidden subgroup states to produce a new state in their Clebsch-Gordan decomposition, and uses the sequence of measurement results to guess whether the hidden subgroup is trivial or nontrivial.  Any such algorithm must use $2^{\Omega(\sqrt n)}$ queries.  Note that this lower bound is only slightly smaller than the best known classical algorithm for graph isomorphism, as mentioned in \sec{hspandapps}.

\subsection{Pretty good measurement}\label{sec:pgm}

Another recent technique for the \HSP is based on implementing the \emph{pretty good measurement} (PGM) on the hidden subgroup states.  Recall from \sec{multireg} that for any group $G$, the PGM applied to $\poly(\log |G|)$ copies of $\rho_H$ identifies $H$ with high probability.  Thus if we can efficiently implement the PGM on sufficiently many copies, we will have found an efficient algorithm for the \HSP.

This approach was considered in \cite{BCD05a,BCD05b} for certain semidirect product groups $A \semidirect \Zp$, where $A$ is an Abelian group and $p$ is prime.  For these groups, the general \HSP can be reduced to the \HSP assuming that the hidden subgroup is chosen from a certain subset.  Furthermore, the PGM turns out to be the optimal measurement for distinguishing the resulting hidden subgroup states, in the sense that it maximizes the probability of correctly identifying the hidden subgroup assuming a uniform distribution over the subgroups under consideration (as can be proven using the characterization of optimal measurement by \textcite{Hol73} and \textcite{YKL75}).
This generalizes the result of \cite{Ip03} that Shor's algorithm implements the optimal measurement for the Abelian \HSP, and suggests that in general, optimal measurements may be good candidates for efficient quantum algorithms.

For general groups of the form $A \semidirect \Zp$, the PGM approach reveals a connection between the original hidden subgroup problem and a related average-case algebraic problem.  Specifically, the PGM succeeds in distinguishing the hidden subgroup states exactly when the average case problem is likely to have solutions, and the PGM can be implemented efficiently by giving an efficient algorithm for solving the average case problem (or more precisely, for \emph{approximately quantum sampling} from the set of solutions to the problem).  Different {\HSP}s correspond to different average case problems, of varying difficulty.  For example, the dihedral \HSP corresponds to the average case subset sum problem \cite{BCD05a}, which appears to be hard.  But other average case problems appearing in the approach are easier, leading to efficient algorithms.  Certain instances of the Abelian \HSP give rise to systems of linear equations.  For the metacyclic {\HSP}s solved in \cite{MRRS07} (and indeed for some additional cases), the average-case problem is a discrete log problem, which can be solved using Shor's algorithm as described in \sec{dlog}.  And for the \HSP in the Heisenberg group\footnote{The Heisenberg group is an example of an \emph{extraspecial group}.  \textcite{ISS07} give an efficient quantum algorithm for the \HSP in any extraspecial group (see \sec{orbitcoset} for more details).  This subsequent algorithm also makes use of the solution of a system of polynomial equations to implement an entangled measurement.} $(\Zp)^2 \semidirect \Zp$, and more generally in any semidirect product $(\Zp)^r \semidirect \Zp$, the average case problem is a problem of solving polynomial equations, which can be done efficiently using Gr\"obner basis techniques provided $r=O(1)$ \cite{BCD05b}.

Here we briefly summarize the algorithm that results from applying the PGM to the \HSP in the Heisenberg group, since this case exemplifies the general approach.
The Heisenberg group can be viewed as the semidirect product $(\Zp)^2 \semidirect_\varphi \Zp$, where $\varphi:\Zp \to \Aut((\Zp)^2)$ is defined by $\varphi(c)(a,b)=(a+bc,b)$.
Equivalently, it is the group of lower triangular $3\times3$ matrices
\be
  \left\{
  \begin{pmatrix}
  1 & 0 & 0 \\
  b & 1 & 0 \\
  a & c & 1
  \end{pmatrix}
  : a,b,c \in \F_p
  \right\}
\ee
over $\F_p$,
or alternatively, the group generated by \emph{generalized Pauli operators} $X,Z \in \C^{p \times p}$ satisfying $X|x\>=|x+1 \bmod p\>$ and $Z|x\> = \rou{p}^x |x\>$, with elements $\rou{p}^a X^b Z^c$.
With any of these descriptions, the group elements are of the form $(a,b,c)$ with $a,b,c \in \Zp$, and the group law is
\be
  (a,b,c) \cdot (a',b',c') = (a+a'+b'c,b+b',c+c').
\ee

Just as the dihedral \HSP can be reduced to the problem of finding a hidden reflection (Footnote~\ref{foot:ehreduction}),
one can show that to solve the general \HSP in the Heisenberg group, it is sufficient to be able to distinguish the following cyclic subgroups of order $p$:
\be
  H_{a,b} :=
  \<(a,b,1)\> = \{(a,b,1)^j: j \in \Zp\},
\ee
where $a,b \in \Zp$.  A simple calculation shows that
\be
  (a,b,1)^x = (xa+\tbinom{x}{2}b,xb,x).
\ee
Furthermore, the cosets of any such subgroup can be represented by the $p^2$ elements $(\ell,m,0)$ for $\ell,m \in (\Zp)^2$.  Thus the coset state \eq{cosetstate} can be written
\be
  |(\ell,m,0)H_{a,b}\> =
  \frac{1}{\sqrt p} \sum_{x \in \Zp} |xa+\tbinom{x}{2}b+\ell,xb+m,j\>.
  \label{eq:heiscoset}
\ee
Our goal is to determine the parameters $a,b \in \Zp$ using copies of this state with $\ell,m \in \Zp$ occurring uniformly at random.

At this point, we could perform weak Fourier sampling over the Heisenberg group without discarding any information.  However, as for the dihedral group (\sec{kuperberg}), it is simpler to consider an Abelian Fourier transform instead of the full non-Abelian Fourier transform.  Using the representation theory of the Heisenberg group (see for example \textcite[Chap.\ 18]{Ter99}), one can show that this procedure is essentially equivalent to non-Abelian Fourier sampling.

Fourier transforming the first two registers over $(\Zp)^2$, we obtain the state
\be
  \frac{1}{p^{3/2}} \sum_{x,s,t \in \Zp}
  \omega_p^{s(\ell+xa+\binom{x}{2}b)+t(m+xb)}|s,t,x\>.
\ee
Now suppose we measure the values $s,t$ appearing in the first two registers.  In fact this can be done without loss of information, since the density matrix of the state (mixed over the uniformly random values of $\ell,m$) is block diagonal, with blocks labeled by $s,t$.
Collecting the coefficients of the unknown parameters $a,b$, the resulting $p$-dimensional quantum state is
\ba
  |\widehat{H_{a,b;s,t}}\>
  &:=  \frac{1}{\sqrt p} \sum_{x \in \Zp}
      \omega_p^{a(sx)+b(s\binom{x}{2}+tx)}|x\>
\ea
where the values $s,t \in \Zp$ are known, and are obtained uniformly at random.  We would like to use samples of this state to determine $a,b \in \Zp$.

With only one copy of this state, there is insufficient information to recover the hidden subgroup: Holevo's theorem (see for example \textcite[Sec.\ 12.1]{NC00}) guarantees that a measurement on a $p$-dimensional quantum state can reliably communicate at most $p$ different outcomes, yet there are $p^2$ possible values of $(a,b) \in (\Zp)^2$.  Thus we must use at least two copies.

However, by making a joint measurement on two copies of the state, we can recover the information about $a,b$ that is encoded in a quadratic function in the phase.  To see this, consider the state
\ba
  |\widehat{H_{a,b;s,t}}\> \otimes |\widehat{H_{a,b;u,v}}\> 
  &= \frac{1}{p} \sum_{x,y \in \Zp}
     \omega_p^{\alpha a+\beta b}|x,y\>,
\label{eq:heistwo}
\ea
where
\ba
  \alpha &:= sx+uy \label{eq:heisw} \\
  \beta  &:= s\tbinom{x}{2}+tx+u\tbinom{y}{2}+vy \label{eq:heisv}
\ea
and where we suppress the dependence of $\alpha,\beta$ on $s,t,u,v,x,y$ for clarity.
If we could replace $|x,y\>$ by $|\alpha,\beta\>$, then the resulting state would be simply the Fourier transform of $|a,b\>$, and an inverse Fourier transform would reveal the solution.  To work toward this situation we compute the values of $\alpha,\beta$ in ancilla registers, giving the state
\be
  \frac{1}{p} \sum_{x,y \in \Zp}
  \omega_p^{\alpha a + \beta b}
  |x,y,\alpha,\beta\>,
\label{eq:beforeqsample}
\ee
and attempt to uncompute the first two registers.

For fixed values of $\alpha,\beta,s,t,u,v \in \Zp$, the quadratic equations \eqs{heisw}{heisv} could have zero, one, or two solutions $x,y \in \Zp$.  Thus we cannot hope to erase the first and second registers by a classical procedure conditioned on the values in the third and fourth registers (and the known values of $s,t,u,v$).  However, it is possible to implement a quantum procedure to erase the first two registers by considering the full set of solutions
\ba
  \begin{aligned}
  S^{s,t,u,v}_{\alpha,\beta} := \{&(x,y) \in (\Zp)^2:~ \\
  &sx+uy=\alpha \text{~and~} \\
  &s\tbinom{x}{2}+tx+u\tbinom{y}{2}+vy=\beta\}.
  \end{aligned}
\ea
The state \eq{beforeqsample} can be rewritten
\be
  \frac{1}{p} \sum_{x,y \in \Zp}
  \omega_p^{\alpha a + \beta b}
  \sqrt{|S^{s,t,u,v}_{\alpha,\beta}|} \, |S^{s,t,u,v}_{\alpha,\beta},\alpha,\beta\>.
\ee
Thus, if we could perform a unitary transformation satisfying
\be
  |S^{s,t,u,v}_{\alpha,\beta}\> \mapsto |\alpha,\beta\>
  \text{~for~} |S^{s,t,u,v}_{\alpha,\beta}| \ne 0
\label{eq:heisqsample}
\ee
(and defined in any way consistent with unitarity for other values of $\alpha,\beta$), we could erase the first two registers of \eq{beforeqsample},\footnote{Note that we can simply apply the transformation \eq{heisqsample} directly to the state \eq{heistwo}; there is no need to explicitly compute the values $\alpha,\beta$ in an ancilla register.} producing the state
\be
  \frac{1}{p} \sum_{\alpha,\beta \in \Zp}
  \omega_p^{\alpha a+\beta b}
  \sqrt{|S^{s,t,u,v}_{\alpha,\beta}|} \, |\alpha,\beta\>.
\label{eq:afterqsample}
\ee
The inverse of the transformation \eq{heisqsample} is called \emph{quantum sampling} because it produces a uniform superposition over the set of solutions, a natural quantum analog of \emph{random sampling} from the solutions.

Since the system of \eqs{heisw}{heisv} consists of a pair of quadratic equations in two variables over $\Fp$, it has either zero, one, or two solutions $x,y \in \Fp$.  For about half the cases, there are zero solutions; for about half the cases, there are two solutions; and for a vanishing fraction of the cases, there is only one solution.  More explicitly, by a straightforward calculation, the solutions can be expressed in closed form as
\ba
  x &= \frac{\alpha s+sv-tu\pm\sqrt\Delta}{s(s+u)} \\
  y &= \frac{\alpha u+tu-sv\mp\sqrt\Delta}{u(s+u)}
\ea
where
\be
  \Delta := (2\beta s+\alpha s-\alpha^2-2\alpha t)(s+u)u + (\alpha u+tu-sv)^2.
\ee
Provided $s u (s+u) \ne 0$, the number of solutions is completely determined by the value of $\Delta$.  If $\Delta$ is a nonzero square in $\Fp$, then there are two distinct solutions; if $\Delta=0$ then there is only one solution; and if $\Delta$ is a non-square then there are no solutions.  In any event, since we can efficiently compute an explicit list of solutions in each of these cases, we can efficiently perform the transformation \eq{heisqsample}.

It remains to show that the state \eq{afterqsample} can be used to recover $a,b$.  This state is close to the Fourier transform of $|a,b\>$ provided the solutions are nearly uniformly distributed.  Since the values of $s,t,u,v$ are uniformly distributed over $\Fp$, it is easy to see that $\Delta$ is uniformly distributed over $\Fp$.  This means that $\Delta$ is a square about half the time, and is a non-square about half the time (with $\Delta=0$ occurring only with probability $1/p$).  Thus there are two solutions about half the time and no solutions about half the time.  This distribution of solutions is uniform enough for the procedure to work.

Applying the inverse quantum Fourier transform over $\Zp \times \Zp$, we obtain the state
\be
  \frac{1}{p^2} \sum_{\alpha,\beta,k,\ell \in \Zp}
  \omega_p^{\alpha (a-k) + \beta (b-\ell)} \sqrt{|S^{s,t,u,v}_{\alpha,\beta}|} \, |k,\ell\>. \label{eq:heisfinalstate}
\ee
Measuring this state, the probability of obtaining the outcome $k=a$ and $\ell=b$ for any particular values of $s,t,u,v$ is
\be
  \frac{1}{p^4} \left( \sum_{\alpha,\beta \in \Zp} \sqrt{|S^{s,t,u,v}_{\alpha,\beta}|} \right)^2.
\ee
Since those values occur uniformly at random, the overall success probability of the algorithm is
\ba
  &\frac{1}{p^8} \sum_{s,t,u,v \in \Zp} \left( \sum_{\alpha,\beta \in \Zp} \sqrt{|S^{s,t,u,v}_{\alpha,\beta}|} \right)^2 \nn \quad
  &\ge \frac{1}{p^{12}} \left( \sum_{s,t,u,v \in \Zp} \sum_{\alpha,\beta \in \Zp} \sqrt{|S^{s,t,u,v}_{\alpha,\beta}|} \right)^2 \\ \quad
  &\ge \frac{1}{p^{12}} \left( \sum_{\alpha,\beta \in \Zp} \frac{p^4}{2+o(1)} \sqrt{2} \right)^2 
  = \frac{1}{2} - o(1),
\ea
which shows that the algorithm succeeds with probability close to $1/2$.

In summary, the efficient quantum algorithm for the \HSP in the Heisenberg group is as follows:
\begin{algorithm}[Heisenberg \HSP] \label{alg:heisenberg}\ \\
\emph{Input:} Black box function hiding $H_{a,b}$. \\
\emph{Problem:} Determine the parameters $a,b$.

\begin{enumerate}
\item Prepare two coset states, as in \eq{heiscoset}.
\item Perform the \QFT $F_{\Zp \times \Zp}$ on the first two registers of each coset state and measure those registers in the computational basis, giving \eq{heistwo}.
\item Perform the inverse quantum sampling transformation \eq{heisqsample}, giving \eq{afterqsample}.
\item Perform the inverse \QFT $F_{\Zp \times \Zp}^\dag$, giving \eq{heisfinalstate}.
\item Measure the resulting state in the computational basis, giving $(a,b)$ with probability $1/2 - o(1)$.
\end{enumerate}
\end{algorithm}

Because the transformation \eq{heisqsample} acts jointly on the two registers, the algorithm described above effectively makes an entangled measurement on two copies of the hidden subgroup state.  However, we do not know whether this is the only way to give an efficient algorithm for the \HSP in the Heisenberg group.  In particular, recall from \sec{sfsample} that Fourier sampling in a random basis provides sufficient information to reconstruct the hidden subgroup \cite{RRS05}.  It would be interesting to know whether there is an \emph{efficient} quantum algorithm using only the statistics of single-register measurements, or if no such algorithm exists.  It would also be interesting to find any group for which Fourier sampling does not suffice, even information-theoretically, but for which there is an efficient quantum algorithm based on multi-register measurements.

The PGM approach outlined above can also be applied to certain state distinguishability problems that do not arise from {\HSP}s.  In particular, it can be applied to the generalized Abelian hidden shift problem discussed in \sec{hiddenshift} (for which the average case problem is an integer program) \cite{CD05} and to hidden polynomial problems of the form \eq{specialhiddenpoly}, as discussed in \sec{nonlinear} (for which the average case problem is again a system of polynomial equations) \cite{DDW07}.

\section{Hidden Shift Problem}
\label{sec:hiddenshift}

The \emph{hidden shift problem} (also known as the \emph{hidden translation problem}) is a natural variant of the hidden subgroup problem.  Its study has shed light on (and indeed, led to new algorithms for) the \HSP. Furthermore, the hidden shift problem has applications that are of interest in their own right.

In the hidden shift problem, we are given two injective functions $f_0: G \to S$ and $f_1: G \to S$, with the promise that
\be
  f_0(g)=f_1(sg) \text{~for some~} s \in G.
\ee
The goal of the problem is to find $s$, the \emph{hidden shift}.  In the non-Abelian hidden shift problem, as in the non-Abelian \HSP, there is an arbitrary choice of left or right multiplication; here we again make the choice of left multiplication.

When $G$ is Abelian, this problem is equivalent to the \HSP in $G \semidirect_\varphi \ZZ{2}$ (sometimes called the $G$-dihedral group), where the homomorphism $\varphi:\ZZ{2} \to \Aut(G)$ is defined by $\varphi(0)(x)=x$ and $\varphi(1)(x)=x^{-1}$. In particular, the hidden shift problem in $\ZZ{N}$ is equivalent to the dihedral \HSP. To see this, consider the function $f:G \semidirect \ZZ{2} \to S$ defined by $f(x,b):=f_b(x)$. This function hides the involution $\<(s,1)\>$, so a solution of the \HSP gives a solution of the hidden shift problem. Conversely, solving the \HSP in $G \semidirect \ZZ{2}$ with the promise that $H$ is an involution is sufficient to solve the \HSP in general (Footnote~\ref{foot:ehreduction}), so a solution of the hidden shift problem gives a solution of the \HSP.
While no polynomial time quantum algorithm is known for the general Abelian hidden shift problem, Kuperberg's sieve (\alg{kuperberg}) solves the problem in time $2^{O(\sqrt{\log |G|})}$, whereas a brute force approach takes $2^{\Omega(\log |G|)}$ steps. 

When $G$ is non-Abelian, the inversion map $x \mapsto x^{-1}$ is not a group automorphism, so we cannot even define a group $G \semidirect_\varphi \ZZ{2}$. However, the hidden shift problem in $G$ is closely connected to an \HSP, namely in the wreath product group $G \wr \ZZ{2} = (G \times G) \semidirect_{\tilde\varphi} \ZZ{2}$, where $\tilde\varphi(0)(x,y)=(x,y)$ and $\tilde\varphi(1)(x,y)=(y,x)$. The hidden shift problem in $G$ reduces to the \HSP in $G \wr \ZZ{2}$ with the hidden subgroup $\<(s,s^{-1},1)\>$. Furthermore, the \HSP in $G \wr \ZZ{2}$ with hidden subgroups of this form reduces to the hidden shift problem in $G \times G$.  Thus, for families of groups in which $G \times G$ is contained in a larger group $G'$ from the same family---such as for the symmetric group, where $S_n \times S_n \le S_{2n}$---the hidden shift and hidden subgroup problems are essentially equivalent \cite{HMRRS06}.
Moreover, by a similar argument to the one in \sec{multireg}, the quantum query complexity of the hidden shift problem in $G$ is $\poly(\log|G|)$ even when $G$ is non-Abelian.

Testing isomorphism of rigid graphs can be cast as a hidden shift problem in the symmetric group. If we let $f(\pi,0)=\pi(\Gamma)$ and $f(\pi,1)=\pi(\Gamma')$, then the hidden shift is $\sigma$, where $\Gamma=\sigma(\Gamma')$. Despite the equivalence between hidden shift and hidden subgroup problems, the hidden shift problem in $S_n$ is arguably a more natural setting for rigid graph isomorphism than the \HSP, since every possible hidden shift corresponds to a possible isomorphism between graphs, whereas the \HSP must be restricted to certain subgroups \cite{CW05}.

In this section we describe quantum algorithms for various hidden shift problems.  We begin by presenting a single-register measurement for the cyclic hidden shift problem (i.e., the dihedral \HSP) that provides sufficient information to encode the hidden shift.  While no efficient way of postprocessing this information is known, we explain how a similar approach leads to an efficient quantum algorithm for the hidden shift problem over $(\Zp)^n$ with $p$ a fixed prime. Since both of these problems are Abelian hidden shift problems, they could equally well be viewed as {\HSP}s, but we discuss them here because the latter is an important ingredient of the \emph{orbit coset} approach, which uses self-reducibility of a quantum version of the hidden shift problem to give efficient quantum algorithms for certain hidden subgroup and hidden shift problems. Then we describe an algorithm for the shifted Legendre symbol problem, a non-injective variant of the dihedral \HSP that can be solved efficiently, and that also leads to an efficient quantum algorithm for estimating Gauss sums.  Finally, we describe a generalization of the hidden shift problem that interpolates to an Abelian \HSP, and that can be solved efficiently in some cases even when the original hidden shift problem cannot.

\subsection{Abelian Fourier sampling for the dihedral \HSP}
\label{sec:abeliandihedral}

Consider the \HSP in the dihedral group $\ZZ{N} \semidirect \ZZ{2}$ with hidden subgroup $\<(s,1)\>$---or equivalently, the hidden shift problem in the cyclic group $\ZZ{N}$ with hidden shift $s$.  Recall from \sec{kuperberg} (specifically, \eq{dihedralqubit}) that the standard method, followed by a measurement of the first register in the Fourier basis (over $\ZZ{N}$), produces the state $\frac{1}{\sqrt 2}(|0\> + \rouN^{sk}|1\>)$ for some uniformly random measurement outcome $k \in \ZZ{N}$. Now suppose we measure this qubit in the basis of states $|\pm\> := \frac{1}{\sqrt 2}(|0\> \pm |1\>)$ (i.e., the Fourier basis over $\ZZ{2}$); then the outcome `$+$' occurs with probability $\cos^2(\frac{\pi s k}{N})$. Thus, if keep only those measured values of $k$ for which the outcome of the second measurement is `$+$', we effectively sample from a distribution over $k \in \ZZ{N}$ with $\Pr(k)=2\cos^2(\frac{\pi s k}{N})/N$.

This procedure was proposed by \textcite{EH00}, who showed that $O(\log N)$ samples of the resulting distribution provide sufficient information to determine $k$ with high probability. This single-register measurement is a much simpler procedure than either the Kuperberg sieve \cite{Kup05} or the optimal measurement described in \cite{BCD05a}, both of which correspond to highly entangled measurements. However, we are left with the problem of post-processing the measurement results to infer the value of $s$, for which no efficient procedure is known.

\subsection{Finding hidden shifts in \texorpdfstring{$(\Zp)^n$}{(Z/pZ)\textasciicircum n}}
\label{sec:zpnhiddenshift}

A similar approach can be applied to the hidden shift problem in the
elementary Abelian $p$-group $(\Zp)^n$ with $p$ a fixed prime, but in this
case the postprocessing can be carried out efficiently. This result
is an important building block in an efficient quantum algorithm for
the hidden shift and hidden subgroup problems in certain families of
solvable groups \cite{FIMSS03}, as discussed in the next section.

Consider the hidden shift problem in $(\Zp)^n$ with hidden shift $s$.
Applying the standard method, we obtain the hidden shift state
\begin{equation} 
 \frac{1}{\sqrt 2}(|z,0\>+|z+s,1\>) 
\end{equation} 
for some unknown $z \in (\Zp)^n$ chosen uniformly at random.  Now
suppose that, as in the measurement for the dihedral group described
above, we perform Abelian Fourier sampling on this state.  In other
words, we Fourier transform the first register over $(\Zp)^n$ and the
second over $\ZZ{2}$; this gives
\begin{equation} 
\frac{1}{2\sqrt{p^n}} \sum_{y \in (\Zp)^n} \sum_{b \in \ZZ{2}}
[\rou{p}^{y \cdot z} + \rou{p}^{y \cdot(z+s)}(-1)^b] |y,b\> \,.  
\end{equation}
Finally, suppose we measure this state in the computational basis.  A straightforward calculation shows that we obtain the outcome $(y,0)$ with probability $\cos^2(\frac{\pi y \cdot s}{p})/p^n$ and the outcome $(y,1)$ with probability $\sin^2(\frac{\pi y \cdot s}{p})/p^n$.  Thus, conditioned on observing $1$ in the second register, we see $y$ in the first register with probability
\begin{equation} 
\Pr(y) = \frac{2}{p^n} \sin^2\Big(\frac{\pi y \cdot s}{p}\Big).
\label{eq:zpn_sample}
\end{equation} 
In particular, notice that there is zero probability of seeing any $y \in (\Zp)^n$ such that $y \cdot s = 0 \bmod{p}$: we see only points that are \emph{not} orthogonal to the hidden shift.  (This may be contrasted with the \HSP in $(\Zp)^n$ with hidden subgroup $\<s\>$, in which Fourier sampling only gives points $x \in (\Zp)^n$ with $x \cdot s = 0$.)

We now argue that $O(n)$ samples from this distribution are information-theoretically sufficient to determine the hidden shift $s$.  Since we only observe points $y$ that are not orthogonal to $s$, the observation of $y$ allows us to eliminate the hyperplane $y \cdot s=0$ of possible values of $s$.  With enough samples, we can eliminate all possible candidate values of $s$ except the true value (and scalar multiples thereof).

For simplicity, suppose we sample uniformly from all $y \in (\Zp)^n$ satisfying $y \cdot s \ne 0$ for the unknown $s$.  While the true distribution \eq{zpn_sample} is not uniform, it is not far from uniform, so the argument given here can easily be modified to work for the true distribution.
Consider some fixed candidate value $s'$ with $s' \ne \alpha s$ for any $\alpha \in \ZZ{p}$.  If $y$ were sampled uniformly at random, then $s'$ would be eliminated with probability $1/p$.  Sampling uniformly from the subset of points $y$ satisfying $y \cdot s \ne 0$ only raises the probability of eliminating $s'$, so a randomly sampled $y$ eliminates $s'$ with probability at least $1/p$.  Thus after $O(n)$ samples, the probability of not eliminating $s'$ is exponentially small, and by a union bound, the probability of any such $s'$ not being eliminated is upper bounded by a constant.

Unfortunately, given $k=\Theta(n)$ samples $y_1,\ldots,y_k$, we do not know how to \emph{efficiently} determine $s$.  We would like to solve the system of inequations $y_1 \cdot s \ne 0, \ldots, y_k \cdot s \ne 0$ for $s \in (\Zp)^n$.  Using Fermat's little theorem, which says that $a^{p-1}=1$ for any $a \in \ZZ{p}$ with $a \ne 0$, we can rewrite these inequations as a system of polynomial equations $(y_1 \cdot s)^{p-1} = \cdots = (y_k \cdot s)^{p-1} = 1$.  However, the problem of solving polynomial equations over a finite field is \NP-hard, so we cannot hope to solve for $s$ quickly using generic methods.

This problem is circumvented in \cite{FIMSS03,Iva08} using the idea of linearization. If we treat each product of $p-1$ components of $s \in (\Zp)^n$ as a separate variable, then we can view $(y \cdot s)^{p-1} = 1$ as a linear equation over a vector space of dimension $\binom{n+p-2}{p-1}$ (the number of ways of choosing $p-1$ items from $n$ items, with replacement and without regard for ordering). Since this method treats variables as independent that are in fact highly dependent, it requires more samples to obtain a unique solution. Nevertheless, \textcite{FIMSS03} show that $O(n^{p-1})$ samples suffice. Since this method only involves linear equations, and the number of equations remains $\poly(n)$ (recall the assumption that $p=O(1)$), the resulting algorithm is efficient.

A similar approach works for the hidden shift problem in $(\ZZ{p^k})^n$, where $p^k$ is any fixed prime power \cite{FIMSS03,Iva08}.  However, no efficient algorithm is known for the case of $(\ZZ{m})^n$ with $m$ not a prime power, even in the smallest case, $m=6$.

\subsection{Self-reducibility, quantum hiding, and the orbit coset problem}\label{sec:orbitcoset}

By combining the result of the previous section with a \emph{self-reducible} variant of the hidden shift problem, \textcite{FIMSS03} also give an efficient quantum algorithm for the \HSP and hidden shift problem in a large family of solvable groups.  The idea of self-reducibility is as follows.  Suppose we could reduce the \HSP in $G$ to the \HSP in subgroups of $G$, and apply such a reduction recursively until the remaining groups are either simple enough that the \HSP can be solved by some known method, or small enough that it can be solved by brute force.  For example, it would be useful if we could reduce the \HSP in $G$ to the \HSP in $N$ and $G/N$, where $N \pnormalin G$ is a proper normal subgroup of $G$.  No approach of this kind has been directly applied to the \HSP or the hidden shift problem, but this self-reducibility concept has proved fruitful for a quantum generalization of the hidden shift problem called the orbit coset problem.

Recall that in the standard method for the \HSP, we prepare the
uniform superposition $|G\>$, query a black-box function $f:G \to S$
satisfying \eq{hsphides}, and discard the resulting function value,
producing a uniformly random coset state $|xH\>$.  More generally,
suppose we have some black-box isometry $F$ satisfying
\begin{equation}
  F|x\> = |x\> \otimes |\phi_x\>
\label{eq:qhidingf}
\end{equation}
for some set of quantum states $\{|\phi_x\>: x \in G\}$ satisfying
\begin{equation}
  \<\phi_x|\phi_y\> = 
  \begin{cases}
    1 & x^{-1}y \in H \\
    0 & \text{otherwise}.
  \end{cases}
\end{equation}
By analogy to \eq{hsphides}, we say that $F$ is a \emph{quantum hiding
  function} for $H$ in $G$.  Querying the quantum black box $F$ on the
uniform superposition $|G\>$ and discarding the second register has
the same effect as the standard method: the result is a uniformly
random coset state $|xH\>$.  But the possibility of using quantum
superpositions for the states $|\phi_x\>$ offers more freedom when
constructing reductions.

One way to produce quantum hiding states $\{|\phi_x\>: x \in G\}$ is
as follows.  Let $\Phi$ be an orthonormal set of quantum states, and
let $\alpha:G \times \Phi \to \Phi$ be a (left) action of $G$ on
$\Phi$.  For some fixed $|\phi\> \in \Phi$, define $|\phi_x\> :=
\alpha(x)(|\phi\>)$.  Then the isometry \eq{qhidingf} is a quantum
hiding function for the \emph{stabilizer} of $|\phi\>$, the subgroup
$\stab(|\phi\>):=\{x \in G: \alpha(x)(|\phi\>) = |\phi\>\} \le G$.
Fixing $G$, $\Phi$, and $\alpha$, the \emph{stabilizer
  problem}\footnote{\textcite{Kit95} gave an efficient algorithm for
  the stabilizer problem in the case where $G$ is Abelian and the
  hiding function is classical, prefiguring the hidden
  subgroup framework.} asks us to find a generating set for
$\stab(|\phi\>)$ given (some number of copies of) the state $|\phi\>$.

In the same sense that the stabilizer problem can be viewed as an \HSP
with a quantum hiding function, the orbit coset problem is analogous
to the hidden shift problem.  The \emph{orbit coset} of
$|\phi_0\>,|\phi_1\> \in \Phi$ is the set $\{x \in G:
\alpha(x)(|\phi_1\>)=|\phi_0\>\}$; it is either empty or a left coset
of $\stab(|\phi_1\>)$ (or equivalently, a right coset of $|\phi_0\>$).
In the \emph{orbit coset problem} (\OCP), we are given (some number of
copies of) $|\phi_0\>,|\phi_1\> \in \Phi$.  The goal is to decide
whether their orbit coset is nonempty, and if so, to find both a
generating set for $\stab(|\phi_1\>)$ and an element $x \in G$ such
that $\alpha(x)(|\phi_0\>)=|\phi_1\>$.

It can be shown that for any group $G$ and any solvable normal
subgroup $N \pnormalin G$, the \OCP in $G$ reduces to the \OCP in
$G/N$ and subgroups of $N$ \cite{FIMSS03}.  While the details are beyond the
scope of this article, the reduction is based on a method for creating
a uniform superposition over the orbit of a state $|\phi\>$ under the
action $\alpha$, building on a technique introduced by Watrous in his
algorithms for solvable groups (\sec{decomposing}).  By combining this
with the efficient quantum algorithm for the hidden shift problem in
$(\Zp)^n$ discussed in \sec{zpnhiddenshift} (which can be
straightforwardly adapted to an efficient algorithm for orbit coset in
$(\Zp)^n$), \textcite{FIMSS03} obtain an efficient quantum algorithm for the
hidden shift problem in smoothly solvable groups, and for the \HSP in
solvable groups with a smoothly solvable commutator subgroup.

Recently, Ivanyos, Sanselme, and Santha have given algorithms for the
\HSP in extraspecial groups \cite{ISS07} and groups of nilpotency
class at most $2$ \cite{ISS08}.  These algorithms use the concept
of a quantum hiding function introduced above to reduce the problem to
an Abelian \HSP.  It would be interesting to develop further
applications of quantum hiding functions to the \HSP, hidden shift,
and related problems.

\subsection{Shifted Legendre symbol and Gauss sums}
\label{sec:legendre}

While no efficient quantum algorithm is known for the cyclic hidden
shift problem (i.e., the dihedral \HSP) for a general function
$f_0:\ZZ{N}\rightarrow S$, the problem can be more tractable given a
hiding function of a particular form.  As a simple example, the hidden
shift problem with the identity function $f_0(x)=x$ is trivial; but
this case is uninteresting as the problem can be solved equally well
with a classical or quantum computer.  However, more interesting
examples can be constructed if we drop the requirement that $f_0$ be
injective.\footnote{Dropping this restriction, the quantum query complexity of
  the hidden shift problem may no longer be polynomial; for example,
  the hidden shift problem with $f_0(x)=\delta_{x,0}$ is equivalent to
  unstructured search, which has quantum query complexity $\Omega(\sqrt N)$
  \cite{BBBV97}.}
For example, the Legendre symbol $\chi$ provides an
example of a function with an efficient quantum algorithm, but no
known efficient classical algorithm.

\subsubsection{Shifted Legendre symbol problem}

For a finite field $\Fp$ with $p$ an odd prime, the value $\chi(x)$ of the
\emph{Legendre symbol} $\chi:\Fp\rightarrow \{-1,0,+1\}$ depends on
whether $x$ is zero, a nonzero square (i.e., a quadratic residue), or
a nonsquare (i.e., a quadratic nonresidue) in $\Fp$.  It is defined by
\begin{equation}
\chi(x) = 
\begin{cases}
0 & x=0 \\
+1 & \exists\, y\neq 0:~ x=y^2 \\
-1 & \text{otherwise}.
\end{cases}
\end{equation}
For example, in $\FF{5}$ we have the values 
\begin{equation*}
\begin{array}{c|ccccc}
x & 0 & 1 & 2 & 3 & 4 \\\hline
\chi(x) & 0 & +1 & -1 & -1 & +1 
\end{array}
\end{equation*}
The Legendre symbol is a \emph{multiplicative character},
as it is easy to verify that $\chi(xy) = \chi(x)\chi(y)$ for all $x,y
\in \Fp$.  This fact can be used to show that $\sum_{x \in \Fp}
\chi(x)=0$. The identity 
\begin{equation} 
 \chi(x) = x^{(p-1)/2}\bmod{p} 
\end{equation} 
shows that repeated squaring modulo $p$ can be used to compute the
value $\chi(x)$ in time $\poly(\log p)$.

In the \emph{shifted Legendre symbol problem} over $\Fp$, we define the functions
$f_0(x):=\chi(x)$ and $f_1(x) := \chi(x+s)$ for all $s\in\Fp$; the
task is to determine the hidden shift $s$ given a black-box
implementation of the function $f_1$.  We emphasize that although the
functions $f_0,f_1$ are not injective, this can nevertheless be viewed
as (a relaxed version of) a hidden shift problem.  The ability to
efficiently solve this particular hidden shift problem quantum
mechanically stems from properties of multiplicative
functions under the (additive) Fourier transform.

No efficient classical algorithm for the shifted Legendre symbol
problem is known.  Although one can show that $O(\log p)$ random
queries to the function $\chi(x+s)$ are sufficient to obtain enough
information to determine $s$ \cite{Dam02}, it is not clear how to do
so efficiently. In fact, the Legendre sequence
$\chi(x),\chi(x+1),\dots$ has been proposed as a
pseudorandom function with potential cryptographic applications \cite{Dam90}.

The following quantum algorithm efficiently solves the shifted Legendre
symbol problem \cite{DHI06}:

\begin{algorithm}[Shifted Legendre symbol]\label{alg:SLS} \ \\
\emph{Input:} Black-box function $\chi(x+s)$ for some unknown $s \in \Fp$. \\
\emph{Problem:} Determine the hidden shift $s$.
\begin{enumerate}
\item Prepare the uniform superposition $|\Fp\>$ and query the function in an ancilla register, giving the state
\begin{equation}
  \frac{1}{\sqrt p} \sum_{x \in \Fp} |x,\chi(x+s)\>.
\end{equation}
\item Measure whether the second register is in the state $|0\>$.  If it is, the first register is left in the state $|-s\>$, and measuring it determines $s$.  Otherwise, we are left with the state
\begin{equation}
  \frac{1}{\sqrt{p-1}} \sum_{x \in \Fp\setminus\{-s\}} |x,\chi(x+s)\>,
\end{equation}
and we continue.
\item Apply the unitary operation $|x,b\> \mapsto (-1)^b |x,b\>$ and uncompute the shifted Legendre symbol, giving the state
\begin{equation}
\frac{1}{\sqrt{p-1}}\sum_{x\in\Fp}\chi(x+s)\ket{x}. 
\end{equation}
\item Apply the Fourier transform over $\Fp$, yielding
\begin{equation}
\frac{1}{\sqrt{p-1}}\sum_{y\in\Fp} \hat\chi(y) \rou{p}^{-sy} \ket{y}. 
\label{eq:ftlegendre}
\end{equation}
where $\hat{\chi}:\Fp\rightarrow \C$ is the normalized Fourier transform of $\chi$ (a normalized \emph{Gauss sum}, cf.\ \eq{gausssum}), namely
\begin{align}
\hat{\chi}(y) & := \frac{1}{\sqrt{p}}\sum_{x\in\Fp}\chi(x)\rou{p}^{xy}.
\end{align}
(Note that $\hat\chi(0)=0$ and $|\hat\chi(y)|=1$ for $y \in \Fpx$.)
\item The equality
\begin{align}
\hat{\chi}(y) 
&= \frac{1}{\sqrt{p}}\sum_{x\in\Fp}\chi(xy^{-1})\rou{p}^{x} \\
&= \chi(y)\hat{\chi}(1)
\end{align}
shows that the state \eq{ftlegendre} is in fact a uniformly weighted
superposition of the elements of $\Fp$, where the state $|y\>$ has a
phase proportional to $\chi(y)\rou{p}^{-sy}$.  Thus we correct the
relative phases by the operation $\ket{y}\mapsto\chi(y)\ket{y}$ for
all $y \in \Fpx$, giving the state
\begin{equation}
\frac{\hat{\chi}(1)}{\sqrt{p-1}}\sum_{y\in\Fpx}{\rou{p}^{-sy}\ket{y}}. 
\end{equation}
\item Perform the Fourier transform over $\Fp$ and measure in the computational basis, giving $s$ with probability $1-O(1/p)$. 
\end{enumerate}
\end{algorithm}

It is easy to see that the above algorithm solves the shifted Legendre symbol problem not only over a prime field $\Fp$, but over any finite field $\Fq$.  To verify this, we need only compute the Fourier transform of the quadratic character $\chi:\FF{q}\rightarrow \{-1,0,+1\}$, namely
\begin{align}
\hat{\chi}(y) & := \frac{1}{\sqrt{q}}\sum_{x\in\FF{q}}{\chi(x)\rou{p}^{\Tr(xy)}} \\
& = \frac{1}{\sqrt{q}}\sum_{x\in\FF{q}}{\chi(xy^{-1})\rou{p}^{\Tr(x)}} \\
& = 
\chi(y)\hat{\chi}(1)
\end{align}
(recall the definition of the Fourier transform over $\Fq$ in \sec{fieldqft}).  Indeed, the solution can be generalized to any shifted multiplicative character of $\Fq$ \cite{DHI06}, and to any  function over $\Fp$ that hides a multiplicative subgroup of polylogarithmic index \cite{MRRS07}. 

For the ring $\ZN$ with $N=p_1^{r_1}\times\cdots \times p_k^{r_k}$ odd, the
generalization of the Legendre symbol is called the \emph{Jacobi symbol} 
$(\cdot/N):\ZN \rightarrow \{-1,0,+1\}$. It is defined as the product 
\begin{equation}
\left(\frac{x}{N}\right) =
\left(\frac{x}{p_1}\right)^{r_1} \cdots
\left(\frac{x}{p_k}\right)^{r_k}
\end{equation}
(where $(x/p) := \chi(x)$ is an alternative notation for the Legendre
symbol that makes the field size explicit).  This is again a
multiplicative character, although its values need not indicate
squares modulo $N$ (for example, $(2/15) = (2/3)(2/5)= (-1)^2 = 1$,
while $2$ is not a square modulo $15$).  Analogous to the shifted
Legendre symbol problem, one can define a shifted Jacobi symbol
problem over $\ZZ{N}$, which also has an efficient quantum algorithm 
\cite{DHI06}.

\subsubsection{Estimating Gauss sums}

In the above solution to the shifted Legendre symbol problem, we encountered the Fourier transform of the multiplicative character $\chi$, which is a Gauss sum.  This naturally leads to a quantum algorithm for approximating Gauss sums \cite{DS02}.

For a finite field $\FF{q}$, a nontrivial multiplicative character
$\chi:\FF{q}\rightarrow \C$, and a nontrivial
additive character $\psi:\FF{q}\rightarrow \C$, the
\emph{Gauss sum} is defined as the inner product between these two characters:
\begin{equation}
G(\chi,\psi) := \sum_{x\in\FF{q}}{\chi(x)\psi(x)}.
\label{eq:gausssum}
\end{equation}
It is not hard to show that any Gauss sum has norm $|G(\chi,\psi)|=\sqrt{q}$, so to learn the value of a Gauss sum, it suffices to determine the phase $\phi\in[0,2\pi)$ of $G(\chi,\psi) = \sqrt{q}\cdot \e^{\ii \phi}$.

There are $q-1$ distinct multiplicative characters $\chi_a:\FF{q}\rightarrow \C$ indexed by $a\in\ZZ{(q-1)}$.  For a fixed multiplicative generator $g$ of $\Fqx$, we have $\chi_a(g^j) := \rou{q-1}^{aj}$ for all $j \in \Z$, and $\chi_a(0):=0$. The $q-2$ nontrivial characters are those with $a\neq 0$. As the discrete logarithm $\log_g(g^j) = j \bmod{q-1}$ can be calculated efficiently with a quantum computer (\sec{dlog}), we can efficiently induce the phase $\chi_a(g^j)$ by subtracting the value $aj$ modulo $q-1$ from the state $\ket{\hat{1}}$, giving
\begin{align}
  \ket{g^j}\otimes\ket{\hat{1}}
  &= \ket{g^j}\otimes \frac{1}{\sqrt{q-1}}\sum_{y\in\ZZ{(q-1)}}\rou{q-1}^{y}\ket{y} \\
  &\mapsto \ket{g^j}\otimes
           \frac{1}{\sqrt{q-1}}\sum_{y\in\ZZ{(q-1)}}\rou{q-1}^{y}\ket{y-aj} \\
  &= \ket{g^j}\otimes
     \frac{\rou{q-1}^{aj}}{\sqrt{q-1}}\sum_{y\in\ZZ{(q-1)}}\rou{q-1}^{y}\ket{y} \\
  &= \chi_a(g^j)\ket{g^j}\otimes\ket{\hat{1}}
\label{eq:phasekickback}
\end{align}
(this is sometimes referred to as the \emph{phase kickback trick}).

The $q$ additive characters $\psi_b:\Fq\to\C$ indexed by $b\in\Fq$ are defined as $\psi_b(x):=\rou{p}^{\Tr(bx)}$ for all $x \in \Fq$.  The character $\psi_0$ is
trivial, and all $b\neq 0$ give nontrivial characters. 

With these definitions in place, the Gauss sum estimation algorithm is as follows.
\begin{algorithm}[Gauss sum estimation] \ \\
\emph{Input:} A finite field $\Fq$, 
a nontrivial multiplicative character $\chi_a$ (where $a\in\ZZx{(q-1)}$), and 
a nontrivial additive character $\psi_{b}$ (where $b\in\Fqx$). \\ 
\emph{Problem:} Approximate within precision $\delta>0$ the angle 
$\phi\in[0,2\pi)$ such that $G(\chi_a,\psi_b) = \sqrt{q}\cdot \e^{\ii\phi}$.

\medskip\noindent
Perform phase estimation (\sec{phaseest}) with precision $\delta$ on the following single-qubit unitary operation (which requires applying the operation $O(1/\delta)$ times), inputting its eigenstate $|1\>$ of eigenvalue $\e^{\ii\phi}$:
\begin{enumerate}
\item For an arbitrary input state $\alpha|0\>+\beta|1\>$, prepare the state $|\Fqx\>$ in an ancilla register.
\item Using the phase kickback trick described in \eq{phasekickback}, transform the state to
\begin{equation}
  \frac{1}{\sqrt{q-1}} \bigg(\!
\alpha\ket{0}\otimes\sum_{x\in\Fq}{\chi^*_a(x)\ket{x}}
+ 
\beta\ket{1}\otimes\sum_{x\in\Fq}{\chi_a(x)\ket{x}}\!\bigg).
\end{equation}
\item Conditional on the qubit being in the state $|1\>$, multiply the ancilla register by $b$ and apply the Fourier transform over $\Fq$, yielding the state
\begin{align}
\Big(\alpha\ket{0}+\hat{\chi}_a(b)\beta\ket{1}\Big)\otimes\frac{1}{\sqrt{q-1}}\sum_{x\in\Fq}{\chi^*_a(x)\ket{x}}
\end{align}
where
\be
  \hat\chi_a(b) := \frac{G(\chi_a,\psi_b)}{\sqrt q} = \e^{\ii\phi}.
\ee
\item Apply the phase rotation $\ket{x}\mapsto \chi_a(x)\ket{x}$
to the ancilla register, returning it to its original state, and giving  
\begin{equation}
\Big(\alpha\ket{0}+\hat{\chi}_a(b)\beta\ket{1}\Big)\otimes|\Fqx\>.
\end{equation}
Discarding the ancilla register, notice that the above steps effectively implement the conditional phase shift $\ket{0} \mapsto \ket{0}$, $\ket{1} \mapsto \e^{\ii \phi}\ket{1}$.
\end{enumerate}
\end{algorithm}

The above quantum algorithm has running time polynomial in $\log q$ and $1/\delta$, whereas classical sampling over the $q$ values $\chi_a(x)\psi_b(x)$ requires $\poly(\sqrt{q}/\delta)$ samples to achieve the same quality of approximation.

Both additive and multiplicative characters can be defined over the ring $\ZN$, and there are corresponding Gauss sums
\begin{equation}
G(\chi_a,\psi_b) = \sum_{x\in\ZZ{N}}{\chi_a(x)\psi_b(x)}
\end{equation}
with $\chi_a(xy) = \chi_a(x)\chi_a(y)$ and $\psi_b(x) = \rou{N}{bx}$ for all $x,y \in \ZN$ (see the comprehensive book by \textcite{BEW98}).  Such Gauss sums over finite rings can be approximated by a quantum computer as well, using the above algorithm in a relatively straightforward way.

As Gauss sums occur frequently in the calculation of the number of points on hypersurfaces over finite fields (see for example \textcite{IR90}), these same quantum algorithms can be used to approximately count such points with an accuracy that does not seem achievable classically \cite{Dam04}.

\subsection{Generalized hidden shift problem}
\label{sec:ghshp}

P{\'o}lya has advised that ``if there is a problem you can't solve,
then there is an easier problem you can solve: find it'' \cite{Pol45}.
In that spirit, we conclude our discussion of the hidden shift
problem by describing a generalization that offers more ways to obtain
information about the hidden shift. At least in the case of cyclic
groups, this problem indeed turns out to be easier than the original
hidden shift problem.

In the \emph{$M$-generalized hidden shift problem} for the group $G$,
we are given a hiding function $f:\{0,\ldots,M-1\} \times G \to S$
satisfying two conditions: for any fixed $j \in \{0,\ldots,M-1\}$,
$f(j,x)$ is an injective function of $x \in G$; and for each $j \in
\{0,\ldots,M-2\}$, $f(j+1,x)=f(j,sx)$. For $M=2$, this problem is
equivalent to the usual hidden shift problem, since the hiding
functions $f_0,f_1$ can be obtained as $f_j(x)=f(j,x)$. However, the
$M$-generalized hidden shift problem appears to become easier for
larger $M$; it trivially reduces to the $M'$-generalized hidden shift
problem with $M'<M$, but larger values of $M$ provide new ways to
query the hiding function. Note that if $s^M=1$, then the
$M$-generalized hidden shift problem is equivalent to the \HSP in
$\ZZ{M} \times G$ with the cyclic hidden subgroup $\<(1,s)\>$.  In
general, the $M$-generalized hidden shift problem in $G$ reduces to
the \HSP in $G \wr \ZZ{M}$ \cite{FZ06}, but notice that this
reduction is only efficient for $M=\poly(\log|G|)$.

The Abelian generalized hidden shift problem could potentially be
applied to solve lattice problems.  Recall from \sec{hspandapps} that
the $\poly(n)$-unique shortest lattice vector problem efficiently reduces to (the
standard approach to) the dihedral \HSP.  In fact the same holds
for the $M$-generalized hidden shift problem in $\ZN$, provided
$M=\poly(\log N)$.

While no efficient algorithm is known for the case where $M=\poly(\log
N)$, efficient algorithms do exist for larger values of $M$.  First,
notice that the $N$-generalized hidden shift problem in $\ZN$ is an
\HSP in $\ZN \times \ZN$, which can be solved by Abelian Fourier
sampling.  Essentially the same strategy works provided $M=\Omega(N)$,
but fails for sublinear values of $M$.  However, there is another
quantum algorithm that is efficient provided $M \ge N^\epsilon$ for
some fixed $\epsilon>0$ \cite{CD05}, based on the pretty good
measurement techniques discussed in \sec{pgm}.  For the
$M$-generalized hidden shift problem in $\ZN$, implementing the PGM reduces to an integer programming problem in $d=\log
N / \log M$ dimensions, which can be solved efficiently for $d=O(1)$
\cite{Len83}.

It would also be interesting to consider the generalized hidden shift
problem in non-Abelian groups.  For example, a solution of this
problem for the symmetric group could be used to solve the
\emph{$M$-generalized graph isomorphism problem}, in which we are
given $M$ rigid $n$-vertex graphs
$\Gamma_0,\Gamma_1,\ldots,\Gamma_{M-1}$ that are either all
non-isomorphic, or sequentially isomorphic with a fixed isomorphism
$\pi \in S_n$, namely $\Gamma_{j+1}=\pi(\Gamma_j)$ for
$j=0,1,\ldots,M-2$.  For large $M$, this problem might seem considerably
easier than graph isomorphism, yet no efficient algorithms for the
corresponding generalized hidden shift problem are known.  Indeed,
very little is known about the non-Abelian generalized hidden shift
problem in general.

\section{Hidden Nonlinear Structures}
\label{sec:nonlinear}

The non-Abelian hidden subgroup problem (\sec{nonabelHSP}) was
originally introduced with the hope of generalizing the success of
Shor's algorithm.  As we have seen, these efforts have so far met with
only limited success: while polynomial-time quantum algorithms are
known for the \HSP in some non-Abelian groups, the cases with
significant applications---namely, the dihedral and symmetric
groups---remain largely unresolved.  Thus there have been several
attempts to generalize the Abelian \HSP in other ways.  The hidden
shift problem (\sec{hiddenshift}) represents one such attempt.  In
this section we discuss a more radical departure from the \HSP, a
class of problems aimed at finding \emph{hidden nonlinear structures}.

Let us return our attention the Abelian \HSP---and more specifically,
to the hidden subgroup problem in the additive group of the
$d$-dimensional vector space $\Fqd$ (where $\Fq$ denotes the finite
field with $q$ elements).  Then we can view the \HSP as a problem of
identifying a \emph{hidden linear structure}: the subgroups of the
additive group $\Fqd$ are precisely its linear subspaces, and their
cosets are parallel affine subspaces, or \emph{flats} (cf.\ step
\ref{item:dloglines} of \alg{dlog}).  Thus in this \HSP, we are given
a function that is constant on sets of points specified by linear
equations, and the goal is to recover certain parameters of those
equations.  It is natural to consider replacing the linear function
by a polynomial of higher degree.  Here we describe three such
hidden nonlinear structure problems: the hidden polynomial problem,
shifted subset problems, and polynomial Legendre symbol problems.

\subsection{The hidden polynomial problem}

Perhaps the most straightforward nonlinear generalization of the
Abelian \HSP is the \emph{hidden polynomial problem}
\cite{CSV07}.  In this problem, the hidden object is a
polynomial $h(x) \in \Fq[x_1,\ldots,x_d]$.  Generalizing
\eq{abelhides}, we say that a black box function $f:\Fqd \to S$ (for
some finite set $S$) \emph{hides the polynomial $h(x)$} if
\be
  f(x)=f(x') \text{~if and only if~} h(x)=h(x')
\ee
for all $x,x' \in \Fqd$.  In other words, the function $f$ is constant on the \emph{level sets}
\be
  L^h_y := h^{-1}(y) = \{x \in \Fqd: h(x)=y\}
\ee
and distinct on different level sets.  The hidden polynomial problem is to determine $h(x)$ up to differences that do not affect its level sets (i.e., up to an overall additive or multiplicative constant).

Notice that the polynomial $h(x)$ trivially hides itself.  But just as there is no a priori relationship between function values and cosets in the general \HSP, we prefer to assume that the association of function values to level sets is arbitrary.  Indeed, if we were promised that $f(x)=h(x)$, even a classical computer could solve the hidden polynomial problem efficiently.  But with no promise on how the level sets are mapped to function values, it is not hard to show that the classical randomized query complexity of the hidden polynomial problem is exponential in $d \log q$ \cite{CSV07}, by a similar argument as for the Abelian \HSP \cite{Sim97b}.

With a quantum computer, we can approach the hidden polynomial problem by closely following the standard method for the \HSP (\sec{standardmeth}).  Querying the function $f$ on the uniform superposition $|\Fqd\>$ and discarding the resulting function value, one is left with the state $|L^h_y\>$ with probability $|L^h_y|/q^d$.  Equivalently, the result is the \emph{hidden polynomial state}
\be
  \rho_h := \sum_{y \in \Fqd} \frac{|L^h_y|}{q^d} |L^h_y\>\<L^h_y|.
\label{eq:hpstate}
\ee
Notice that these states are quite similar to the hidden subgroup states \eq{hsstate}, modulo the fact that level sets of a polynomial can have different sizes, unlike the cosets of a subgroup.  Just as we upper bounded the query complexity of the \HSP by analyzing the statistical distinguishability of the states \eq{hsstate}, so we can upper bound the query complexity of the hidden polynomial problem by doing the same for the states \eq{hpstate}.  Following a similar argument as in \sec{multireg}, one can show that
\be
  F(\rho_h,\rho_{h'})^2 \le \frac{1}{q^d} \sum_{y,y' \in \Fq} \frac{|L^h_y \cap L^h_{y'}|^2}{|L^h_y|}
\ee
(cf.\ \eq{cosetintersect}).  Thus, the hidden polynomial states are pairwise distinguishable provided their level sets do not intersect too much.  Since almost all polynomials are \emph{absolutely irreducible} (i.e., they do not have any nontrivial factors, even over an extension of the base field), this suffices to show that if the dimension $d$ and the maximum degree of the polynomials are fixed, then the query complexity of the hidden polynomial problem is $\poly(\log q)$ for almost all polynomials \cite{CSV07}.

Moving beyond query complexity, we would like to know whether there is an efficient quantum algorithm---i.e., one with running time $\poly(\log q)$---for the hidden polynomial problem.  Just as for the \HSP, the most general version of this question is currently open.  However, suppose we are promised that the hidden polynomial has the form
\be
  h(x_1,\ldots,x_d) = g(x_1,\ldots,x_{d-1}) - x_d
\label{eq:specialhiddenpoly}
\ee
for some $(d-1)$-variate polynomial $g(x_1,\ldots,x_{d-1}) \in \Fq[x_1,\ldots,x_{d-1}]$.
(A simple example is the \emph{hidden parabola problem}, in which $h(x,y)=\alpha x^2 + \beta x - y$ for some unknown $\alpha,\beta \in \F_q$ that we would like to determine.)  For such a hidden polynomial, 
the level sets are simply translates of each other, namely $L^h_y = L^h_0 + (0,\ldots,0,y)$.  Provided the maximum degree of the polynomial is at most some fixed constant, there is a quantum algorithm that determines $h$ (up to an additive offset) in time $\poly(d \log q)$ \cite{DDW07}.  This algorithm is based on the pretty good measurement approach described in \sec{pgm}.  Recall that the implementation of the PGM relies on quantum sampling from the solutions of an average-case algebraic problem.  For the hidden polynomial problem with a polynomial of the form \eq{specialhiddenpoly}, this problem is a system of polynomial equations, much like the pair of quadratic equations \eqs{heisw}{heisv} that arise in the algorithm for the \HSP in the Heisenberg group.

\subsection{Shifted subset problems and exponential sums}

Other families of hidden nonlinear structure problems arise in the setting of \emph{shifted subset problems}.  Such problems are most naturally stated directly in terms of quantum state distinguishability.\footnote{Although the construction is somewhat technical, it is possible to formulate shifted subset problems in terms of a black box from which the state $\rho_{S,T}$ can be efficiently prepared on a quantum computer, but that typically must be queried exponentially many times to determine $S,T$ on a classical computer \cite{CSV07}.}  Suppose that for fixed subsets $S,T \subseteq \Fqd$, we are given the quantum state $|S+t\>$ (a uniform superposition over the elements of $S+t$), where $t$ is chosen uniformly at random from $T$.  In other words, we are given the mixed quantum state
\be
  \rho_{S,T} = \frac{1}{|T|} \sum_{t \in T} |S+t\>\<S+t|.
\ee
In the shifted subset problem, the goal is to determine some property of $S$ or $T$ (or both) using samples of $\rho_{S,T}$.

In \cite{CSV07}, two examples of shifted subset problems are considered in which the set $S$ is a $d$-dimensional sphere, i.e., the set of points
\be
  \mathcal{S}_r := L^{\sum_{i=1}^d x_i^2}_r =\left\{ x \in \Fqd : \sum_{i=1}^d x_i^2 = r \right\}
\ee
for some $r \in \Fq$.  In the \emph{hidden radius problem}, $T=\Fqd$, and the goal is to learn $r$.  In the \emph{hidden flat of centers problem}, we are promised that $r=1$, and $T$ is some unknown flat in $\Fqd$; the goal is to determine this flat.

In general, when $T=\Fqd$, symmetry ensures that $\rho_{S,T}$ is diagonal in the Fourier basis.  Then the goal is to learn $S$ from samples of a distribution given by its Fourier transform (recall \sec{fieldqft}), namely
\be
  \Pr(k) = \frac{1}{q^d |S|} \left| \sum_{x \in S} \rou{p}^{\Tr(kx)} \right|^2
\ee
where $p$ is the characteristic of $\Fq$ and $\Tr:\Fq\to\Fp$ denotes
the trace map.  In particular, when $S=\mathcal{S}_r$ is a
$d$-dimensional sphere, the distribution is proportional to an
exponential sum known as a \emph{Kloosterman sum} for $d$ even, or a
\emph{Sali\'e sum} (a kind of twisted Kloosterman sum) for $d$ odd.
In either case, these distributions are information-theoretically
distinguishable for different values of $r$.  Moreover, a closed form
expression for Sali\'e sums gives an efficient quantum algorithm for
determining whether $r$ is a quadratic residue, provided $d$ is odd.

On the other hand, suppose $S$ is fixed and $T$ is an unknown flat
(or, more generally, some low-degree surface).  If we could perform
the transformation $|S+t\> \mapsto |t\>$, then we could sample from
points on the flat, and thereby reconstruct it.  Unfortunately, this
transformation is generally not unitary, since $S$ could intersect
with its translates.  However, we can attempt to approximate such a
transformation using the continuous-time quantum walk on the Cayley
graph of $\Fqd$ generated by $S$.  When $S = \mathcal{S}_1$,
this Cayley graph is known as the \emph{Winnie Li graph}.  Its
eigenvalues are given by Kloosterman or Sali\'e sums, depending on whether $d$ is even or odd.  For $d$ odd, the explicit
expression for Sali\'e sums provides an efficient implementation of
the quantum walk, which in turn gives an efficient quantum algorithm
for the hidden flat of centers problem.

Of course, it is possible to make many other choices for $S$ and $T$,
so the above examples just begin to explore potential quantum
algorithms for shifted subset problems.  However, these simple
examples already reveal a connection between the calculation of
exponential sums\footnote{Computing exponential sums is also closely
  related to counting the solutions of finite field equations.
  Indeed, Kedlaya's algorithm (\sec{kedlaya}) can be used to 
  efficiently approximate Kloosterman sums when the field 
  characteristic is small (see \textcite{CSV07}).}
and the implementation of quantum walk that could perhaps be developed
further.  It would also be interesting to find concrete algorithmic
applications of shifted subset problems.

\subsection{Polynomial reconstruction by Legendre symbol evaluation}

The quantum algorithm for the shifted Legendre symbol problem (\sec{legendre}) recovers the constant term $s$ of a linear function $f(x) = x+s$ hidden in the black-box function $\chi(f(x))= \chi(x+s)$, where $\chi$ is the Legendre symbol.
As a precursor to the efficient quantum algorithm, it was shown that the quantum query complexity is $O(1)$, while the classical query complexity is $\Omega(\log p)$ \cite{Dam02}.  Here we discuss the generalization to a nonlinear function $f(x)$ hidden in the black-box function $\chi(f(x))$.  \textcite{RS04} showed that the quantum query complexity is significantly lower than the classical query complexity even in this more general case.  Whether there exists an efficient quantum algorithm to reconstruct the polynomial remains open.

Let $f \in \Fp[x]$ be an unknown polynomial.  Given a black box for $\chi(f(x)))$, with $\chi$ the Legendre symbol over $\Fp$, we want to reconstruct $f$ using as few queries as possible.  Note that for any $c\in \Fpx$, $\chi(c^2 f(x)) = \chi(f(x))$, making it impossible to tell the difference between $f(x)$ and $c^2 f(x)$ on the basis of the black box $\chi(f(x))$.  Moreover, if the factorization of $f(x)$ contains a square, i.e., if $f(x) = g^2(x)\cdot h(x)$, then $\chi(f(x)) = \chi(g^2(x))\chi(h(x))$, which is identical to $\chi(h(x))$ (except possibly at the zeros of $g$).  Thus we restrict our attention to polynomials that are monic and squarefree.

In the case where $f(x)=x+s$, the reason that $O(1)$ quantum queries suffice is that the states $\sum_x \chi(x+s)\ket{x}$ are nearly orthogonal for different values of $s \in \Fp$.  This follows from the identity
\begin{equation}
\sum_{x\in\Fp}{\chi(x+r)\chi(x+s)} = 
\begin{cases}
p-1 & s=r \\
-1 & s\neq r.
\end{cases}
\end{equation}
For polynomials $f,g$ of degree $d$ that are monic and squarefree, the generalization of this fact is provided by the Weil bound \cite{LN97}, which implies that
\begin{align} \label{eq:weil1}
  \sum_{x\in\Fp}\chi(f(x))^2 &\ge p-d \\
  \sum_{x\in\Fp}\chi(f(x))\chi(g(x)) &\le 2d\sqrt{p} \quad \text{if~} f \ne g.
  \label{eq:weil2}
\end{align}
Note that for $d \geq \sqrt{p}/2$, \eq{weil2} is trivial.  However, for $d\leq p^{1/2-\varepsilon}$ with $\varepsilon > 0$, we find the following.

Given a black box function $\chi(f(x))$ where $f \in \Fp[x]$ is an unknown monic, squarefree polynomial of degree $d$, two queries can be used to create the state
\begin{equation}
  \ket{\tilde\chi(f)} := \frac{1}{\sqrt{p}}\sum_{x\in\Fp}\tilde\chi(f(x))\ket{x},
\end{equation}
where $\tilde\chi$ is identical to $\chi$ except that $\tilde\chi(0)=1$.  (This adjustment to the Legendre symbol is required to deal with the otherwise zero amplitudes for the zeros of $f$.)  Using \eqs{weil1}{weil2}, it follows that
\begin{equation}
  |\braket{\tilde\chi(f)}{\tilde\chi(g)}| \leq \frac{2d}{\sqrt{p}}.
\end{equation}
Since there are $p^d$ monic polynomials of degree $d$ over $\Fp$, \thm{bk} (and specifically, \eq{bkcopies}) shows that there is a measurement on $O(d)$ copies of $\ket{\tilde\chi(f)}$ that determines the $d$ unknown coefficients of $f$ with probability $1-O(1/p)$.
The classical query complexity of this problem can be shown to be $\Omega(d\log p)$, which therefore gives a separation between classical and quantum query complexity.

\section{Approximating \texorpdfstring{\cclass{\#P}}{\#P}-Complete Problems}
\label{sec:jones}

Recently, there has been considerable interest in quantum algorithms for approximately solving various \cclass{\#P}-complete problems.  The first such algorithms were for approximating the Jones polynomial; more recently, similar ideas have been used to give approximate solutions to other \cclass{\#P}-complete problems.  These algorithms are not as closely related to Shor's as most of those discussed in this article, but they are decidedly algebraic, relying heavily on group representation theory.

The \emph{Jones polynomial} is a central object in low-dimensional topology with surprising connections to physics.  \textcite{Wit89} showed that the Jones polynomial is closely related to topological quantum field theory (TQFT).  \textcite{FKLW03} investigated the relationship between TQFT and topological quantum computing, showing that quantum computers can efficiently simulate TQFTs \cite{FKW02}, and that in fact TQFTs essentially capture the power of quantum computation \cite{FLW02}. In particular, \textcite{FKW02} showed that quantum computers can efficiently approximate the Jones polynomial at a fifth root of unity.  Subsequently, \textcite{AJL06} described an explicit quantum algorithm for approximating the Jones polynomial, generalizing to any primitive root of unity (see also the work by \textcite{WY06}). 

To define the Jones polynomial, we must first introduce the concepts of knots and links.  A \emph{knot} is an embedding of the circle in $\R^3$, i.e., a closed loop of string that may wrap around itself in any way.  More generally, a \emph{link} is a collection of any number of knots that may be intertwined.  In an \emph{oriented link}, each loop of string is directed.  It is natural to identify links that are \emph{isotopic}, i.e., that can be transformed into one another by continuous deformation of the strings.

The \emph{Jones polynomial} of an oriented link $L$ is a Laurent polynomial $V_L(t)$ in the variable $\sqrt{t}$, i.e., a polynomial in $\sqrt{t}$ and $1/\sqrt{t}$.  It is a \emph{link invariant}, meaning that $V_L(t)=V_{L'}(t)$ if the oriented links $L$ and $L'$ are isotopic.  While it is possible for the Jones polynomial to take the same value on two non-isotopic links, it can often distinguish links; for example, the Jones polynomials of the two orientations of the trefoil knot are different.

Given an oriented link $L$, one way to define its Jones polynomial is as follows \cite{Kau87}.  First, let us define the \emph{Kauffman bracket} $\<L\>$, which does not depend on the orientation of $L$.  Each crossing in the link diagram can be opened in one of two ways, and for any given crossing we have
\begin{equation}
  \Big\< \; \raisebox{.75pc}{\xy0;/r1pc/: \xoverv \endxy} \; \Big\> 
  = t^{1/4} \Big\< \; \raisebox{.75pc}{\xy0;/r1pc/: \xunoverv \endxy} \; \Big\> 
  + t^{-1/4} \Big\< \; \raisebox{.75pc}{\xy0;/r1pc/: \xunoverh \endxy} \; \Big\>,
\end{equation}
where the rest of the link remains unchanged.  Repeatedly applying this rule, we eventually arrive at a link consisting of disjoint unknots.  The Kauffman bracket of a single unknot is $\<\bigcirc\>:=1$, and more generally, the Kauffman bracket of $n$ unknots is $(-t^{1/2}-t^{-1/2})^{n-1}$.  By itself, the Kauffman bracket is not a link invariant, but it can be turned into one by taking into account the orientation of the link, giving the Jones polynomial.  For any oriented link $L$, we define its \emph{writhe} $w(L)$ as the number of crossings of the form $\raisebox{.75pc}{\xy0;/r1pc/: \UseComputerModernTips\xoverv<<|< \endxy}$ minus the number of crossings of the form $\raisebox{.75pc}{\xy0;/r1pc/: \UseComputerModernTips\xunderv<<|< \endxy}$.  Then the Jones polynomial is defined as 
\begin{equation}
 V_L(t) := (-t^{-1/4})^{3 w(L)} \<L\>.
\label{eq:jonespoly}
\end{equation}

It is useful to view links as arising from \emph{braids}.
A braid is a collection of $n$ parallel strands, with adjacent strands allowed to cross over or under one another.  Two braids on the same number of strands can be composed by placing them end to end.  The \emph{braid group} $B_n$ on $n$ strands is an infinite group with generators $\{\sigma_1,\ldots,\sigma_{n-1}\}$, where $\sigma_i$ denotes a twist in which strand $i$ passes over strand $i+1$, interchanging the two strands.  More formally, the braid group is defined by the relations $\sigma_i \sigma_{i+1} \sigma_i = \sigma_{i+1} \sigma_i \sigma_{i+1}$ and $\sigma_i \sigma_j = \sigma_j \sigma_i$ for $|i-j| > 1$.

Braids and links differ in that the ends of a braid are open, whereas a link consists of closed strands.  We can obtain a link from a braid by connecting the ends of the strands in some way.  One simple way to close a braid is via the \emph{trace closure}, in which the $i$th strand of one end is connected to the $i$th strand of the other end for each $i=1,\ldots,n$, without crossing the strands.  A theorem of \textcite{Ale23} states that any link can be obtained as the trace closure of some braid.

The Jones polynomial of the trace closure of a braid can be expressed in terms of the Markov trace (a weighted variant of the usual trace) of a representation of the braid group defined over the Temperley-Lieb algebra \cite{Jon85}.  When evaluating the Jones polynomial $V_L(t)$ at the root of unity $t=\e^{2\pi\ii/k}$, this representation is unitary.  This naturally suggests a quantum algorithm for approximating the Jones polynomial.  Suppose that we can implement unitary operations corresponding to twists of adjacent strands on a quantum computer.  By composing such operations, we can implement a unitary operation corresponding to the entire braid.  It remains to approximate the Markov trace of this operator.

The trace of a unitary operation $U$ can be approximated on a quantum computer using the \emph{Hadamard test}.  If a conditional $U$ operation is applied to the state $|+\> \otimes |\psi\>$ and the first qubit is measured in the $|\pm\>$ basis, where $|\pm\> := \frac{1}{\sqrt 2}(|0\> \pm |1\>)$, the expectation value of the outcome is precisely $\Re(\<\psi|U|\psi\>)$.  (This is simply the phase estimation procedure described in \sec{phaseest} with $n=1$, i.e., with a single bit of precision.)  Replacing the states $|\pm\>$ by the states $|\pm \ii\> := \frac{1}{\sqrt 2}(|0\> \pm \ii |1\>)$, we can approximate $\Im(\<\psi|U|\psi\>)$.  Using a maximally mixed state as input instead of the pure state $|\psi\>$ and sampling sufficiently many times from the resulting distribution, we can obtain an approximation of $\Re(\Tr U)$ or $\Im(\Tr U)$.  Similarly, we can approximate a weighted trace by sampling from an appropriate distribution over pure states.

Applying this approach to the relevant unitary representation of the braid group, one obtains a quantum algorithm for approximating the Jones polynomial of the trace closure of a braid at a root of unity.  In particular, for a braid on $n$ strands, with $m$ crossings, and with $t=\e^{2\pi\ii/k}$, there is an algorithm running in time $\poly(n,m,k)$ that outputs an approximation differing from the actual value $V_L(t)$ of the Jones polynomial by at most $(2\cos\frac{\pi}{k})^{n-1}/\poly(n,k,m)$, with only exponentially small probability of failure \cite{AJL06}.

Given a braid with an even number of strands, another natural way to create a link is called the \emph{plat closure}.  Here, we simply join adjacent pairs of strands at each end of the braid.  The plat closure can be viewed as the trace closure of a braid on $2n$ strands together with $2n$ additional straight strands.  Using this fact, we can express the Jones polynomial of the plat closure of a braid at $t=\e^{2\pi\ii/k}$ as the expectation value of a particular unitary representation of the braid group in a pure quantum state.  Thus the Jones polynomial of the plat closure can also be approximated using the Hadamard test, but now using a pure input state instead of a mixed one.  This gives an efficient quantum algorithm for an additive approximation of the Jones polynomial of the plat closure of a braid at a root of unity \cite{AJL06}.

Notice that these algorithms only provide \emph{additive} approximations, meaning that the error incurred by the algorithm is independent of the value being approximated, which is undesirable when that value is small.  (In fact, note that the additive error increases exponentially with $n$, the number of strands in the braid.)  It would be preferable to obtain a \emph{multiplicative} approximation, or better still, an exact calculation.  However, exactly computing the Jones polynomial is \sharpP-hard \cite{JVW90}, and hence unlikely to be possible even with a quantum computer.  Furthermore, obtaining the additive approximation achieved by \cite{AJL06} for the Jones polynomial of the plat closure of a braid is as hard as any quantum computation \cite{FLW02,BFLW05,WY06,AA06}.

To implement the quantum algorithm for approximating the trace closure of a braid, it is only necessary to have a single pure qubit (the qubit initialized to $|+\>$ in the Hadamard test), and many mixed ones.  Thus it can be carried out in the one clean qubit model introduced by \textcite{KL98} to investigate the power of mixed state quantum computation.  In fact, the problem of estimating the Jones polynomial of the trace closure of a braid at a fifth root of unity (to the precision described above) exactly characterizes the power of this model \cite{SJ07,JW08}, just as the approximation of the plat closure characterizes general quantum computation.

We conclude by briefly mentioning various extensions of these results.  \textcite{WY06} show how to evaluate the Jones polynomial of a generalized closure of a braid, and how to evaluate a generalization of the Jones polynomial called the HOMFLYPT polynomial.  Recent work of \textcite{AAEL07} shows how to approximate the Tutte polynomial of a planar graph, which in particular gives an approximation of the partition function of the Potts model on a planar graph; this problem also characterizes the power of quantum computation, albeit only for unphysical choices of parameters.  More generally, there are efficient quantum algorithms to compute additive approximations of tensor networks \cite{AL08}.

\begin{acknowledgments}
We thank Sean Hallgren for discussions of algorithms for number fields.
We also thank Dorit Aharonov, Greg Kuperberg, Fr{\'e}d{\'e}ric Magniez, Cris Moore, Miklos Santha, John Watrous, and Pawel Wocjan for comments on a preliminary version.
This article was written in part while AMC was at the Institute for Quantum Information at Caltech, where he received support from the National Science Foundation under grant PHY-456720 and from the Army Research Office under grant W9111NF-05-1-0294.  AMC was also supported in part by MITACS, NSERC, and the US ARO/DTO.
WvD was supported in part by the Disruptive Technology Office (DTO) under Army Research Office (ARO) contract number W911NF-04-R-0009 and by an NSF CAREER award.
\end{acknowledgments}

\appendix
\section{Number Theory} \label{app:nt}

\subsection{Arithmetic modulo \texorpdfstring{$N$}{N}}

When performing calculations with integers modulo $N$ we use the equivalence relation $x=y\bmod{N}$ if and only if $x-y\in N\Z=\{\ldots,-N,0,N,2N,\ldots\}$.  Often we omit the notation `$\bmod N$' and instead consider $x$ and $y$ as elements of the ring $\ZN$.  Other ways of denoting this ring are $\Z_N$ and $\Z/\!(N)$; in this article we use the notation $\ZN$, which is conventional in computational number theory.  Although formally the elements of $\ZN$ are the sets $\{\dots,-N+x,x,x+N,x+2N,\dots\}$, we often simply represent such an element by the integer $x$; this representation is unique if we require $x\in\{0,\dots,N-1\}$ .

Addition modulo $N$ corresponds to the additive group $(\ZN,+)$, which has $N$ elements.  For example, with $N=2$ we have $0+0=0$, $1+0=0+1=1$, and $1+1=0$.  If $N$ has the prime factorization $N=p_1^{r_1}\cdots p_k^{r_k}$, then the additive group $\ZN$ can be decomposed as $(\Zpr{1}) \times \dots \times (\Zpr{k})$.

Multiplication modulo $N$ is more complicated than addition as not all elements of $\ZN$ have a multiplicative inverse. For example, $5\cdot 5 = 1 \bmod{6}$, but there is no element $x$ such that $2x=1\bmod{6}$. In general, there exists a $y$ such that $xy=1\bmod{N}$ if and only if $\gcd(x,N)=1$, where $\gcd(x,N)$ is the greatest common divisor of $x$ and $N$.  The set of such invertible elements of $\ZN$ make up the multiplicative group $\ZNx$.  It is easy to check that in $\ZZ{6}$ there are only two invertible elements: $\ZZx{6} = \{1,5\}$. The size of the multiplicative group $\ZNx$  depends on the prime factorization of $N$; one can show that for $N=p_1^{r_1}\cdots p_k^{r_k}$,
\begin{equation}
    \varphi(N) := |\ZNx| = (p_1-1)p_1^{r_1-1}\cdots (p_k-1)p_k^{r_k-1},
\end{equation}
where $\varphi$ is called Euler's totient function.  Similarly to the additive case, one also has the multiplicative group isomorphism $\ZNx \cong \Zprx{1} \times \dots \times \Zprx{k}$.

By combining the isomorphisms for the additive and the multiplicative groups of integers modulo $N$, we obtain the Chinese remainder theorem.  This states that for $N=p_1^{r_1}\cdots p_k^{r_k}$, the bijection between the elements of $x\in\ZN$ and the $k$-tuples $(x_1,\dots,x_k)\in(\Zpr{1})\times\dots\times (\Zpr{k})$ (with $x_i = x \bmod{p_i^{k_i}}$ for all $i$) respects both addition and multiplication in the ring $\ZN$. This fact often allows us to break up algebraic problems in $\ZN$ into $k$ smaller problems in $\Zpr{i}$, which can be easier to deal with.

\subsection{Finite fields and their extensions}

For a prime number $p$ we have $\varphi(p)=p-1$, which means that all but the zero element of $\Zp$ have a multiplicative inverse modulo $p$. Thus $\Zp$ is a finite field, which we denote by $\Fp$. Just as $\R$ is a field that can be extended to $\C$ by including the solutions to polynomial equations such as $\alpha^2+1=0$, so can the finite field $\Fp$ be extended to $\Fpr$ for any positive integer $r$.  Any finite field has order $q=p^r$ with $p$ some prime, and for each prime power $p^r$ there is a finite field of that order.  Up to isomorphism, this finite field is in fact unique, so we can refer to \emph{the} finite field $\Fq$ without ambiguity.  The additive group of $\GF{p^r}$ is isomorphic to the additive group $(\ZZ{p})^r$, while the multiplicative group $\FFx{p^r}$ is cyclic, and is isomorphic to the additive group $\ZZ{(p^r-1)}$. Note that $\GF{p^r}$ is very different from $\ZZ{p^r}$ for $r>1$, as $|\FFx{p^r}|=p^r-1$ while $|\ZZx{p^r}| = (p-1)p^{r-1}$.

A standard way of explicitly constructing a finite field $\Fpr$ is by extending $\Fp$ with a formal variable $\alpha$ satisfying $T(\alpha)=0$, where $T$ is an irreducible polynomial of degree $r$ in $\Fp[\alpha]$.  The finite field $\Fpr$ is isomorphic to the ring of polynomials $\Fp[\alpha]$ modulo the polynomial $T(\alpha)$, i.e., $\Fpr \cong \Fp[\alpha]/T(\alpha)$.
 
\begin{example}[Construction of $\GF{8}$]
Modulo $2$, the polynomial $T(\alpha)=\alpha^3+\alpha+1$ is irreducible: $T(\alpha)$ cannot be written as the product of two nontrivial polynomials. Hence $\GF{2}[\alpha]/(\alpha^3+\alpha+1)$ is the finite field $\GF{8}$.  
The addition in this field is the straightforward addition of quadratic polynomials modulo $2$, such that, for example, $(\alpha^2 + \alpha) + (\alpha^2 + 1) = \alpha+1$. Multiplication of the elements is slightly more involved, but the explicit multiplication table (\tab{multF8}) confirms that 
$\GF{2}[\alpha]/(\alpha^3+\alpha+1)$ is indeed a field. Note for example that $\alpha$ has multiplicative inverse $\alpha^2+1$, as $\alpha(\alpha^2+1) = \alpha^3+\alpha=1$ by the equality $\alpha^3+\alpha+1=0$.

\begin{table*}
\begin{tabular}{c@{\hspace {2ex}}|@{\hspace {2ex}}c@{\hspace {2ex}}c@{\hspace {2ex}}c@{\hspace {2ex}}c@{\hspace {2ex}}c@{\hspace {2ex}}c@{\hspace {2ex}}c@{\hspace {2ex}}c}
$\times$ 
& $0$ & $1$ & $\alpha$ & $\alpha+1$ & $\alpha^2$ & $\alpha^2+1$ & $\alpha^2+\alpha$ & $\alpha^2+\alpha+1$\\\hline
$0$       & $0$ & $0$       & $0$       & $0$ & $0$ & $0$ & $0$ & $0$ \\
$1$       & $0$ & $1$       & $\alpha$       & $\alpha+1$ & $\alpha^2$ & $\alpha^2+1$ & $\alpha^2+\alpha$ & $\alpha^2+\alpha+1$\\
$\alpha$       & $0$ & $\alpha$       & $\alpha^2$     & $\alpha^2+\alpha$ & $\alpha+1$ & $1$ & $\alpha^2+\alpha+1$ & $\alpha^2+1$\\
$\alpha+1$     & $0$ & $\alpha+1$     & $\alpha^2+\alpha$   & $\alpha^2+1$ & $\alpha^2+\alpha+1$ & $\alpha^2$ & $1$&  $\alpha$ \\
$\alpha^2$     & $0$ & $\alpha^2$     & $\alpha+1$     & $\alpha^2+\alpha+1$ & $\alpha^2+\alpha$ & $\alpha$ & $\alpha^2+1$ &  $1$ \\
$\alpha^2+1$   & $0$ & $\alpha^2+1$   & $1$       & $\alpha^2$ & $\alpha$ & $\alpha^2+\alpha+1$ & $\alpha+1$ & $\alpha^2+\alpha$ \\
$\alpha^2+\alpha$   & $0$ & $\alpha^2 + \alpha$ & $\alpha^2+\alpha+1$ & $1$ & $\alpha^2+1$ & $\alpha+1$ & $\alpha$ & $\alpha^2$ \\ 
$\alpha^2+\alpha+1$ & $0$ & $\alpha^2+\alpha+1$ & $\alpha^2+1$   & $\alpha$ & $1$ & $\alpha^2+\alpha$ & $\alpha^2$ & $\alpha+1$
\end{tabular}
\caption{The multiplication table of the finite field $\GF{8}$ represented by the elements of $\GF{2}[\alpha]/(\alpha^3+\alpha+1)$.}
\label{tab:multF8}
\end{table*}
\end{example}

Obviously, $\GF{8}$ contains the subfield $\GF{2}$, but less obviously, $\GF{8}$ does not contain $\GF{4}$. In general, $\GF{q_1}$ contains the finite field $\GF{q_2}$ if and only if $q_1$ is a power of $q_2$, hence if and only if $q_1=p^{r_1}$ and $q_2=p^{r_2}$ where $r_2$ divides $r_1$. For a finite field $\GF{q}$ with $q=p^r$ and $p$ prime, we call $\Fp$ the \emph{base field} of $\GF{q}$, and we call $\GF{p^r}$ the degree $r$ \emph{extension} of the field $\Fp$. By taking the limit of arbitrarily high degree $r$, we obtain the \emph{algebraic closure} $\bar{\F}_p$ of $\Fp$, which is an infinite field.

Although the construction of an extension field using an irreducible polynomial makes it easy to explicitly perform calculations, the procedure soon becomes cumbersome, as \tab{multF8} already shows. Furthermore, the representation depends on the specific polynomial being used, so it introduces a certain arbitrariness.  Hence, whenever possible, we talk about finite fields without specifying a particular representation.

\subsection{Structure of finite fields}

Starting from the infinite field $\bar{\F}_p$, the elements of $\GF{p^r}$ can be characterized as the $q=p^r$ solutions to the equation $x^{q}=x$. This immediately implies the above statement that $\GF{p^{r_1}}$ contains $\GF{p^{r_2}}$ if and only $r_2$ divides $r_1$.

Within the finite field $\Fpr$, the \emph{Frobenius automorphism} $\phi:\Fpr \to \Fpr$ is the map defined by $\phi(x) = x^{p}$.  It is a field automorphism, meaning that $\phi(x+y)=\phi(x)+\phi(y)$ and $\phi(xy)=\phi(x)\phi(y)$ for all $x,y\in\Fpr$.  Iterating the Frobenius automorphism gives the $r$ different maps $\phi^j:x\mapsto x^{p^j}$ for $j=0,1,\dots,r-1$, which are all automorphisms of $\Fpr$. Because $\phi(a)=a$ for all base field elements $a\in\Fp$, we see that if $x\in\Fpr$ is a root of a polynomial $F(X) = a_d X^d + \cdots + a_1 X + a_0 \in\Fp[X]$ with coefficients in the base field $\Fp$, then so are its conjugates $\phi^j(x)$ as, assuming $F(x)=0$, we have $F(\phi(x)) = \sum_i {a_i (\phi(x))^i} = \sum_i {\phi(a_ix^i)} = \phi(F(x))= 0$. This result generalizes to multivariate polynomials $F\in\GF{p}[X_1,\dots,X_n]$ with roots $x=(x_1,\dots,x_n)\in\FFsup{p^r}{n}$: if $F(x)=0$ then also $F(\phi^j(x)) = F(\phi^j(x_1),\dots,\phi^j(x_n))=0$.  Hence the set of solutions $\{x\in\FFsup{p^r}{n} : F(x)=0 \}$ is invariant under the Frobenius automorphism.

\section{Representation Theory of Finite Groups} \label{app:repr}

In this appendix, we briefly review the theory of group representations needed to study the non-Abelian HSP.  
Here it is sufficient to restrict our attention finite groups, and to representations over finite-dimensional complex vector spaces.
For a more detailed introduction to representation theory, see \textcite{Ser77,Ham89}.

\subsection{General theory}

A \emph{linear representation} (or simply \emph{representation}) of a finite group $G$ over the vector space $\C^n$ is a \emph{homomorphism} $\sigma:G\rightarrow \GL(\C^n)$, i.e., a map from group elements to nonsingular $n \times n$ complex matrices satisfying $\sigma(x)\sigma(y)=\sigma(xy)$ for all $x,y\in G$.
Clearly, $\sigma(1)=\I$ and $\sigma(x^{-1})=\sigma(x)^{-1}$.
We say that $\C^n$ is the \emph{representation space} of $\sigma$, where $n$ is called its \emph{dimension} (or \emph{degree}), denoted $d_\sigma$.

Two representations $\sigma$ and $\sigma'$ with representation spaces $\C^{n}$ are \emph{isomorphic} (denoted $\sigma \sim \sigma'$) if and only if there is an invertible linear transformation $M \in \C^{n \times n}$ such that $M \sigma(x) = \sigma'(x) M$ for all $x\in G$.  (Representations of different dimensions cannot be isomorphic.)
Every representation is isomorphic to a \emph{unitary representation}, i.e., one for which $\sigma(x)^{-1} = \sigma(x)^\dag$ for all $x \in G$.  Thus we can restrict our attention to unitary representations without loss of generality.

The simplest representations are those of dimension one, such that 
$\sigma(x)\in\C$ with $|\sigma(x)|=1$ for all $x \in G$. 
Every group has a one-dimensional representation called the \emph{trivial representation}, defined by $\sigma(x)=1$ for all $x \in G$.

Two particularly useful representations of a group $G$ are its \emph{left regular representation} and its \emph{right regular representation}.  Both of these representations have dimension $|G|$, and their representation space is the \emph{group algebra} $\C G$, i.e., the $|G|$-dimensional complex vector space spanned by basis vectors $|x\>$ for $x \in G$.
The left regular representation $L$ satisfies $L(x)|y\> = |xy\>$, and the right regular representation $R$ satisfies $R(x)|y\> = |yx^{-1}\>$.
In particular, both regular representations are \emph{permutation representations} as each consists entirely of permutation matrices.
  
Given two representations $\sigma:G\rightarrow V$ and $\sigma':G\rightarrow V'$, we can define their \emph{direct sum}, a representation $\sigma \oplus \sigma':G \to V \oplus V'$ of dimension $d_{\sigma \oplus \sigma'} = d_\sigma + d_{\sigma'}$.  The representation matrices of $\sigma \oplus \sigma'$ are of the form
\begin{equation}
(\sigma\oplus \sigma')(x) =  
\left(
\begin{array}{cc}
\sigma(x) & 0 \\
0 & \sigma'(x)
\end{array}
\right)
\end{equation}
for all $x\in G$.

A representation is \emph{irreducible} if it cannot be decomposed as the direct sum of two other representations.  Any representation of a finite group $G$ can be written as a direct sum of irreducible representations (or \emph{irreps}) of $G$.  Up to isomorphism, $G$ has a finite number of irreps.  The symbol $\hat G$ denotes a complete set of irreps of $G$, one for each isomorphism type.

Another way to combine two representations is with the \emph{tensor product}.  The tensor product of $\sigma:G\rightarrow V$ and $\sigma':G\rightarrow V'$ is $\sigma \otimes \sigma': G \to V \otimes V'$, a representation of dimension $d_{\sigma \otimes \sigma'} = d_\sigma d_{\sigma'}$. 

The \emph{character} of a representation $\sigma$ is the function $\chi_\sigma:G \to \C$ defined by $\chi_\sigma(x) := \Tr \sigma(x)$.  We have $\chi_\sigma(1)=d_\sigma$, $\chi(x^{-1}) = \chi(x)^*$, and $\chi(yx)=\chi(xy)$ for all $x,y\in G$.  For two representations $\sigma,\sigma'$, we have $\chi_{\sigma \oplus \sigma'} = \chi_\sigma + \chi_{\sigma'}$ and $\chi_{\sigma\otimes \sigma'} = \chi_{\sigma}\cdot \chi_{\sigma'}$.

Perhaps the most useful result in representation theory is \emph{Schur's Lemma}, which can be stated as follows:
\begin{theorem}[Schur's Lemma]\label{thm:schur}
Let $\sigma$ and $\sigma'$ be two irreducible representations of $G$, and let $M \in \C^{d_\sigma \times d_{\sigma'}}$ be a matrix satisfying $\sigma(x) M = M \sigma'(x)$ for all $x \in G$.  Then if $\sigma \not\sim \sigma'$ we have $M=0$; and if $\sigma = \sigma'$, then $M$ is a scalar multiple of the identity matrix.
\end{theorem}

Schur's Lemma can be used to prove the following orthogonality relation for irreducible representations:
\begin{theorem}\label{thm:orthorep}
For two irreps $\sigma,\sigma' \in \hat G$, we have 
\be
  \frac{d_\sigma}{|G|} \sum_{x \in G} \sigma(x)^*_{i,j} \, \sigma'(x)_{i',j'}
  = \delta_{\sigma,\sigma'} \delta_{i,i'} \delta_{j,j'},
\ee
where $\delta_{\sigma,\sigma'}$ is $1$ if $\sigma = \sigma'$, and $0$ otherwise.
\end{theorem}
In particular, this implies a corresponding orthogonality relation for the \emph{irreducible characters} (i.e., the characters of the irreducible representations):
\begin{theorem}\label{thm:orthochar}
For two irreps $\sigma,\sigma' \in \hat G$, we have 
\be
  (\chi_\sigma,\chi_{\sigma'}) := \frac{1}{|G|}
  \sum_{x\in G}{\chi_{\sigma}(x)^* \, \chi_{\sigma'}(x)} = \delta_{\sigma,\sigma'}.
\ee
\end{theorem}
Characters provide a simple test for irreducibility.  In particular,
for any representation $\sigma$, $(\chi_\sigma,\chi_\sigma)$ is a positive integer, and is equal to $1$ if and only if $\sigma$ is irreducible.

Any representation of $G$ can be broken up into its irreducible components.
The regular representations of $G$ are useful for understanding such decompositions, since they contain every possible irrep of $G$, each occurring a number of times equal to its dimension.
In particular,
\be
  L \cong \bigoplus_{\sigma \in \hat G} \big(\sigma\otimes\I_{d_\sigma}\big)
  ,\quad
  R \cong \bigoplus_{\sigma \in \hat G} \big(\I_{d_\sigma}\otimes\sigma^*\big),
\ee
where $\I_d$ denotes the $d \times d$ identity matrix.
In fact, this holds with the same isomorphism for both $L$ and $R$, since they  are commutants of each other.  The isomorphism is simply the Fourier transform over $G$.  For its precise definition, as well as a proof that it decomposes the  regular representations, see \sec{nonabelQFT}.

Considering $\chi_L(1)=\chi_R(1)=|G|$ and using this decomposition, we find the well-known identity
\be
  \sum_{\sigma \in \hat G} d_\sigma^2 = |G|.
\label{eq:dim2sum}
\ee
Also, noting that $\chi_L(x)=\chi_R(x)=0$ for any $x \in G \setminus \{1\}$, we see that
\be
  \sum_{\sigma \in \hat G}{d_\sigma \, \chi_\sigma(x)} = 0.
\label{eq:regrepsum}
\ee

In general, the multiplicity of the irrep $\sigma \in \hat G$ in an arbitrary representation $\tau$ of $G$ is given by $\mu^\tau_\sigma := (\chi_\sigma,\chi_\tau)$.  Then we have the decomposition
\be
  \tau \cong \bigoplus_{\sigma \in \hat G} \sigma \otimes \I_{\mu^\tau_\sigma}.
\ee
The projection onto the \emph{$\sigma$-isotypic subspace} of $\tau$ is given by
\be
  \Pi^\tau_\sigma
  := \frac{d_\sigma}{|G|} \sum_{x \in G} \chi_\sigma(x)^* \, \tau(x).
\ee

Any representation $\sigma$ of $G$ can also be viewed as a representation of any subgroup $H \le G$, simply by restricting its domain to elements of $H$.  We denote the resulting \emph{restricted representation} by $\Res^G_H \sigma$.  Even when $\sigma$ is irreducible over $G$, it will in general \emph{not} be irreducible over $H$.  (It is also possible to extend any representation $\sigma'$ of $H$ to an \emph{induced representation} $\Ind^G_H \sigma'$ of $G$, but we will not need the definition here.)

We conclude with some examples of groups and their irreducible represenations.

\subsection{Abelian groups}

The irreducible representations of any finite Abelian group are all one-dimensional.
(Conversely, any non-Abelian group has some irrep of dimension greater than $1$.)

For a cyclic group $G=\ZZ{n}$, all irreps are of the form $\sigma_k:\ZZ{n} \to \C$ with $\sigma_k(x) := \e^{2\pi\ii kx/n}$, where $k\in\ZZ{n}$ uniquely labels the representation.  Hence there are indeed $n$ inequivalent irreps of $\ZZ{n}$, all of dimension $1$. 

Any finite Abelian group can be written as a direct product of cyclic factors, and its irreducible representations are given by products of irreps of those factors.  For example, the irreducible representations of the group $G=(\ZZ{n})^2$ are given by $\sigma_{k}(x) := \e^{2 \pi \ii (k_1 x_1 + k_2 x_2)/n}$, where $k=(k_1,k_2) \in \ZZ{n}^2$ uniquely labels the irrep.

\subsection{Dihedral group}
\label{sec:dihedralgroup}

The dihedral group of order $2n$ is $D_n = \ZZ{n} \semidirect \ZZ{2}$, with the group law
\begin{equation}
  (x,a)\cdot (y,b) = (x+(-1)^a y,a+b)
\end{equation}
for $x,y\in\ZZ{n}$ and $a,b\in\ZZ{2}$.

For $n$ even, we have the following $1$-dimensional representations:
\begin{align}
\sigma_\text{tt}((x,a)) &:= 1 \\
\sigma_\text{ts}((x,a)) &:= (-1)^a \\
\sigma_\text{st}((x,a)) &:= (-1)^x \\
\sigma_\text{ss}((x,a)) &:= (-1)^{x+a};
\end{align}
for $n$ odd, we have only $\sigma_\text{tt}$ and $\sigma_\text{ts}$.
The $2$-dimensional representations are of the form
\begin{align}
\sigma_h((x,0)) &:=
\begin{pmatrix}
 \e^{2\pi \ii hx/n} & 0 \\
0 &  \e^{-2\pi \ii hx/n} 
\end{pmatrix}
\end{align} 
and
\begin{align}
\sigma_h((x,1)) &:=
\begin{pmatrix}
0 &  \e^{-2\pi \ii hx/n} \\
\e^{-2\pi \ii hx/n} & 0 
\end{pmatrix}
\end{align}
for some $h\in\{1,2,\dots,\ceil{\frac{n}{2}}-1\}$.
It is straightforward to check that these representations are all irreducible and that the sum of the dimensions squared gives $2n$.
\section{Curves Over Finite Fields} \label{app:curves}

Kedlaya's quantum algorithm for counting the number of points on a curve over a finite field relies on several results in algebraic geometry.  Here we explain some of the central concepts that are necessary to understand the algorithm.  For concreteness, we limit ourselves to the case of planar algebraic curves.  Our notation follows \cite{Lor96}, a highly recommended textbook for more information on this topic.

Given a bivariate polynomial $f\in\Fq[X,Y]$, we can consider the solutions to the equation $f(x,y)=0$ with $x,y$ elements of the base field $\Fq$ or of an extension field $\Fqr$.  The set of these solutions is the planar curve denoted by $C_f(\Fq)$ or $C_f(\Fqr)$, respectively.  Often we drop the subscript $f$ when it is clear from context.

\subsection{Affine and projective spaces}

The theory of algebraic equations works more generally if we allow points at infinity to be possible solutions as well.  We frequently work over the projective plane $\P^2$, which for a given finite field $\Fqr$ can be expressed as
\begin{equation}
\P^2(\Fqr) = (\FFsup{q^r}{3}\setminus \{(0,0,0)\})/\sim
\end{equation}
where two points are equivalent, $(x,y,z) \sim (x',y',z')$, if and only if there exists a $\lambda\in\FFx{q^r}$ such that $(\lambda x,\lambda y,\lambda z) = (x',y',z')$. These rays in $\FFsup{q}{3}$ are denoted by $(x:y:z)$, i.e.,
\begin{equation}
  (x:y:z) = \{(\lambda x,\lambda y,\lambda z) : 
\lambda \in \FFx{q^r}\} \subset \FFsup{q^r}{3}  
\end{equation}
for all $(x,y,z)\neq (0,0,0)$. 

One can easily verify that the projective plane $\P^2(\Fq)$ consists of $q^2+q+1$ points, of which $q^2$ lie in the affine plane $\A^2(\Fq) = \{(x,y,1) : (x,y)\in\Fq\}$; the remaining $q+1$ points are the \emph{line at infinity} $\{(x,1,0) : x\in\Fq\}$ and the \emph{point at infinity} $\{(1,0,0)\}$. This decomposition can be summarized by the equation $\P^2 = \A^2\cup \P^1 = \A^2\cup \A^1 \cup \A^0$. (For clarity, an affine space is often indicated by $\A^n(\Fq)$ rather than by the equivalent set $\FFsup{q}{n}$, as the latter suggests a vector space with an origin, a concept that plays no role in affine spaces. In this article we ignore this subtlety.)

\subsection{Projective curves}

The affine solutions to the polynomial equation $f(X,Y)=0$ over $\Fq$ consist of the set $\{(x,y)\in\FFsup{q}{2} : f(x,y)=0\}$, but for the solutions in the projective plane $\P^2(\Fq)$ we must make the following adjustment.  To define $f(X,Y)$ in $\P^2$, we introduce a third variable $Z$ that allows us to translate $f$ into a homogeneous polynomial, such that if $f(x,y,z)=0$ for $(x,y,z)\in\FFsup{q}{3}\setminus\{(0,0,0)\}$, then $f(\lambda x,\lambda y,\lambda z) = 0$ for all $\lambda\in \Fqx$.  For example, with $f(X,Y)=Y^2+X^3+X+1$, we have $f(X,Y,Z)=Y^2Z+X^3+XZ^2+Z^3$.

In other words, an algebraic curve $C_f$ in the projective plane is defined by a homogeneous polynomial $f\in\F_q[X,Y,Z]$, and its set of $\Fqr$-rational solutions is given by
\begin{equation}
C_f(\Fqr) = \{(x:y:z) : f(x,y,z)=0 \} \subset \P^2(\Fqr). 
\end{equation}
Notice that for each extension degree $r$ there is a different 
set of solutions $C_f(\Fqr)$.  Explicit examples of curves are given in Sections~\ref{sec:elliptic} and \ref{sec:kedlaya}.

\subsection{Properties of curves}

Let $f\in\Fq[X,Y,Z]$ define a planar, projective curve $C_f$.  A point $(x:y:z)\in C_f(\Fqr)$ is called \emph{nonsingular} if and only if 
\begin{equation}
\bigg(\frac{\partial f}{\partial X},
\frac{\partial f}{\partial Y},
\frac{\partial f}{\partial Z}\bigg)
(x,y,z) \neq (0,0,0), 
\end{equation}
where $\partial f/\partial X$ denotes the \emph{formal derivative} of $f$ with respect to $X$. A projective curve is called \emph{smooth} if all its points are nonsingular. 
In many ways, curves over finite fields are analogous to compact Riemann surfaces.  Most importantly, one can assign a genus $g$ to a smooth projective curve $C_f$, just as one can for a compact Riemann surface.  (This is why we use the projective curve: the affine curve is not compact.)  The projective line $\P^1$, defined by a linear equation such as $X=0$, has genus $0$; elliptic curves, defined by cubic equations, have genus $1$; and in general, a degree $d$ polynomial gives a curve with genus $g=\frac{1}{2}(d-1)(d-2)$.  The complexity of algorithms for curves often depends critically on the genus of the curve, and hence on the degree of the defining polynomial $f$.

\subsection{Rational functions on curves}

Similar to the case of Riemann surfaces, the geometric properties of a smooth, projective curve are closely related to the behavior of rational functions on the same surface.  For a smooth, projective curve $C_f$ defined by the homogeneous polynomial $f\in\F_q[X,Y,Z]$, we define the \emph{function field} of rational functions by
\begin{equation}
 \Fq(C_f) = \biggl\{\frac{g(X,Y,Z)}{h(X,Y,Z)} :  
\deg(g)=\deg(h) \biggr\}/\sim
\end{equation}
with $g$ and $h$ homogenous polynomials in $\Fq[X,Y,Z]$ of identical degree, and with equivalence between functions defined by
\begin{equation}
  \frac{g}{h} \sim \frac{g'}{h'} 
\quad\text{if and only if}\quad hg' - gh' \in (f) 
\end{equation}
where $(f)$ is the ideal generated by $f$.  Notice that by the requirement that $g$ and $h$ are of the same degree, we have
\begin{equation}
 \frac{g(\lambda x,\lambda y,\lambda z)}{h(\lambda x,\lambda y,\lambda z)}
= 
\frac{\lambda^{\deg(g)}}{\lambda^{\deg(h)}} \frac{g(x,y,z)}{h(x,y,z)}
= \frac{g(x,y,z)}{h(x,y,z)}, 
\end{equation}
which shows that $g/h$ is indeed well-defined on the points $(x:y:z)$ in the projective space $\P^2(\Fq)$.

It is an important fact that each non-constant rational function on $C_f$ has both roots (points where $g=0$) and poles ($h=0$), and that the number of roots equals the number of poles, counting multiplicity.  See \sec{kedlaya} for an example of the structure of $\Fq(C_f)$ for an elliptic curve over $\FF{2}$.

\bibliography{Bib_qa_rmp}
\bibliographystyle{apsrmp_improved}

\end{document}